\documentclass[runningheads]{llncs}
\usepackage{graphicx}
\usepackage{amssymb,amsmath}

\usepackage{amsthm}
\usepackage{xcolor}
\usepackage[inline]{enumitem}
\usepackage{pifont}
\usepackage{listings}
\usepackage{graphicx}
\usepackage{wrapfig}
\usepackage{tikz}
\usetikzlibrary{shapes.misc,arrows.meta}

\usepackage{float}
\usepackage{mathtools}
\usepackage{etoolbox}
\usepackage{comment}
\usepackage{multirow}
\usepackage{stmaryrd}

\makeatletter
\newtheorem*{rep@theorem}{\rep@title}
\newcommand{\newreptheorem}[2]{%
\newenvironment{rep#1}[1]{%
 \def\rep@title{#2 \ref{##1}}%
 \begin{rep@theorem}}%
 {\end{rep@theorem}}}
\makeatother

\newreptheorem{theorem}{Theorem}
\newreptheorem{lemma}{Lemma}

\newcommand{\mpcomment}[1]{}

\newcommand{\figlabel}[1]{\label{fig:#1}}
\newcommand{\figref}[1]{Fig.~\ref{fig:#1}}
\newcommand{\seclabel}[1]{\label{sec:#1}}
\newcommand{\secref}[1]{Section~\ref{sec:#1}}
\newcommand{\exlabel}[1]{\label{ex:#1}}
\newcommand{\exref}[1]{Example~\ref{ex:#1}}
\newcommand{\deflabel}[1]{\label{def:#1}}
\newcommand{\defref}[1]{Definition~\ref{def:#1}}
\newcommand{\thmlabel}[1]{\label{thm:#1}}
\newcommand{\thmref}[1]{Theorem~\ref{thm:#1}}
\newcommand{\proplabel}[1]{\label{prop:#1}}
\newcommand{\propref}[1]{Proposition~\ref{prop:#1}}
\newcommand{\lemlabel}[1]{\label{lem:#1}}
\newcommand{\lemref}[1]{Lemma~\ref{lem:#1}}

\newcommand{\corlabel}[1]{\label{cor:#1}}

\newcommand{\applabel}[1]{\label{app:#1}}
\newcommand{\appref}[1]{Appendix~\ref{app:#1}}
\newcommand{\Cc}{\mathcal{C}}

\newcommand{\Ff}{\mathcal{F}}
\newcommand{\Rr}{\mathcal{R}}

\newcommand{\Aa}{\mathcal{A}}
\newcommand{\Mm}{\mathcal{M}}
\newcommand{\Uu}{\mathcal{U}}

\newcommand{\Vv}{\mathcal{V}}

\newcommand{\Ll}{\mathcal{L}}
\newcommand{\Tt}{\mathcal{T}}

\renewcommand{\vec}[1]{{\bf #1}}
\newcommand{\set}[1]{\{#1\}}
\newcommand{\setpred}[2]{\{#1 \mid #2\}}
\newcommand{\sem}[1]{\llbracket #1 \rrbracket}
\newcommand{\angular}[1]{\langle #1 \rangle}

\renewcommand{\emptyset}{\varnothing}

\newcommand{\xdownarrow}[1]{%
  {\left\downarrow\vbox to #1{}\right.\kern-\nulldelimiterspace}
}
\newcommand{\longproj}[2]{{#1}\xdownarrow{0.27cm}_{#2}}

\newcommand{\stmt}{\angular{\textrm{stmt}}}
\newcommand{\cond}{\angular{\textrm{cond}}}

\newcommand{\code}[1]{ \texttt{#1}}
\newcommand{\codekey}[1]{{\textbf{#1}}}
\newcommand{\cd}[1]{\code{#1}}
\newcommand{\cdm}[1]{ \mathtt{#1}}
\newcommand{\pskip}{\codekey{skip}}
\newcommand{\passume}{\codekey{assume}}
\newcommand{\pif}{\codekey{if}}
\newcommand{\pthen}{\codekey{then}}
\newcommand{\pelse}{\codekey{else}}
\newcommand{\pwhile}{\codekey{while}}
\newcommand{\passign}{:=}

\newcommand{\exec}{\textsf{ Exec}}

\newcommand{\comp}{\textsf{ TEval}}
\newcommand{\Terms}{\textsf{ Terms}}
\newcommand{\cohexec}{\textsf{\small CoherentExecs}}

\newcommand{\init}[1]{\widehat{#1}}

\newcommand{\dblqt}[1]{\text{``}#1\text{''}}

\newcommand{\fun}{\textsf{fun}}
\newcommand{\rel}{\textsf{rel}}

\newcommand{\congcl}[1]{\cong_{#1}}
\newcommand{\ncongcl}[1]{{\ncong}_{#1}}
\newcommand{\eqcl}[2]{[#1]_{#2}}

\newcommand{\undf}{\mathsf{undef}}

\newcommand{\reject}{q_\mathsf{reject}}

\newcommand{\trans}{\textsf{trans}}
\newcommand{\sym}{\textsf{symm}}
\newcommand{\refl}{\textsf{refl}}
\newcommand{\irrefl}{\textsf{irref}}

\newcommand{\comm}{\textsf{comm}}
\newcommand{\idem}{\textsf{idem}}
\newcommand{\spo}{\textsf{SPO}}
\newcommand{\sto}{\textsf{STO}}
\newcommand{\assoc}{\textsf{assoc}}

\newcommand{\pspc}{\mathsf{PSPACE}}

\newcommand{\drawdirectedline}{\raisebox{0pt}{\scalebox{0.7}{\tikz{\draw[-{Latex[length=2mm, width=2mm]}, thick](0,0) -- (7mm,0);}}}}
\newcommand{\drawdirecteddash}{\raisebox{0pt}{\scalebox{0.7}{\tikz{\draw[-{Latex[length=2mm, width=2mm]}, thick, dashed](0,0) -- (7mm,0);}}}}

\newcommand{\epr}{\textsf{EPR}}
\newcommand{\scc}{\textsf{SCC}}

\newcommand{\minimal}{\textsf{min}}

\newcommand{\valplain}{\textsf{eval}}
\newcommand{\val}[3]{\valplain_{#1}(#2, #3)}

\newcommand{\oneext}{\mathsf{1Ext}}
\newcommand{\closureext}[2]{\mathsf{ClosureExt}^{#1}_{#2}}
\newcommand{\transclosureext}[2]{\mathsf{Trans}\closureext{#1}{#2}}
\newcommand{\execrestrict}{\mathsf{ExecRestr}}

\begin{document}
\title{What's Decidable About Program Verification Modulo Axioms?} 

\author{Umang Mathur\and%\orcidID{0000-0002-7610-0660} \and
P. Madhusudan \and
Mahesh Viswanathan}
\authorrunning{U. Mathur et al.}
\institute{University of Illinois, Urbana Champaign, USA \\
% \email{\{umathur3,madhu,vmahesh\}@illinois.edu}
}
\maketitle
%

%!TEX root = main.tex

\begin{abstract}
% We consider the problem of verification of an uninterpreted program ---
% checking whether a given program meets a specification under \emph{all}
% interpretations of the functions and predicates in the 
% program.
% While this problem is undecidable in general,
% the recently introduced subclass of \emph{coherent} 
% uninterpreted programs, on the other hand, admits decidable verification~\cite{coherence2019}.
%
%Verification of uninterpreted programs--- checking 
%whether a given program meets a specification under all 
%interpretations of the functions and predicates in the 
%program--- is, in general, undecidable. 
%The recently introduced subclass of \emph{coherent} uninterpreted programs, 
%on the other hand, admits decidable %verification~\cite{coherence2019}. 
We consider the decidability of the verification problem of programs \emph{modulo axioms} --- automatically verifying whether programs satisfy their
assertions, when the function and relation symbols are interpreted as arbitrary functions and relations that satisfy a set of first-order axioms.
Though verification of uninterpreted programs (with no axioms) is already undecidable, a recent work introduced a subclass of \emph{coherent} uninterpreted programs, 
and showed that they admit decidable verification~\cite{coherence2019}. 
%in
%when the functions and predicates in the program are 
%constrained using theories specified by first order axioms.
%We show that this problem is undecidable even when the axioms
%are universally quantified first order sentences.
We undertake a systematic study of various natural axioms for relations and functions, 
and study the decidability of the coherent verification problem. 
%
%specific theories for which verification is decidable.
Axioms include relations being reflexive, 
symmetric, transitive, or total order relations, %and their combinations, 
functions restricted to being associative, idempotent or commutative, and combinations of such axioms as well. 
Our comprehensive results unearth a rich landscape that shows that though several axiom classes admit decidability for coherent programs, coherence is not a panacea as several others continue to be undecidable.
\end{abstract}

%!TEX root = main.tex

\section{Introduction}
\seclabel{intro}
Programs are proved correct against safety specifications typically by induction--- the induction hypothesis is specified using \emph{inductive invariants} of the program, and one proves that the reachable states of the program stays within the region defined by the invariants, inductively. Though there has been tremendous progress in the field of \emph{decidable logics} for proving that invariants are inductive, finding inductive invariants is almost never fully automatic. 
And completely automated verification of programs is almost always undecidable.

Programs can be viewed as working over a data-domain, with variables storing values over this domain and being updated using constants, functions and relations defined over that domain. Apart from the notable exception of finite data domains, program verification is typically undecidable when the data domain is infinite. In a recent paper, Mathur et. al.~\cite{coherence2019} establish new decidability results when the data domain is infinite. Two crucial restrictions are imposed --- data domain functions and relations are assumed to be \emph{uninterpreted} and programs are assumed to be \emph{coherent} (the meaning of coherence is discussed later in this introduction). The theory of uninterpreted functions is an important theory in SMT solvers that is often used (in conjunction with other theories) to solve feasibility of loop-free program snippets, in bounded model-checking, and to validate verification conditions. The salient aspect of~\cite{coherence2019} is to show that entire program verification is decidable for the class of coherent programs, without any user-provided inductive invariants (like loop invariants). While the results of~\cite{coherence2019} were mainly theoretical, there has been recent work on applying this theory to verifying memory-safety of heap-manipulating programs~\cite{memsafeMMKMV19}.

Data domain functions and relations used in a program usually satisfy special properties and are not, of course, entirely uninterpreted. The results of~\cite{coherence2019} can be seen as an approximate/abstraction-based verification method in practice --- if the program verifies assuming functions and relations to be uninterpreted, then the program is correct for \emph{any} data domain. However, properties of the data domain are often critical in establishing correctness. For example, in order to prove that a sorting program 
results in sorted arrays, it is important that the binary relation $<$ used to compare
elements of the array is a total ordering on the underlying data sort. Consequently, constraining the data domain to satisfy certain axioms results in more accurate modeling for verification.

%The first question to ask is, perhaps, whether program verification is decidable when the underlying data domain is a particular model (or a complete theory) like arithmetic. However, notice that even when the data domain has only the successor function on natural numbers, programs can compute addition and multiplication (using loops) and hence we can reduce the problem of checking whether Diophantine equations are unsolvable to program verification. Moreover, a reasonable restriction of programs (such as an adaptation of the notion of coherence~\cite{coherence2019}) that avoids such a reduction seems hard.

\emph{In this paper, we undertake a systematic study of the verification of uninterpreted programs when the data-domains are constrained using theories specified by (universally quantified) axioms.} The choice of the axioms we study are guided by two principles. First, we study natural mathematical properties of functions and relations. 
% Second, we study those axioms under which the \emph{quantifier-free} fragment of first order logic is decidable. 
Second, we choose to study axioms that have a decidable \emph{quantifier-free} fragment of first order logic. 
The reason is that even single program executions can easily encode quantifier-free formulae (by computing the terms in variables, and assert Boolean combinations of atomic relations and equality on them). Since we are seeeking decidable  verification for programs \emph{with loops/iteration}, it makes little sense to examine axioms where even verification of single exectutions is
undecidable.
%quantifier-free fragment is undecidable. 

\begin{comment}
The first axioms on relations we study include axioms of reflexivity, irreflexivity, symmetry, and transitivity, and axioms that define partial and total orders, and combinations of these for different sets of relational symbols. Turning to functions, we study commutativity, associativity, and idempotence, all of which are 
axioms of the kind $\forall ~ \overline{x}.~ t_1 = t_2$, where $t_1, t_2$ are terms built using symbols from the vocabulory and the variables $\overline{x}$. 
In general, we study only universally quantified axioms --- existential quantification can be handled by Skolemizing and treating the Skolem functions as uninterpreted. Furthermore, apart from the relations and functions constrained using axioms, we also allow other completely uninterpreted functions and relations that the program can manipulate.
\end{comment}

\subsubsection*{Coherence modulo theories:}
Mathur et. al.~\cite{coherence2019} define a subclass of programs, called \emph{coherent programs}, 
for which program verification on uninterpreted domains is decidable; 
without the restriction of \emph{coherence}, program verification on uninterpreted domains is undecidable. Since our framework is strictly more powerful, we adapt the notion of coherence to incorporate theories. 
A coherent program~\cite{coherence2019} is one where all executions satisfy two properties --- memoizing and early-assumes. The memoizing property demands that the program computes any term, modulo congruence induced by the equality assumes in the execution, only once. More precisely, if an execution recomputes a term, the term should be stored in a current variable. The early-assumes restriction demands, intuitively, that whenever the program assumes two terms to be equal, it should do so early, before computing superterms of them.

We adapt the above notion to \emph{coherence modulo theories}\footnote{We adapt the definition in a way that preserves the spirit of the definition of coherence. Moreover, if we do not adapt the definition, essentially all axioms classes we study in this paper would be undecidable.}. The memoizing and early-assumes property are now required modulo the equalities that are entailed by the axioms. More precisely, if the theory is characterized by a set of axioms $\Aa$, the memoizing property demands that if a program computes a term $t$ and there was another term $t'$ that it had computed earlier which is equivalent to $t$ modulo the assumptions made thus far \emph{and the axioms $\Aa$},  then $t'$ must be currently stored in a variable. Similarly, the early-assumes condition is also with respect to the axioms --- if the program execution observes a new assumption of equality or a relation holding between terms, then we require that any equality entailed newly by it, the previous assumptions \emph{and the axioms $\Aa$}  do not involve a dropped term. This is a smooth extension of the notion of coherence from~\cite{coherence2019}; when $\Aa = \emptyset$, we essentially retrieve the notion from~\cite{coherence2019}.
\subsection*{Main Contributions}

Our first contribution is an extension of the notion of coherence in~\cite{coherence2019} to handle the presence of axioms, as 
described above; this is technically nontrivial and we provide a 
natural extension.

Under the new notion of coherence, we first study axioms on relations. The EPR (effectively propositional reasoning)~\cite{Ramsey1987} fragment of first order logic is one of the few fragments of first order logic that is decidable, and has been exploited for bounded model-checking and verification condition validation in the literature~\cite{ivy2016,paxosepr2017,Padon2016}. We study axioms written in EPR (i.e., universally quantified formulas involving only relations) and show that %{\bf (a)} 
verification for even coherent programs, modulo EPR axioms, is undecidable.

Given the negative result on EPR, we look at particular natural axioms for relations, which are nevertheless expressible in EPR. In particular, we look at reflexivity, irreflexivity, and symmetry axioms, and show that verification of coherent programs is decidable when the interpretation of some relational symbols is constrained to satisfy these axioms. Our proof 
proceeds by instrumenting the program with auxiliary {\tt assume} statements that preserve coherence and subtle arguments that show that verification can be reduced to the case without  axioms; decidability then follows from results established in~\cite{coherence2019}. 

We then show a much more nontrivial result that
verification of coherent programs remains decidable when some relational symbols are constrained to be transitive. The proof relies on
new automata constructions that compute streaming congruence closures while interpreting the relations to be transitive.

\begin{comment}
We then turn to other particular axioms on relations, and study {\bf (b)} reflexivity, {\bf (c)} irreflexivity, and {\bf (d)} symmetry axioms. We show that verification of coherent programs is decidable even with these axioms. (Note that these axioms can be expressed in the EPR fragment.) We develop a general program instrumentation technique that encodes certain axioms within executions themselves, while maintaining coherence, and subtle model constructions to prove that they reduce the verification problem under reflexivity/irreflexivity/symmetry axioms 
to the uninterpreted setting proved to be decidable in~\cite{coherence2019}.
\end{comment}

%We next consider %{\bf (e)} 
%the transitivity axiom, and show that even if some relations are required to be transitive, verification is decidable for coherent programs. 
%This proof of this nontrivial result relies on
%new automata based algorithms that compute streaming congruence closures while interpreting the relations to be transitive.
Furthermore, we show that %{\bf (f)} 
combinations of reflexivity, irreflexivity, symmetry, and transitivity, admit a decidable verification problem for coherent program. Using this observation, we conclude decidability of verification when certain relations are required to be  strict partial orders (irreflexive and transitive) or equivalence relations.  

We then consider %{\bf (g)} 
axioms that capture total orders and show that they too admit a decidable coherent verification problem. 
% (this too requires a new automaton construction). 
Total orders are also expressible in EPR and their formulation in EPR has been used in program verification, as they can be used in lieu of the ordering on integers when only ordering is important. For example, they can be used to model data in sorting algorithms, array indices in modeling distributed systems to model process ids and the states of processes, etc.~\cite{ivy2016,paxosepr2017}.

Our next set of results consider axioms on functions. Associativity and commutativity are natural and fundamental properties of  functions (like $+$ and $*$) and are hence natural ways to capture/abstract using these axioms. (See~\cite{Gulwani2007} where such abstractions are used in program analysis.) We first show that verification of coherent programs is decidable when some functions are assumed to be commutative or idempotent. Our proof, similar to the case of reflexive and symmetric relations, relies on reducing verification to the case without axioms using program instrumentation that capture the commutativity and idempotence axioms. However, when a function is required to be associative, the verification problem for coherent programs becomes undecidable. This undecidability result was surprising to us.

\begin{comment}
Our first result is that {\bf (h)} verification of coherent programs is decidable when some functions are assumed to be commutative. This proof works, in a sense, similar to the proofs of reflexivity and symmetric relations in that it augments executions to impose commutativity constraints and we show a reduction to coherent program verification without axioms. We can also show that {\bf (i)} coherent verification problem is decidable for idempotent functions. 

We consider {\bf (j)} associative functions next. In a result that was surprising to us, we show that, unfortunately, in this case the coherent verification problem is \emph{undecidable}. 
\end{comment}

The decidability results established for properties of individual relation or function symbols discussed above can be combined to yield decidable verification modulo a set of axioms. That is, the verification of coherent programs with respect to models where relational symbols satisfy some subset of reflexivity/irreflexivity/symmetery/transitivity axioms or none, and function symbols are either uninterpreted, commutative, or idempotent, is decidable. 

Decidability results outlined above, apply to programs that are coherent modulo the axioms/theories. However, given a program, in order to verify it using our techniques, we would also like to decide whether the program \emph{is} coherent modulo axioms. We prove %{\bf (k)} 
that for all the decidable axioms above, checking whether programs are coherent modulo the axioms is a decidable problem. Consequently, under these axioms, we can both check whether programs are coherent modulo the axioms and if they are, verify them.

%The notion of $k$-coherence (where $k \in \mathbb{N}$) was introduced by
%Mathur et. al.~\cite{coherence2019} to accomodate verification of larger classes of programs than coherent programs. An execution is $k$-coherent if one can use $k$ extra \emph{ghost} variables and assign values to them to help store terms to satisfy the memoizing property of coherent executions.
%; a program is $k$-coherent if all its executions are. 
%It was proved in~\cite{coherence2019} that verification of $k$-coherent programs for (purely) uninterpreted functions/relations is decidable and checking if program is $k$-coherent is decidable. We can extend both these results for the axiom classes we study --- %{\bf (l)} 
%checking $k$-coherence modulo these theories is decidable and verification of programs that are $k$-coherent programs modulo these theories is also decidable.

There are several other results that we mention only in passing. For instance, we show that even for single executions, verifying them modulo equational axioms is undecidable as it is closely related to the word problem for groups. And our positive results for program verification under axioms for functions (commutativity, idempotence), also shows that bounded model-checking under such axioms is decidable, which can have its own applications.

%In summary, we systematically explore natural classes of axioms involving functions and relations, and under a natural extension of coherence that incorporates theories, study the verification problem, uncovering a rich landscape of decidable and undecidable problems.
%, as summarized in Table~\ref{tab:results}. 

%Technically, we overcome several challenges to prove our results. First, the notion of coherence in~\cite{coherence2019} needed to be extended to handle the presence of axioms; and this had to be done carefully
%to also make sure checking coherence  (modulo axioms, in the decidable cases) continues to be decidable. Second, we develop a general program instrumentation technique that encodes certain axioms within executions themselves, and helps in proving decidability results for them. However, handling transitivity axioms and the total order axioms are significantly more complex, and the proofs need new automata constructions. The techniques we develop for the various axiom classes however do combine well, and we find, pleasantly, that all combinations of them continue to be decidable.

%On the negative side, we find that the restriction to coherent programs, though usually effective in making verification 
%decidable, is not a panacea --- there are several simple settings (EPR and associativity axioms), where verification of coherent programs is undecidable. 

Due to the large number of results and technically involved proofs, we give only the main theorems and proof gists for some of these in the paper; details can be found 
% in~\cite{techreport}.
in the Appendix.

%!TEX root = main.tex

\section{Illustrative Example}
\seclabel{illustrative-ex}

%!TEX root = main.tex

\begin{figure}[t]
% \scalebox{1.0}{
\scalebox{0.8}{
\begin{minipage}{1.3\textwidth}
\begin{minipage}[t]{0.5\textwidth}
% \rule[1mm]{1cm}{0.5pt} $P_\text{sorted-list}$ \rule[1mm]{1cm}{0.5pt} \\\\
% \passume~(\cd{\!\!T~$\neq$~F});\\
% \cd{sorted \passign~T;}\\
% \pwhile (\cd{x} $\neq$ \cd{NIL})~\{\\
% \rule[1mm]{0.30cm}{0pt}
% \cd{y}~\passign\cd{next(x)};\\
% \rule[1mm]{0.30cm}{0pt}
% \pif (\cd{y} $\neq$ \cd{NIL}) \pthen~\{ \\
% \rule[1mm]{0.60cm}{0pt}
% \pif (\cd{k(x)} $\not\leq$ \cd{k(y)}) \pthen \cd{sorted \passign~F;} \\ %~\{ \\
% \rule[1mm]{0.30cm}{0pt}
% \} \\
% \rule[1mm]{0.25cm}{0pt}
% \cd{x}~\passign\cd{y};\\
% \}\\

\rule[1mm]{1cm}{0.5pt} $P_\text{check-key}$ \rule[1mm]{1cm}{0.5pt} \\\\
\passume~(\cd{\!\!T~$\neq$~F});\\
\cd{exists \passign~F;}\\
\pwhile (\cd{x} $\neq$ \cd{NIL})~\{\\
\rule[1mm]{0.30cm}{0pt}
\pif (\cd{k} = \cd{key(x)}) \pthen \cd{exists{\passign}~T;} \\ %~\{ \\
\rule[1mm]{0.25cm}{0pt}
\cd{y}~\passign\cd{next(x)};\\
\rule[1mm]{0.25cm}{0pt}
\cd{x}~\passign\cd{y};\\
\}\\
\end{minipage}
\quad
\begin{minipage}[t]{0.55\textwidth}
\rule[1mm]{0.5cm}{0.5pt} $P_\text{check-key-sorted}$ \rule[1mm]{0.5cm}{0.5pt} \\\\
\passume~(\cd{\!\!T~$\neq$~F});\\
\cd{found \passign~F;}\\
\cd{stop \passign~F;}\\
\pwhile (\cd{x} $\neq$ \cd{NIL})~\{\\
\rule[1mm]{0.30cm}{0pt}
\pif (\cd{stop} $=$ \cd{F}) \pthen~\{ \\
\rule[1mm]{0.60cm}{0pt}
\pif (\cd{k} $=$ \cd{key(x)}) \pthen \cd{found \passign~T;} \\ %~\{ \\
\rule[1mm]{0.60cm}{0pt}
\pif (\cd{k} $\leq$ \cd{key(x)}) \pthen \cd{stop \passign~T;} \\ %~\{ \\
\rule[1mm]{0.30cm}{0pt}
\} \\
\rule[1mm]{0.25cm}{0pt}
\cd{y}~\passign\cd{next(x)};\\
\rule[1mm]{0.25cm}{0pt}
\cd{x}~\passign\cd{y};\\
\}\\
\end{minipage}
\end{minipage}
}
\caption{
% Program $P_\text{sorted}$ (left top) checks if the list starting at $\cd{x}$ is sorted
% by checking if consecutive keys in the list are ordered by $\leq$,
% and sets the variable $\cd{sorted}$ accordingly.
Program $P_\text{check-key}$ (left) % (left bottom)  
checks if the key $\cd{k}$
exists in the list starting at $\cd{x}$ and sets the variable $\cd{exists}$
to $\cd{T}$ if it does.
Program $P_\text{check-key-sorted}$ (right) checks if the list 
starting at $\cd{x}$ contains the key $\cd{k}$ and works as expected
on a sorted list.
$<$ is interpreted as a strict total order.
The condition $a \leq b$ is shorthand for $a < b \lor a = b$
}
\figlabel{sep-progs}
\end{figure}

% \begin{example}

Consider the problem of searching for an element $\cd{k}$ in a sorted list. 
There are two simple algorithms for this problem. 
Algorithm 1 
(\figref{sep-progs}, left) 
walks through the list from beginning to end, 
and if it finds $\cd{k}$, it sets a Boolean variable $
\cd{exists}$ to $\cd{T}$. 
Notice this algorithm does not exploit the sortedness property of the list. 
Algorithm 2 
(\figref{sep-progs}, right) 
also walks through the list, 
but it stops as soon as it either finds $\cd{k}$ or reaches an element that is larger than $\cd{k}$. 
If it finds the element it sets a Boolean variable $\cd{found}$ to $\cd{T}$. 
If both algorithms are run on the same sorted list, 
then their answers (namely, $\cd{exists}$ and $\cd{found}$) must be the same.

%!TEX root = main.tex

% assume (T != F);
% found := F;
% stop := F;
% exists := F;
% while (x != NIL) {
% 	if (stop = F) {
% 		if (k = key(x)) then {
% 			found := T;
% 		}
% 		if (k <= key(x)) then {
% 			stop := T;
% 		}
% 	}
% 	if (k = key(x)) then {
% 		exists := T;
% 	}
% 	y := next(x);
% 	if(y != NIL) then {
% 		assume (key(x) <= key(y));
% 	}
% 	x := y;
% }
% @post: found = exists

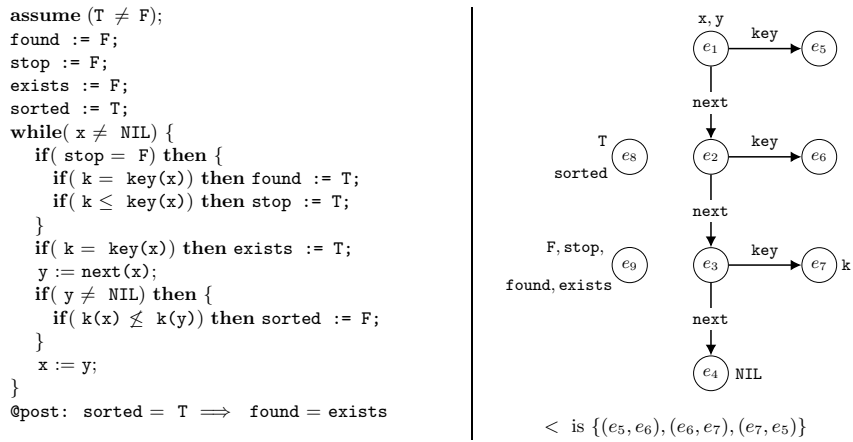
\begin{figure}[t]
\scalebox{0.8}{
% \begin{minipage}[H]{0.45\textwidth}
% \passume~(\cd{\!\!T~$\neq$~F});\\
% \cd{found \passign~F;}\\
% \cd{stop \passign~F;}\\
% \cd{exists \passign~F;}\\
% \pwhile (\cd{x} $\neq$ \cd{NIL})~\{\\
% \rule[1mm]{0.30cm}{0pt}
% \pif (\cd{stop} $=$ \cd{F}) \pthen~\{ \\
% \rule[1mm]{0.60cm}{0pt}
% \pif (\cd{k} $=$ \cd{key(x)}) \pthen~\{ \\
% \rule[1mm]{0.70cm}{0pt}
% \cd{found \passign~T;}\\
% \rule[1mm]{0.60cm}{0pt}
% \} \\
% \rule[1mm]{0.60cm}{0pt}
% \pif (\cd{k} $\leq$ \cd{key(x)}) \pthen~\{ \\
% \rule[1mm]{0.70cm}{0pt}
% \cd{stop \passign~T;}\\
% \rule[1mm]{0.60cm}{0pt}
% \} \\
% \rule[1mm]{0.30cm}{0pt}
% \} \\
% \rule[1mm]{0.30cm}{0pt}
% \pif (\cd{k} $=$ \cd{key(x)}) \pthen~\{ \\
% \rule[1mm]{0.60cm}{0pt}
% \cd{exists \passign~T;}\\
% \rule[1mm]{0.30cm}{0pt}
% \} \\
% \rule[1mm]{0.25cm}{0pt}
% \cd{y}~\passign\cd{next(x)};\\
% \rule[1mm]{0.30cm}{0pt}
% \pif (\cd{y} $\neq$ \cd{NIL}) \pthen~\{ \\
% \rule[1mm]{0.60cm}{0pt}
% \passume~(\cd{k(x)} $\leq$ \cd{k(y)}); \\
% \rule[1mm]{0.30cm}{0pt}
% \} \\
% \rule[1mm]{0.25cm}{0pt}
% \cd{x}~\passign\cd{y};\\
% \}\\
% \code{@post:} \code{found} =\code{exists}
% \\
% \end{minipage}
\noindent
\begin{minipage}[H]{0.62\textwidth}
\passume~(\cd{\!\!T~$\neq$~F});\\
\cd{found \passign~F;}\\
\cd{stop \passign~F;}\\
\cd{exists \passign~F;}\\
\cd{sorted \passign~T;}\\
\pwhile (\cd{x} $\neq$ \cd{NIL})~\{\\
\rule[1mm]{0.30cm}{0pt}
\pif (\cd{stop} $=$ \cd{F}) \pthen~\{ \\
\rule[1mm]{0.60cm}{0pt}
\pif (\cd{k} $=$ \cd{key(x)}) \pthen \cd{found \passign~T;} \\ %~\{ \\
% \rule[1mm]{0.70cm}{0pt}
% \cd{found \passign~T;}\\
% \rule[1mm]{0.60cm}{0pt}
% \} \\
\rule[1mm]{0.60cm}{0pt}
\pif (\cd{k} $\leq$ \cd{key(x)}) \pthen \cd{stop \passign~T;} \\ %~\{ \\
% \rule[1mm]{0.70cm}{0pt}
% \cd{stop \passign~T;}\\
% \rule[1mm]{0.60cm}{0pt}
% \} \\
\rule[1mm]{0.30cm}{0pt}
\} \\
\rule[1mm]{0.30cm}{0pt}
\pif (\cd{k} $=$ \cd{key(x)}) \pthen \cd{exists \passign~T;} \\ %~\{ \\
% \rule[1mm]{0.60cm}{0pt}
% \cd{exists \passign~T;}\\
% \rule[1mm]{0.30cm}{0pt}
% \} \\
\rule[1mm]{0.25cm}{0pt}
\cd{y}~\passign\cd{next(x)};\\
% \rule[1mm]{0.30cm}{0pt}
% \pif (\cd{y} $\neq$ \cd{NIL}) \pthen~\{ \\
% \rule[1mm]{0.60cm}{0pt}
% \passume~(\cd{k(x)} $\leq$ \cd{k(y)}); \\
% \rule[1mm]{0.30cm}{0pt}
% \} \\
\rule[1mm]{0.30cm}{0pt}
\pif (\cd{y} $\neq$ \cd{NIL}) \pthen~\{ \\
\rule[1mm]{0.60cm}{0pt}
\pif (\cd{k(x)} $\not\leq$ \cd{k(y)}) \pthen  \cd{sorted \passign~F;} \\% ~\{ \\
% \rule[1mm]{0.70cm}{0pt}
% \cd{sorted \passign~F;}\\
% \rule[1mm]{0.60cm}{0pt}
% \} \\
\rule[1mm]{0.30cm}{0pt}
\} \\
\rule[1mm]{0.25cm}{0pt}
\cd{x}~\passign\cd{y};\\
\}\\
% \code{@post:} \code{found} =\code{exists}
\code{@post:} \code{sorted} = \cd{T} $\implies$ \code{found} =\code{exists}
\\
\end{minipage}
\vrule
\quad
\begin{minipage}[H]{0.5\textwidth}
\scalebox{0.9}{
\begin{tikzpicture}

\node (node-1) at (1,0) [draw, circle, minimum width = 1em] {$e_1$};
\node (node-2) at (1,-2) [draw, circle, minimum width = 1em] {$e_2$};
\node (node-3) at (1,-4) [draw, circle, minimum width = 1em] {$e_3$};
\node (node-4) at (1,-6) [draw, circle, minimum width = 1em] {$e_4$};

\node (key-1) at (3,0) [draw, circle, minimum width = 1em] {$e_5$};
\node (key-2) at (3,-2) [draw, circle, minimum width = 1em] {$e_6$};
\node (key-3) at (3,-4) [draw, circle, minimum width = 1em] {$e_7$};

% \node (node-1) at (1,0) [draw, circle, minimum width = 1em] {$e_1$};
% \node (node-2) at (1,-1.5) [draw, circle, minimum width = 1em] {$e_2$};
% \node (node-3) at (1,-3) [draw, circle, minimum width = 1em] {$e_3$};
% \node (node-4) at (1,-4.5) [draw, circle, minimum width = 1em] {$e_4$};

% \node (key-1) at (3,0) [draw, circle, minimum width = 1em] {$e_5$};
% \node (key-2) at (3,-1.5) [draw, circle, minimum width = 1em] {$e_6$};
% \node (key-3) at (3,-3) [draw, circle, minimum width = 1em] {$e_7$};

\node (x-y) at (1,0.5) [rounded rectangle] {$\cd{x}, \cd{y}$};
\node (NIL) at (1.7,-6) [rounded rectangle] {$\cd{NIL}$};

\node (k) at (3.5,-4) [rounded rectangle] {$\cd{k}$};

\node (T) at (-1,-1.7) [rounded rectangle] {$\cd{T}$};
\node (sorted) at (-1.4,-2.3) [rounded rectangle] {$\cd{sorted}$};
\node (hat-T) at (-0.5,-2) [draw, circle, minimum width = 2em] {$e_8$};
\node (F) at (-1.5,-3.7) [rounded rectangle] {$\cd{F}, \cd{stop}, $};
\node (F) at (-1.8,-4.4) [rounded rectangle] {$\cd{found}, \cd{exists}$};
\node (hat-F) at (-0.5,-4) [draw, circle, minimum width = 2em] {$e_9$};

\draw (node-1) edge[-{Latex[length=2mm, width=2mm]}, thick] node[fill=white] {\cd{next}} (node-2);
% \node (next-1-2) at (1,-0.6) [rounded rectangle, fill=white] {$\cd{next}$};
\draw (node-2) edge[-{Latex[length=2mm, width=2mm]}, thick] node[fill=white] {\cd{next}} (node-3);
% \node (next-2-3) at (1,-2.1) [rounded rectangle, fill=white] {$\cd{next}$};
\draw (node-3) edge[-{Latex[length=2mm, width=2mm]}, thick] node[fill=white] {\cd{next}} (node-4);
% \node (next-3-4) at (1,-3.6) [rounded rectangle, fill=white] {$\cd{next}$};

\draw (node-1) edge[-{Latex[length=2mm, width=2mm]}, thick] node[above] {\cd{key}} (key-1);
% \node (key-edge-1) at (2,0.2) [rounded rectangle] {$\cd{key}$};
\draw (node-2) edge[-{Latex[length=2mm, width=2mm]}, thick] node[above] {\cd{key}} (key-2);
% \node (key-edge-2) at (2,-1.3) [rounded rectangle] {$\cd{key}$};
\draw (node-3) edge[-{Latex[length=2mm, width=2mm]}, thick] node[above] {\cd{key}} (key-3);
% \node (key-edge-3) at (2,-2.8) [rounded rectangle] {$\cd{key}$};
\end{tikzpicture}
}
\[< \text{ is } \set{(e_5,e_6), (e_6,e_7), (e_7,e_5)}\]
\end{minipage}
}
\caption{Left: Uninterpreted program for finding a key $\cd{k}$ in a list starting at $\cd{x}$ with $<$ interpreted as a strict total order. The condition $a \leq b$ is shorthand for $a < b \lor a = b$.
Right: A model in which $<$ is not interpreted as a strict total order. The elements in the universe of the model are denoted using circles. Some elements are labeled with variables denoting the initial values of these variables.  The edges \protect\drawdirectedline~represent subterm relation. Not all functions are shown in the figure. The model does not satisfy the post-condition on the program on left.}
\figlabel{example}
\end{figure}

\figref{example} (on the left) shows a program that weaves the above two algorithms together (treating Algorithm~1 as the specification for Algorithm~2). 
The variable $\cd{x}$ walks down the list using the $\cd{next}$ pointer. 
The variable $\cd{stop}$ is set to $\cd{T}$ when Algorithm 2 stops searching in the list. 
% The pre condition, namely that the input list is sorted, 
% is captured by \passume-statements that assert that consecutive
% elements are ordered as the list is traversed. 
The precondition, namely that the input list is sorted, 
is captured by tracking another variable $\cd{sorted}$
whose value is $\cd{T}$ if
consecutive elements are ordered as the list is traversed. 
The post condition demands that whenever the list is sorted,
$\cd{found}$ and $\cd{exists}$ 
be equal when the list has been fully traversed. 
Note that the program's correctness is specified using only
quantifier-free assertions using the same vocabulary as the program. 
% Notice that the post condition must hold when the algorithm runs on a sorted list.

The program works on a data domain that provides interpretations for the functions $\cd{key}$, $\cd{next}$, the initial values of the variables, and the relation $<$. When $<$ is interpreted to be a strict total order, the program is correct. However, if $<$ is not interpreted as a total order, then the program may be incorrectly deemed as buggy. To see this, consider the data model shown on the right in \figref{example}. The data domain has 9 elements in its universe, with the functions $\cd{next}$ and $\cd{key}$ interpreted as shown. Initially, $\cd{x}, \cd{y}$ have value $e_1$, $\cd{NIL}$ is $e_4$, $\cd{k}$ is $e_7$, 
$\cd{T}$ and $\cd{sorted}$ are $e_8$, and $\cd{F}, \cd{found}, \cd{exists}$, and $\cd{stop}$ are $e_9$. 
The interpretation of $<$ is as follows --- $e_5 < e_6$, $e_6 < e_7$, and $e_7 < e_5$. 
Clearly $<$ is not an order, but the program's
\emph{sortedness} check $\dblqt{\cd{sorted} = \cd{T}}$ will pass.
% $\passume$-statements that capture ``sortedness'' will pass. 
After the entire list is processed, $\cd{exists}$ will be set to $\cd{T}$ when $\cd{x} = e_3$. 
On the other hand, $\cd{stop}$ will be set to $\cd{T}$ when $\cd{x} = e_1$ because $\cd{k} = e_7 < \cd{key}(\cd{x})$. 
Therefore, at the end $\cd{found} = \cd{F} \neq \cd{exists}$. 
The work presented in~\cite{coherence2019}, where all functions and relations are uninterpreted, would therefore declare this program to be incorrect.

The goal of this paper is to explore several natural restrictions on data models and study the problem of verifying coherent programs for them. When $<$ is constrained to be a total order, the program in \figref{example} is correct and coherent. Our results (see~\secref{strict-total-order}) show that verification of such programs when relations are constrained to be strict total orders is decidable, and hence we can build automatic decision procedures that will correctly verify such programs. 

%!TEX root = main.tex

\section{Preliminaries}
\seclabel{prelim}

We briefly recall the syntax and semantics of uninterpreted programs
and the verification problem modulo axioms.
Our presentation closely follows~\cite{coherence2019}.

\subsection{Program Syntax}

We consider imperative programs with loops over a fixed finite 
set of variables $V$ and use constant 
($\Cc$), function ($\Ff$), and predicate ($\Rr$) 
symbols belonging to some first order signature 
$\Sigma = (\Cc, \Ff, \Rr)$. 
Programs are then given by the syntax below:
\begin{align*}
\stmt ::=& 
\,\,  x \passign c \,
\mid \, x \passign y \, 
\mid \, x \passign f(\vec{z}) \, 
\mid \, \passume \, (\cond) \, 
\mid \, \pskip \,
\mid \, \stmt \, ;\, \stmt \\
&
\mid \, \pwhile \, (\cond) \, \stmt \,
\mid \, \pif \, (\cond) \, \pthen \, \stmt \, \pelse \, \stmt \, \\
\cond ::=& 
\, x = y \,
\mid \, x = c \,
\mid \, c = d \,
\mid \, R(\vec{z}) \,
\mid \, \cond \lor \cond \, 
\mid \, \neg \cond
\end{align*}
Here, $f \in \Ff$, $R \in \Rr$, $c, d \in \Cc$, $x,y \in V$, and 
$\vec{z}$ is a tuple of variables in $V$ and constants in $\Cc$.
The syntax allows programs to have assignment statements,
conditionals (\pif-\pthen-\pelse), looping constructs (\pwhile)
and sequencing.
Since constants can be modeled using variables that are never re-assigned,
we will assume, without loss of generality, 
that the programs do not use constants.
Further, arbitrary Boolean combinations of atomic predicates
can be expressed using the \pif-\pthen-\pelse construct, and
henceforth, we will also assume that all conditionals are atomic
(i.e., of the form $x = y$, $x \neq y$, $R(\vec{z})$ or $\neg R(\vec{z})$). 

\subsection{Executions and Semantics of Uninterpreted Programs}
\seclabel{exec-semantics}

Executions of programs over $\stmt$ are words over the following alphabet
\begin{align*}
\Pi = \setpred{ \dblqt{x \passign y}, \dblqt{x \passign f(\vec{z})},  
\dblqt{\passume (x=y)}, \dblqt{\passume (x\neq y)}, \\
\dblqt{\passume (R(\vec{z}))},
\dblqt{\passume (\neg R(\vec{z}))}
}{
x, y, \vec{z} \textit{~are~in~} V
}
\end{align*}
For a program $s \in \stmt$, the set of executions of
of $s$, denoted $\exec(s)$ is a regular language over the alphabet $\Pi$
and is given as follows (similar to~\cite{coherence2019}).
\begin{align*}
\begin{array}{rcl}
\exec(\pskip) &=& \epsilon  \\
\exec(x \passign y) &=& \dblqt{x \passign y} \\
\exec(x \passign f(\vec{z})) &=& \dblqt{x \passign f(\vec{z})} \\
\exec(\passume(c)) &=& \dblqt{\passume(c)} \\
\exec(\pif\,c \, \pthen\, s_1\,  \pelse\, s_2~~) &=& \dblqt{\passume(c)} \cdot \exec(s_1)
 + \dblqt{\passume(\neg c)} \cdot \exec(s_2) \\
\exec(~~s_1 ; s_2~~) &=& \exec(s_1) \cdot \exec(s_2) \\
\exec(\pwhile\, c\, \{s\}~~) &=&
        [\dblqt{\passume(c)} \cdot \exec(s_1)]^* \cdot \dblqt{\passume(\neg c)}
\end{array}
\end{align*}
The set of partial executions of $s$ is the set of 
prefixes of words in $\exec(s)$ and is also regular.

A data model $\Mm = (U_\Mm, \sem{}_\Mm)$ for signature $\Sigma$
is a first order structure with a universe $U_\Mm$ of elements 
and interpretations for the constants ($\setpred{\sem{c}_\Mm}{c \in \Cc}$),
functions ($\setpred{\sem{f}_\Mm}{f \in \Ff}$) and 
relations ($\setpred{\sem{R}_\Mm}{R \in \Rr}$). 
Given a first order structure $\Mm$ over $\Sigma$
(also refered to as a \emph{data model} in the rest of the presentation),
and an execution $\rho \in \Pi^*$, the semantics
of $\rho$ on $\Mm$ is given by $\valplain_\Mm : \Pi^* \times V \to U_\Mm$
that gives the the valuation of variables in $V$
at the end of an execution, and is defined as follows.
Below, we assume that every variable $x \in V$
is associated with a designated constant $\init{x} \in \Cc$
which denotes its initial value.
\begin{align*}
\begin{array}{rcll}
\val{\Mm}{\epsilon}{x} &=& \sem{\init{x}}_\Mm & \text{for every } x \in V \\ 
\val{\Mm}{\rho \cdot \dblqt{x \passign y}}{z} &=& \val{\Mm}{\rho}{y} & \text{if } z \text{ is } x\\
\val{\Mm}{\rho \cdot \dblqt{x \passign f(z_1, \ldots, z_r)}}{y} &=& \sem{f}_\Mm(\val{\Mm}{\rho}{z_1}, \ldots, \val{\Mm}{\rho}{z_r}) & \text{if } y \text{ is } x\\
\val{\Mm}{\rho \cdot a}{x} &=& \val{\Mm}{\rho}{x} & \text{otherwise}
\end{array}
\end{align*}

\begin{example}
\exlabel{example1}
Let us consider the program
in~\figref{example}.
While the program does not strictly obey the syntax of $\stmt$,
it can be easily transformed into one --- 
all statements of the form $\pif~(c)~\pthen~s$ can be transformed to 
$\pif~(c)~\pthen~s~\pelse~\pskip$. Further, complex assume statements
like `$\passume(\cd{k} = \cd{key(x)})$' can be transformed
using additional variables --- in this case to 
`$\cd{kx} \passign \cd{key(x)} ; \passume(\cd{k} = \cd{kx})$',
where $\cd{kx}$ is a new variable.

Now, let us consider the following execution of this program.
\begin{align*}
\pi &= \pi_0
\cdot \passume(\cd{x} \neq \cd{NIL}) \cdot \pi_1 
\cdot \passume(\cd{x} \neq \cd{NIL})  \cdot \pi_2 
\cdot \passume(\cd{x} \neq \cd{NIL}) \cdot \pi_3
\cdot \passume(\cd{x} = \cd{NIL}) 
\end{align*}
This execution corresponds to entering the loop body
exactly three times. $\pi_0$ corresponds to the statements executed 
prior to entering the loop for the first time,
and $\pi_1$, $\pi_2$ and $\pi_3$ correspond to the body of the loop
in the first, second and third iteration:
\begin{align*}
\begin{array}{lcl}
\pi_0\!\!\!\!\!\!&=&\!\!\!\!\passume(\cd{T} \neq \cd{F}) \cdot 
\cd{found} \passign \cd{F} \cdot \cd{stop} \passign \cd{F}
\cdot \cd{exists} \passign \cd{F} \cdot \cd{sorted} \passign \cd{T} \\
\pi_1\!\!\!\!\!\!&=&\!\!\!\!\passume(\cd{stop} = \cd{F}) \cdot \passume(\cd{k} \neq \cd{key(x)})
\cdot \passume(\cd{k} < \cd{key(x)})
\cdot \cd{stop} \passign \cd{T} \\
&&\!\!\!\!\!\!\cdot \passume(\cd{k} \neq \cd{key(x)}) \cdot \cd{y} \passign \cd{next(x)}
\cdot \passume(\cd{y} \neq \cd{NIL}) \cdot \passume(\cd{key(x)} < \cd{key(y)})
\cdot \cd{x} \passign \cd{y} \\
\pi_2\!\!\!\!\!\!&=&\!\!\!\!\passume(\cd{stop} \neq \cd{F}) \cdot \passume(\cd{k} \neq \cd{key(x)}) \cdot \cd{y} \passign \cd{next(x)}
\cdot \passume(\cd{y} \neq \cd{NIL}) \\
&&\!\!\!\!\!\!\cdot \passume(\cd{key(x)} < \cd{key(y)})
\cdot \cd{x} \passign \cd{y} \\
\pi_3\!\!\!\!\!\!&=&\!\!\!\!\passume(\cd{stop} \neq \cd{F}) \cdot \passume(\cd{k} = \cd{key(x)}) 
\cdot \cd{exists} \passign \cd{T}
\cdot \cd{y} \passign \cd{next(x)}
\cdot \passume(\cd{y} = \cd{NIL}) \cdot \cd{x} \passign \cd{y}
\end{array}
\end{align*}

Now consider the model $\Mm$ shown in~\figref{example} on the right.
For this model we have
$\val{\Mm}{\pi}{\cd{sorted}} = \val{\Mm}{\pi}{\cd{T}} = \val{\Mm}{\pi}{\cd{exists}} = e_8$,
and $\val{\Mm}{\pi}{\cd{found}} = \val{\Mm}{\pi}{\cd{F}} = e_9$.
\end{example}

\subsection{Feasibility of Executions Modulo Axioms}

An execution is said to be \emph{feasible} in a data model, if every
assumption made in the execution, holds on the model.  More precisely,
an execution $\rho$ is feasible in $\Mm$ if for every prefix $\sigma'
= \sigma \cdot \dblqt{\passume~c}$ of $\rho$, we have
\begin{enumerate*}[label=(\alph*)]
\item $\val{\Mm}{\sigma}{x} = \val{\Mm}{\sigma}{y}$ if $c = (x = y)$,
\item $\val{\Mm}{\sigma}{x} \neq \val{\Mm}{\sigma}{y}$ if $c = (x \neq
  y)$,
\item $(\val{\Mm}{\sigma}{z_1}, \ldots,\val{\Mm}{\sigma}{z_r}) \in
  \sem{R}_\Mm$ if $c = R(z_1, \ldots, z_r)$, and
\item $(\val{\Mm}{\sigma}{z_1}, \ldots, \val{\Mm}{\sigma}{z_r})$
  $\not\in \sem{R}_\Mm$ if $c = \neg R(z_1, \ldots, z_r)$.
\end{enumerate*}

Let $\Aa$ be a set of first order sentences, including possible ground
atomic predicates~\footnote{A ground atomic predicate is of the form
  $t_1 \sim t_2$, or $R(t_1,\ldots t_k)$ or $\neg R(t_1,\ldots t_k)$,
  where $\sim \in \set{=,\neq}$, $R$ is a relation symbol, and $t_i$s
  are ground terms.}.  We say that a data model $\Mm$ is an
$\Aa$-model, denoted $\Mm \models \Aa$, if for every $\varphi \in
\Aa$, we have $\Mm \models \varphi$.  A formula $\varphi$ is
$\Aa$-valid, denoted $\Aa\models \varphi$, if $\phi$ holds in every
model $\Mm$ that satisfies $\Aa$.

An execution $\rho$ is said to be \emph{feasible modulo $\Aa$} 
if there is an $\Aa$-model $\Mm$ such that $\rho$ is feasible in $\Mm$.

\begin{example}
Let us again consider the execution $\pi$ from~\exref{example1}.
We first observe that $\pi$ is feasible on the model $\Mm$ from~\figref{example} (right).

Now let us consider the set of axioms $\Aa_\sto$ that states 
that the relation symbol $<$
used in the program in~\figref{example} (left) is interpreted to be
a strict total order.
That is
\begin{align*}
\Aa_\sto &=& \{ \underbrace{\forall x.~\neg(x < x)}_{\text{irreflexivity}}, \,
\underbrace{\forall x, y, z.~x < y \land y < z{\implies}x < z}_{\text{transitivity}}, \,
\underbrace{\forall x, y.~x = y \lor x < y \lor y < x}_{\text{totality}}\}
\end{align*}
Observe that the model $\Mm$ is not a $\Aa_\sto$-model because
there is a cyclic dependency --- $e_5 < e_6$, $e_6 < e_7$ and $e_7 < e_5$.
Now consider the model $\Mm'$ which differs from $\Mm$ only in the interpretation
of $<$ as:  $\sem{<}_{\Mm'} = \set{(e_5, e_6), (e_6, e_7), (e_5, e_7)}$.
It is easy to see that $\Mm'$ is an $\Aa_\sto$ model
and the execution $\pi$ is not feasible on $\Mm'$.
In fact, there is no $\Aa_\sto$-model on which $\pi$ is feasible, or, as we say,
$\pi$ is \emph{infeasible modulo} $\Aa_\sto$.
\end{example}

\subsection{Program Verification Modulo Axioms}

We consider programs annotated with post-conditions
that are over the following syntax below.
Here, $x, y$ and $\vec{z}$ belong to the set of program variables $V$
and $R \in \Rr$ is a relation symbol in $\Sigma$.
\begin{align*}
\Ll: ~~~~~\varphi ::=   x\!=\!y ~\mid~ R(\vec{z}) ~\mid~ \varphi \vee \varphi ~\mid~ \neg \varphi
\end{align*}
\begin{definition}[Program Verification Modulo Axioms]
For a program $s$ and a set of axioms $\Aa$, we say that $s$ satisfies
a postcondition $\varphi$ over the syntax $\Ll$ \emph{modulo} $\Aa$ if
for every $\Aa$-model $\Mm$ and for execution $\rho \in \exec(s)$ 
that is feasible in $\Mm$, $\Mm$ satisfies $\varphi[\val{\Mm}{\rho}{V}/V]$
(i.e., where each variable $x \in V$ is replaced by $\val{\Mm}{\rho}{V}$).
\end{definition}

We remark that one can alternatively phrase the verification problem
stated above in terms of feasibility.
That is, a program $s$ satisfies a postcondition $\varphi$ modulo $\Aa$ iff 
every execution $\rho$ of 
$s'$ is infeasible modulo $\Aa$ (i.e., there is no $\Aa$-model $\Mm$
such that $\rho$ is feasible in $\Mm$),
where $s' = s; \passume(\neg \varphi)$.

%!TEX root = main.tex

\section{Coherence Modulo Axioms}
\seclabel{coherence}

In this section we extend the notion of coherence 
from~\cite{coherence2019}, adapting it to
our current setting where we restrict data models using axioms $\Aa$.
We will first recall the notion of terms computed by an execution,
which will be used to define the notion of coherence.

\subsection{Terms Computed and Assumptions Accumulated by Executions}

We will associate a syntactic term with each variable after a partial
execution $\rho$. This, intuitively, is the term computed by $\rho$
and stored in $x$. Let $\Terms_\Sigma$ be the set of terms built using
constants and functions in $\Sigma$. The term stored in $x$ after
$\rho$ is defined inductively on $\rho$ as follows.
\[
\begin{array}{rcll}
\comp(\epsilon,x) &=& \init{x} & \text{for every } x \in V \\
\comp(\rho \cdot \dblqt{x \passign y},z) &=& \comp(\rho,y) & \text{if\
 } z \text{ is } x\\
\comp(\rho \cdot \dblqt{x \passign f(z_1, \ldots, z_r)},y) &=&
f(\comp(\rho,z_1), \ldots, \comp(\rho,z_r)) & \text{if } y \text{ is }
x\\
\comp(\rho \cdot a,x) &=& \comp(\rho,x) & \text{otherwise}
\end{array}
\]
The set of terms computed by an execution $\rho$ is $\Terms(\rho) =
\setpred{\comp(\rho', x)}{\rho' \text{ is a prefix of } \rho, x \in
  V}.$

As an execution proceeds, it accumulates assumptions over the terms it
computes, and we will use $\kappa(\rho)$ to denote the assumptions
made by the execution $\rho$.  In~\cite{coherence2019}, relations are
modeled using functions (to Booleans) and hence relational assumes
were avoided.  In the current exposition, however, we will treat
relations as first class objects and the set of assumptions will also
include relational predicates.  Formally, $\kappa(\rho)$ is a set of
ground predicates over $\Sigma \cup \set{=}$ defined as follows.
\begin{align*}
\kappa(\varepsilon) &= \emptyset \\
\kappa(\sigma \cdot \dblqt{\passume(x = y)}) &= \kappa(\sigma) \cup  \set{\comp(\sigma, x) = \comp(\sigma, y)} \\
\kappa(\sigma \cdot \dblqt{\passume(x \neq y)}) &= \kappa(\sigma) \cup  \set{\comp(\sigma, x) \neq \comp(\sigma, y)} \\
\kappa(\sigma \cdot \dblqt{\passume(R(z_1, z_2, \ldots, z_k))}) &= \kappa(\sigma) \cup  \set{R(\comp(\sigma, z_1), \ldots, \comp(\sigma, z_k))} \\
\kappa(\sigma \cdot \dblqt{\passume(\neg R(z_1, z_2, \ldots, z_k))}) &= \kappa(\sigma) \cup  \set{\neg R(\comp(\sigma, z_1), \ldots, \comp(\sigma, z_k))} \\
\kappa(\sigma \cdot a) &= \kappa(\sigma) \quad \text{otherwise}
\end{align*}

In~\cite{coherence2019}, relations are
modeled using functions (to Booleans) and hence relational assumes
were avoided.  In the current exposition, however, we will treat
relations as first class objects and the set of assumptions will also
include relational predicates.  Formally, $\kappa(\rho)$ is a set of
ground predicates over $\Sigma \cup \set{=}$ defined as follows.
\begin{align*}
\kappa(\varepsilon) &= \emptyset \\
\kappa(\sigma \cdot \dblqt{\passume(x = y)}) &= \kappa(\sigma) \cup  \set{\comp(\sigma, x) = \comp(\sigma, y)} \\
\kappa(\sigma \cdot \dblqt{\passume(x \neq y)}) &= \kappa(\sigma) \cup  \set{\comp(\sigma, x) \neq \comp(\sigma, y)} \\
\kappa(\sigma \cdot \dblqt{\passume(R(z_1, z_2, \ldots, z_k))}) &= \kappa(\sigma) \cup  \set{R(\comp(\sigma, z_1), \ldots, \comp(\sigma, z_k))} \\
\kappa(\sigma \cdot \dblqt{\passume(\neg R(z_1, z_2, \ldots, z_k))}) &= \kappa(\sigma) \cup  \set{\neg R(\comp(\sigma, z_1), \ldots, \comp(\sigma, z_k))} \\
\kappa(\sigma \cdot a) &= \kappa(\sigma) \quad \text{otherwise}
\end{align*}

\subsection{Coherence}

Our definition of coherence modulo axioms 
is a smooth generalization of the definition of 
coherence in~\cite{coherence2019}. 
The notion of coherence consists of two properties --- 
\emph{memoizing} and \emph{early equality assumes}. 
The memoizing property says, intuitively, when a term $t$ is computed
after executing some prefix $\sigma$ of an execution,
if $t$ is equivalent to some other term modulo the assumptions 
made in the execution so far,
then $t$ must not have been \emph{dropped} at the end of $\sigma$,
i.e., a program variable must already hold this term. 
We replace the notion of equivalence of terms in this definition 
by equivalence modulo the axioms as well. 

The notion of early assumes in~\cite{coherence2019} intuitively says
that assumptions of equality (on terms $t_1$ and $t_2$) should be
encountered early --- earlier than \emph{dropping} any superterm of
$t_1$ or $t_2$.  This notion of early assumes allows for effectively
computing \emph{congruence closure} on the set of terms computed by
the execution, which in turn, is necessary to accurately maintain
which terms are equivalent.  However, we observe that the notion
in~\cite{coherence2019} is too restrictive and not entirely necessary.
In our paper, we generalize this notion in several ways, to a more
semantic one as follows.  Whenever an execution encounters an
assumption of equality between two term, we instead demand that only
the equivalences that are \emph{additionally} implied by this new
assumption, can be infered \emph{locally} using the already known
congruence between terms in the \emph{window}, i.e., the set of terms
pointed to by the program variables when the equality assumption is
encountered. Next, we incorporate axioms into this definition, by
requiring that the notion of equivalence is also modulo the axioms,
and further require that \emph{all} assumptions (equality,
disequality, relational) are required to be early (as against only
restricting equality assumptions to be early like
in~\cite{coherence2019}).  We will elaborate on these differences
using an example after presenting the formal definition next.
% We now formalize these ideas into a definition below.

Given a set of first order sentences $\Gamma$ and 
ground terms $t_1$ and $t_2$, we say that 
$t_1 \congcl{\Gamma} t_2$ if $\Gamma \models t_1 = t_2$.

\begin{definition}[Coherence modulo axioms]
\deflabel{coherence-def}
Let $\Aa$ be a set of axioms and let $\rho$ be a 
complete or partial execution over variables $V$. 
Then, $\rho$ is said to be coherent modulo $\Aa$ 
if it satisfies the following two properties.
\begin{description}
\item[Memoizing.]  Let $\pi = \sigma \cdot
  \dblqt{x\!\passign\!\!f(\vec{z})}$ be a prefix of $\rho$ and let $t
  =\comp(\pi , x)$.  If there is a term $t' \in \Terms(\sigma)$ such
  that $t' \congcl{\Aa \cup \kappa(\sigma)} t$, then there must exist
  some variable $y \in V$ such that $\comp(\sigma, y) \congcl{\Aa \cup
    \kappa(\sigma)} t$.
	
\item[Early Assumes.]  Let $\pi = \sigma \cdot \dblqt{\passume(c)}$ be
  a prefix of $\rho$, where $c$ is any of $x\!=\!y$, $x\!\neq\!y$,
  $R(\vec{z})$, or $\neg R(\vec{z})$.  Let $t \in \Terms(\sigma)$ be a
  term computed in $\sigma$ such that $t$ has been \emph{dropped},
  i.e., for every $x \in V$, we have $\comp(\sigma, x) \ncongcl{\Aa
    \cup \kappa(\sigma)} t$.  For any term $t' \in \Terms(\sigma)$, if
  $t \congcl{\Aa \cup \kappa(\pi)} t'$, then $t \congcl{\Aa \cup
    \kappa(\sigma)} t'$.
\end{description}
\end{definition}

\textbf{Remark.} We remark that every execution that is coherent
as per the definition in~\cite{coherence2019}, is also coherent
modulo $\Aa = \emptyset$ as in~\defref{coherence-def}.
However, the converse is not true and we illustrate this difference below.

\begin{example}
Let us fix $\Aa = \emptyset$ for this example.
Consider the execution $\rho = \sigma\cdot \passume(\cd{x} = \cd{y})$ 
where,
\begin{align*}
\sigma = \passume(\cd{x} = \cd{y}) \cdot \cd{x'} \passign \cd{f(x)} 
\cdot \cd{y'} \passign \cd{f(y)} \cdot \cd{x'} \passign \cd{f(x')}
\cdot \cd{y'} \passign \cd{f(y')}
\end{align*}
We first observe that the prefix $\sigma$ is coherent
both with respect to the definition in~\cite{coherence2019} and~\defref{coherence-def}.
First there are no superterms of $\init{x} = \comp(\epsilon, x)$
and $\init{y} = \comp(\epsilon, y)$ when the first statement $\passume(\cd{x} = \cd{y})$ is observed, and thus, this assume is early.
Second, even though the statements $\dblqt{\cd{y'} \passign \cd{f(y)}}$
and $\dblqt{\cd{y'} \passign \cd{f(y')}}$ are computing a term
that has been equivalently computed before
(modulo the assumption $\set{\init{x} = \init{y}}$),
a copy of these terms is available in some program variable
(variable $\cd{x'}$ in both the cases)
at the time of the execution, thus respecting the memoizing restriction.

Now let us discuss the execution $\rho$. This execution is not coherent with respect to~\cite{coherence2019}.
In particular, the last assume $\passume(\cd{x} = \cd{y})$ 
is not early,
as superterms $\cdm{f(\init{x})}$ 
and $\cdm{f(\init{y})}$
have been computed but \emph{dropped} in the prefix $\sigma$.
However, observe that 
$\cdm{f(\init{x})} \congcl{\Aa \cup \kappa(\sigma)} \cdm{f(\init{y})}$
(here, $\Aa \cup \kappa(\sigma) = \set{\cdm{\init{x}} = \cdm{\init{y}}}$) 
and thus, $\rho$ meets the early assumes restriction as per~\defref{coherence-def},
making $\rho$ coherent.
\end{example}

Let us now consider an example which illustrates the
notion of coherence in the presence of axioms.
\begin{example}
Let us now illustrate the
notion of coherence in the presence of axioms
using the execution $\rho$ below.
\begin{align*}
\rho = 
		\cdm{z_1} \passign \cdm{f(x, y)} 
\cdot 	\cdm{z_2} \passign \cdm{f(y, x)}
\cdot 	\cdm{z_3} \passign \cdm{g(z_1)} 
\cdot 	\cdm{z_4} \passign \cdm{g(z_2)}
\cdot 	\cdm{z_3} \passign \cdm{z_5}
\cdot 	\cdm{z_6} \passign \cdm{g(z_1)}
\end{align*}
Let $\rho_i$ denote the prefix of $\rho$ of length $i$.
Here, $\comp(\rho_3, \cdm{z_3}) = \cdm{g(f(\init{x}, \init{y}))}$,
$\comp(\rho_5, \cdm{z_3}) = \cdm{\init{z_5}} \neq \cdm{g(f(\init{x}, \init{y}))}$
and $\comp(\rho_6, \cdm{z_6}) = \cdm{g(f(\init{x}, \init{y}))}$.
When the set of axioms is $\Aa = \emptyset$,
this execution is not coherent modulo $\Aa$
as it violates the memoizing requirement at the 
last statement $\cdm{z_6} \passign \cdm{g(z_1)}$
(no variable stores the term $\cdm{g(f(\init{x}, \init{y}))}$
after $\rho_5$).

Now, consider the axiom set denoting commutativity of $\cdm{f}$,
i.e., $\Aa_\comm = \set{\forall u, v. \cd{f}(u, v) = \cd{f}(u, v)}$.
In this case, we observe that 
$\cdm{f(\init{x}, \init{y})} \congcl{\Aa_\comm} \cdm{f(\init{y}, \init{x})}$
and thus 
$\cdm{g(f(\init{x}, \init{y}))} \congcl{\Aa_\comm} \cdm{g(f(\init{y}, \init{x}))}$.
Also, $\comp(\rho_5, \cdm{z_4}) = \cdm{g(f(\init{y}, \init{x}))} \congcl{\Aa_\comm} \cdm{g(f(\init{x}, \init{y}))}$.
This ensures that $\rho$ is indeed coherent modulo $\Aa_\comm$.
\end{example}

Let $\cohexec(\Sigma, V, \Aa)$ denote the set of executions over the signature 
$\Sigma$ and variables $V$ that are coherent modulo the set of axioms $\Aa$. 
\begin{definition}
A program $s$ over signature $\Sigma$ and variables $V$ 
is said to be coherent modulo $\Aa$ if $\exec(s) \subseteq \cohexec(\Sigma, V, \Aa)$.
\end{definition}

In this paper, we explore several classes of axioms, 
studying when the verification problem for coherent programs modulo the axioms is decidable.

\subsection{Recap of results from\cite{coherence2019}}
We briefly state the main decidability results 
from~\cite{coherence2019} about coherent programs, 
using the notation defined above, so the set of axioms $\Aa$ is empty.
The results hold even when the early assumes condition is 
generalized (\defref{coherence-def}) and relations are treated as 
first class objects, as we do in this paper.

\begin{theorem}[Essentially~\cite{coherence2019}]
\thmlabel{popl19}
Let $\Sigma$ be a first order signature and $V$ a finite set of
variables. The following observations hold when the set of axioms is
empty.
\begin{enumerate}[topsep=0pt]
\item There is a finite automaton $\Ff$ (effectively constructable) of
  size $O(2^{\text{poly}(|V|)})$ such that for any coherent execution
  $\rho$, $\Ff$ accepts $\rho$ iff $\rho$ is feasible.
\item There is a finite automaton $\Cc$ (effectively constructible) of
  size $O(2^{\text{poly}(|V|)})$ such that $L(\Cc) = \cohexec(\Sigma, V,
  \emptyset)$.
\end{enumerate}
As a consequence, the following problems are decidable in $\pspc$.
\begin{itemize}[topsep=0pt]
\item Given a coherent program $P$, determine if $P$ is correct.
\item Given a program $P$, determine if $P$ is coherent.
\end{itemize}
The problems of verifying coherent programs and checking coherence,
are also $\pspc$-hard.
\end{theorem}

\begin{proof}[Proof Sketch]
These observations have been proved in~\cite{coherence2019}, but the
proof is also sketched in~\appref{automata-relations} for completeness
and to account for the modified definitions. Intuitively, the automata
to check feasibility and coherence of executions, track equivalences
between program variables, functional and relational correspondences
between them that hold based on the $\passume$s seen. Crucial to
establishing the correctness of the automata constructions is the
observation that, when the set of axioms is empty, equality of two
terms does not depend on disequality and relational assumes seen in
the execution. That is, if $\kappa(\rho)_{\text{eq}}$ denotes the set
of equality assumes in $\rho$, then for any computed terms $t_1,t_2$,
$t_1 \congcl{\kappa(\rho)} t_2$ iff $t_1
\congcl{\kappa(\rho)_{\text{eq}}} t_2$.
\end{proof}

%!TEX root = main.tex

\section{Axioms over Relations}
\seclabel{relations}

In this section, we investigate the decidability of the
verification problem for coherent programs modulo
relational axioms, i.e., axioms which only involve relation
symbols $\Rr$ in the signature $\Sigma$.

\subsection{Verification modulo EPR axioms}

A first-order formula is said to be an EPR formula~\cite{Ramsey1987}
if it is of the form
$$ \exists x_1 \ldots x_k \forall y_1, \ldots y_m ~\varphi$$
where $\varphi$ is quantifier-free and purely relational (uses no function symbols). 

It is well known that satisfiability of EPR formulas is decidable, in
fact by a reduction to Boolean
satisfiability~\cite{epr}. Consequently, the problem of checking
whether a single execution is feasible under axioms written in EPR can
be shown to be decidable, and has been exploited in bounded
model-checking.

Consequently, we could reasonably ask whether verification of coherent
programs under EPR axioms is decidable. Surprisingly, we show that
they are not 
% (proof details can be found in~\cite{techreport}).
(proof details can be found in~\appref{epr-undec}).
\begin{theorem}
\thmlabel{epr-undec}
Verification of uninterpreted 
coherent programs modulo EPR axioms is undecidable.\qed
\end{theorem}

Given the above result, we turn to several classes of quantified
axioms, which are all expressible in EPR (and hence have a decidable
bounded model checking problem) and examine their decidability for
coherent program verification.

\subsection{Reflexivity, Irreflexivity, and Symmetry}
\seclabel{refl-irrefl-symm}

We consider program verification under the following
axioms (individually):

\begin{align}
\begin{array}{lll}
\varphi^R_\refl \,\,\,\,\, \triangleq & \forall x \cdot R(x, x) & \quad \text{(reflexivity)}\\
\varphi^R_\irrefl \,\,\, \triangleq & \forall x \cdot \neg R(x, x) & 
  \quad \text{(irreflexivity)}\\
\varphi^R_\sym \triangleq & \forall x, y \cdot R(x, y) \implies R(y, x) & 
  \quad \text{(symmetry)}
\end{array}
\end{align}

We show that verification is decidable modulo these axioms using a
technique that we call \emph{program instrumentation}. Let us fix a
relation $R$ and an axiom $\varphi^R_p$, where $p \in
\set{\refl,\irrefl,\sym}$. The idea is to find a function (in fact, a
string homomorphism) $h^R_p$ such that for any program $P$, $P$ is
correct/coherent modulo $\set{\varphi^R_p}$ iff $h^R_p(\exec(P))$ is
correct/coherent modulo the empty axiom set. Decidability then follows
by exploiting the results of~\cite{coherence2019}. The function $h^R_p$ will
capture the properties of the axiom it is trying to eliminate, and so
it will be different for different axioms. We first outline these
function $h^R_p$, then state their property and prove the decidability
result.

For reflexivity, we transform an execution $\rho$ of $P$ to $\rho'$
where $\rho'$ is essentially $\rho$, except that whenever we see the
computation of a term, using an assignment of the form $\dblqt{x :=
  f(\vec{z})}$, we immediately insert an assume statement that states
that $R(x,x)$ holds.  More precisely, the homomorphism is defined as,
\[
h^R_\refl(a) = 
\begin{cases}
a \cdot \dblqt{\passume(R(x, x))}  & \text {if } a = \dblqt{x \passign f(\vec{z})} \\
a & \text{ otherwise }
\end{cases}
\]
The homomorphisms used for irreflexivity and symmetry follow
similar lines and are outlined 
% in~\cite{techreport}.
in~\appref{app-refl-irrefl-symm}.

\begin{theorem}
\thmlabel{refl-irrefl-symm}
For any relation symbol $R$ and $p \in \set{\refl,\irrefl,\sym}$, the
problems of coherent verification modulo $\set{\varphi^R_p}$ and
checking coherence modulo $\set{\varphi^R_p}$ are $\pspc$-complete.
\end{theorem}

\begin{comment}
\begin{proof}
The result is essentially a corollary of \lemref{rel-preservation} and
\thmref{popl19}. Since $h^R_p$ is a homomorphism, $h^R_p(\exec(P))$
for any program $P$ is regular. By \lemref{rel-preservation}, $P$ is
correct/coherent modulo $\set{\varphi^R_p}$ iff $h^R_p(\exec(P))$ is
correct/coherent modulo $\emptyset$. The upper bound follows by
exploting the automata constructed in the proof \thmref{popl19}.
\end{proof}
\end{comment}

\subsection{Transitivity}

We now consider the transitivity axiom for a relation $R$ which says
\begin{align}
\varphi^R_\trans = \forall x, y, z \cdot R(x, y) \land R(y, z)
\implies R(x, z) & \quad \text{(transitivity)}
\end{align}

The proof for decidability modulo this axiom is
different and more complex that the proofs for reflexivity,
irreflexivity, and symmetry. Intuitively, the program instrumentation
approach does not seem to work for transitivity. This is because
transitivity effects can be global.  For example, we may have that the
execution asserts the sequence of relational assumes $R(t_1, t_2),
R(t_2, t_3), \ldots R(t_{n-1}, t_n)$ (here, $t_1, \ldots t_n$ are
terms computed by the execution), where some of the intermediate terms
may have been dropped by the program (i.e., the variables holding
these terms were reassigned). 
Consequently, relating $t_1$ and (the
possibly newly constructed term) $t_n$ requires 
a principally new machinery. 
We modify the automaton construction from~\cite{coherence2019}
so that it maintains the transitive closure of the assumptions 
the program makes.
Our main observation is the following:
% Assumptions of negations of relations done by the program also
% need special care.  
% This yields the following result.
% 
%!TEX root = main.tex

\begin{figure}[t]
\centering
\scalebox{1.0}{
\begin{tikzpicture}

\node (x-1) at (0, 0) [rounded rectangle] {$x$};
\node (y-1) at (1, 1.5) [rounded rectangle] {$y$};
\node (z-1) at (2, 0) [rounded rectangle] {$z$};

\draw (x-1) edge[-{Latex[length=2mm, width=2mm]}, thick] (z-1);
\draw (x-1) edge[-{Latex[length=2mm, width=2mm]}, thick] (y-1);
\draw (y-1) edge[-{Latex[length=2mm, width=2mm]}, thick, dashed] (z-1);

\node (neg-R-x-z-1) at (1,-0.3) [rounded rectangle] {$\neg R$};
\node (R-x-y-1) at (0.2,0.9) [rounded rectangle] {$R$};
\node (neg-R-y-z-1) at (1.8,0.9) [rounded rectangle] {$\neg R$};

\node (x-2) at (5, 0) [rounded rectangle] {$x$};
\node (y-2) at (6, 1.5) [rounded rectangle] {$y$};
\node (z-2) at (7, 0) [rounded rectangle] {$z$};

\draw (x-2) edge[-{Latex[length=2mm, width=2mm]}, thick] (z-2);
\draw (x-2) edge[-{Latex[length=2mm, width=2mm]}, thick, dashed] (y-2);
\draw (y-2) edge[-{Latex[length=2mm, width=2mm]}, thick] (z-2);

\node (neg-R-x-z-1) at (6,-0.3) [rounded rectangle] {$\neg R$};
\node (R-x-y-1) at (5.2,0.9) [rounded rectangle] {$\neg R$};
\node (neg-R-y-z-1) at (6.8,0.9) [rounded rectangle] {$R$};
\end{tikzpicture}
}
\caption{Implied negative relational assumes for a transitive relation $R$.
The dashed edges (\protect\drawdirecteddash) represent the inferred relationship implied
from the relations marked by bold edges (\protect\drawdirectedline).}
\figlabel{transitivity}
\end{figure}
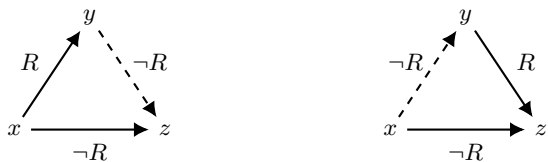
% \begin{comment}
%
\begin{theorem}
\thmlabel{transitivity-reg}
Let $\Sigma$ be a first order signature and $V$ a finite set of
program variables. Let $\Aa = \setpred{\varphi^R_\trans}{R \in
  \Rr_\trans}$ for some set of relation symbol $\Rr_\trans$ in
$\Sigma$. The following observation hold.
\begin{enumerate}
\item There is a finite automaton $\Ff_\trans$ (effectively
  constructable) of size $O(2^{\text{poly}(|V|)})$ such that for any
  coherent execution $\rho$ that is coherent modulo $\Aa$, 
  $\Ff_\trans$ accepts $\rho$ iff $\rho$ is
  feasible.
\item There is a finite automaton $\Cc_\trans$ (effectively
  constructible) of size $O(2^{\text{poly}(|V|)})$ such that $L(\Cc_\trans) =
  \cohexec(\Sigma, V, \Aa)$.
\end{enumerate}
\end{theorem}

\begin{proof}[Proof Sketch]
These are in some sense a generalization of the automata constructions
used to establish decidabiloty in~\cite{coherence2019}.% \thmref{popl19}.
The automata
$\Ff_\trans$ and $\Cc_\trans$ rely on tracking equivalence between
values stored in variables, and functional and relational
correspondences between these values. However, now since some
relations maybe transitive, additional relational correspondences (or
their absence) maybe implied for $R \in \Rr_\trans$. The basic idea is
to maintain for transitive relations $R$ (a) the transitive closure of
the positive relation assumes $\passume(R(\cdot, \cdot))$, and (b) the
negative relational assumes implied by the relational assumes seen in
an execution. More precisely, if the execution sees assumes
$\passume(R(x, y))$ and $\passume (R(y, z))$, then we also add the
constraint $R(x, z)$ in the automaton's state.  Further, if the
execution observes $\passume(R (x, y))$ and $\passume(\neg R (x, z))$,
then one can infer the constraint $\neg R(y, z)$, and in this case, we
accumulate this additional constraint in the state of the automaton.
Similarly, if the execution observes $\passume(R (y, z))$ and
$\passume(\neg R (x, z))$, then one can infer the constraint $\neg
R(x, y)$, which is added in the automaton's state. Both these
scenarios are illustrated in~\figref{transitivity}. 
% A detailed proof of this result is given in~\cite{techreport}.
A detailed proof of this result is given in~\appref{app-transitivity}.
% Like in the proof of \thmref{popl19}, a crucial property to establish
% correctness is that for transitivity axioms, the equality of two terms
% does not depend on the disequality and relational assumes seen in the
% execution. Establishing this property relies on showing that if an
% execution is feasible in a model then there is a canonical,
% \emph{minimal} model in which it is feasible. Existence of canoncial
% models maybe of independent interest. Full details can be found in
% \appref{app-transitivity}.
\end{proof}
As a consequence we have the following result.
% \end{comment}
\begin{theorem}
\thmlabel{transitivity}
For $\Aa = \setpred{\varphi^R_\trans}{R \in \Rr_\trans}$, the problems
of coherent verification modulo $\Aa$ and checking coherence modulo
$\Aa$ are $\pspc$-complete.
\end{theorem}

\begin{comment}
\begin{proof}
Based on \thmref{transitivity-reg} and the fact that the set of
executions of a program is regular, we reduce these problems to
checking the emptiness of a regular language. The size of the
resulting automata is exponential in $|V|$, and so the $\pspc$ upper
bound follows.
\end{proof}
\end{comment}

\subsection{Strict Partial Orders}

We now turn our attention to axioms that dictate that certain
relations be partial or total orders. The anti-symmetry axiom that
holds for non-strict orders introduces subtle complications. Recall
that $R$ is anti-symmetric if $\forall x, y. R(x,y) \wedge R(y,x)
\Rightarrow x= y$; this axiom can imply equality between terms if $R$
holds between a pair of terms. Concretely, if $R$ is anti-symmetric,
and the program makes assumptions in an execution that $R(t_1, t_2)$
and $R(t_2, t_1)$ hold, then any model in which such an execution is
feasible must also ensure that $t_1=t_2$. This implicit equality
assumption interferes with the notions of coherence and the automata
constructions (proofs of the results in~\cite{coherence2019} and~\thmref{transitivity}) that compute a congruence closure on
terms in a streaming fashion. 
\begin{comment}
Notice crucial to the correctness of
automata construction used in the proofs of Theorems~\ref{thm:popl19}
and~\ref{thm:transitivity-reg} was the observation that disequality
and relational assumes could be disregarded to determine if two terms
are equal. This would break in the presence of an anti-symmetric
relation.
\end{comment}

Hence, we only consider \emph{strict} partial orders in this section.
Recall that a relation $R$ is a strict partial order if it satisfies
the irreflexivity axiom and the transitivity axiom, together denoted
$\Aa^R_\spo$. We can prove decidability for problems modulo
$\Aa^R_\spo$ by using our algorithm for irreflexivity and
transitivity.

\begin{theorem}
\thmlabel{strict-partial-order}
The following problems are $\pspc$-complete.
\begin{enumerate}[topsep=0pt]
\item Given a program $P$ that is coherent modulo $\Aa^R_\spo$,
  determine if $P$ is correct.
\item Given a program $P$, determine if $P$ is coherent modulo
  $\Aa^R_\spo$
\end{enumerate}
\end{theorem}

\begin{comment}
\begin{proof}[Proof Sketch]
The proof is similar to the proof of \thmref{refl-irrefl-symm}. The
idea is to show that an execution $\rho$ is coherent/feasible modulo
$\Aa^R_\spo$ iff $h^R_\irrefl(\rho)$ is coherent/feasible modulo
$\set{\varphi^R_\trans}$. This observation does not follow from
\lemref{rel-preservation} and it requires that lemma to be
generalize. We will do that and actually prove a more general result
about combination of axioms (\thmref{combination}), and the current
theorem will be a corollary of that observation.
\end{proof}
\end{comment}

\subsection{Strict Total Orders}
\seclabel{strict-total-order} A relation $R$ is a strict total order
if it is a strict partial order and satisfies: %
 % the \emph{totality  axiom}:
\begin{align}
\forall x, y \cdot x\neq y \implies R(x,y) \lor R(y, x) & \quad \text{(totality)}
\end{align}

Strict total orders are again tricky to handle as the axiom for
totality can result in implicit equality between terms. For example,
if $\neg R(x,y)$ and $\neg R(y,x)$ then it must be the case that
$x=y$. However, if we restrict ourselves to executions that only have
assumes of the form $\passume(R(x,y))$ and do not have any assumes on
$\neg R$, i.e., of the form $\passume(\neg R(x,y))$ then there are no
implicit equalities that are entailed.

Unfortunately, in general, program executions can contain negative
assumes on $R$ (i.e., assumes of the form $\passume(\neg R(x, y))$).
In order to ensure that executions contain only \emph{positive}
assumptions on $R$, we must be careful when identifying executions of
programs with conditionals --- branches where the assumption $\neg
R(x,y)$ holds must be translated to a branch that assumes $R(y,x)$ and
a branch that assumes $x=y$. 
% We present a detailed translation in~\cite{techreport}.

 That is, we modify the following rules
defining executions of programs for branch statements; for all other
statements, the rules are the same as in~\secref{exec-semantics}.
\begin{align*}
\begin{array}{rcl}
\exec(~\passume(\neg R(x, y))~) &=& \dblqt{\passume(R(y, x))} + \dblqt{\passume(x = y)} \\
\exec(~\pif\, R(x, y) \, \pthen\, s_1\,  \pelse\, s_2~~) &=& \dblqt{\passume(R(x, y))}{\cdot}\exec(s_1) + \exec(\passume(\neg R(x, y))){\cdot}\exec(s_2) \\
\exec(~\pif\, \neg R(x, y) \, \pthen\, s_1\,  \pelse\, s_2~~) &=& \exec(\passume(\neg R(x, y))){\cdot}\exec(s_1) + \dblqt{\passume(R(x, y))}{\cdot}\exec(s_2) \\
\exec(~\pwhile\, R(x, y)\, \{s\}~~) &=&
        [\dblqt{\passume(R(x, y))} \cdot \exec(s_1)]^* \cdot \exec(\passume( \neg R(x, y))) \\
\exec(~\pwhile\, \neg R(x, y)\, \{s\}~~) &=&
        [\exec(\passume( \neg R(x, y))) \cdot \exec(s_1)]^* \cdot \dblqt{\passume(R(x, y))}
\end{array}
\end{align*}

After such a translation, executions can now have additional equality assumes even
if they did not appear in the program.  When we refer to coherent
programs, we mean that they are coherent according to the above
modified notion of executions.  This means for such programs to be
coherent, all executions must ensure that the additional equality
assumes are \emph{early}.  And when we talk about coherent
verification of programs with total orders, we mean verification for
programs that are coherent after this transformation.

\begin{comment}
We observe that when executions have only positive $R$ assumes,
checking properties modulo $\Aa^R_\sto$ is equivalent to checking
properties modulo $\Aa^R_\spo$. This will allow us to reduce the case
of strict total orders to the case of strict partial orders.

\begin{proof}
Let us prove these properties in order. Observe that if $\rho$ is
feasible in a $\Aa \cup \Aa^R_\sto$-model $\Mm$ then since $\Mm$ is
also a $\Aa \cup \Aa^R_\spo$-model, we have $\rho$ is feasible modulo
$\Aa \cup \Aa^R_\spo$. Suppose $\rho$ is feasible in a $\Aa \cup
\Aa^R_\spo$-model $\Mm$. Let $\Mm'$ have the same universe and
interpretation for all symbols except $R$. Let $S$ be any
linearization of $\sem{R}_\Mm$ and take $\sem{R}_{\Mm'}$ to be
$S$. Since $R$ is not mentioned in any sentence in $\Aa$ and $\Mm$ and
$\Mm'$ agree on all symbols except $R$, $\Mm'$ is a
$\Aa$-model. Further by definiton, $\Mm'$ satisfies
$\Aa^R_\sto$. Since $\sem{R}_\Mm \subseteq \sem{R}_{\Mm'}$, $\Mm$ and
$\Mm'$ agree on all symbols except $R$, and $\rho$ does not have any
negative $R$-assumes, all assumes in $\rho$ must hold in $\Mm'$
because they hold in $\Mm$.

The extension to prove coherence follows the proof template as the one
to prove \lemref{rel-preservation}, and exploits the model
construction outlined for feasibility above. The detailed proof is
left to the reader.
\end{proof}
\end{comment}

We observe that in the absence of any assumes of the form $\neg R(x, y)$
the verification problem modulo strict total orders reduces that
modulo strict partial orders, giving us the following
($\Aa^R_\sto$ denote the axioms of irreflexivity, transitivty and
totality for the relation $R$).
% \vspace{-0.1in}
\begin{theorem}
\thmlabel{total}
The problems of coherent verification, and checking coherence modulo
$\Aa^R_\sto$ are $\pspc$-complete.
\end{theorem}

\begin{comment}
\begin{proof}
Based on \lemref{total-to-partial}, we can reduce checking these
properties modulo $\Aa^R_\spo$. Then the result follows from
\thmref{strict-partial-order}.
\end{proof}
\end{comment}

% Now, we observe that, for an execution $\sigma$, if there is model
% $\Mm$ in which $R$ is a strict partial order and $\sigma$ is feasible
% in $\Mm$, then one can construct a model $\Mm'$ where $R$ is a strict
% total order and $\sigma$ is feasible in $\Mm'$.  This observation
% again relies on the fact that $\sigma$ does not have assumes of the
% form $\passume(\neg R(x, y) )$.  We give detailed proof
% in~\appref{app-strict-total-order} We hence have the following result.

%!TEX root = main.tex

\section{Axioms Over Functions}
\seclabel{functions}

We now discuss computational problems modulo axioms that involve
function symbols. The treatment of axioms involving functions in the
verification of coherent programs is inherently hard. This is because,
like in the case of (nonstrict) partial orders and strict total
orders, the axioms along with the $\passume$-steps in the execution,
can imply equalities between terms beyond those entailed by just the
$\passume$ steps in the execution. For example, consider the axiom
$\forall x, y \cdot f(x, y) = f(y, x)$ constraining $f$ to be a
commutative function. Then terms like $f(f(x,y),z)$ are equal to terms
like $f(z, f(x, y))$, and hence when building models we must make sure
that functions/relations on such terms are defined in the same
way. Terms made equivalent by the functional axioms can be
syntactically very different, and keeping track of the equivalence on
unbounded executions is hard using finite memory. We consider many
natural classes of axioms, and proving both positive and negative
results that help delineate the decidability/undecidability boundary.

\subsection{Associativity}
\seclabel{associativity}

We now consider the associativity axiom for a function $f$.
\begin{align}
\varphi^f_\assoc \triangleq & \, \forall x, y, z \cdot f(x, f(y, z)) = f(f(x, y), z) 
& \quad \text{(associativity)}
\end{align}
We show, surprisingly to us, that coherent verification is undecidable
modulo $\set{\varphi^f_\assoc}$, i.e., even when we have only one
axiom that requires only one function to be associative.  
In fact, the situation is a lot worse --- checking the 
feasibility of even a \emph{single} (even coherent) execution 
is undecidable, in the presence of a single associative function.
The proof of the following result uses a reduction from the
word problem for finitely generated semigroups~\cite{post1947}.

\begin{theorem}
\thmlabel{assoc-undec}
Given a a trace $\rho$ that is coherent modulo
$\set{\varphi^f_\assoc}$, it is undecidable to determine if $\rho$ is
feasible. Therefore, the problem checking if a program $P$ that is coherent modulo $\set{\varphi^f_\assoc}$ is undecidable.
\end{theorem}

\begin{comment}
\input{fig-undec-assoc-wordproblem}

\begin{proof}
We show the following reduction. Given an instance $I =
(\Delta,\circ,u_1,v_1,\ldots u_n,v_n,u_0,v_0)$ there is an execution
$\rho$ that is coherent modulo $\set{\varphi^f_\assoc}$ such that $I$
is a YES instance of the work problem iff $\rho$ is infeasible modulo
$\set{\varphi^f_\assoc}$.

The constructed execution $\rho$ is shown in
\figref{undec-assoc-wordproblem}. The signature $\Sigma$ consists of a
binary function $\cd{f}$ which obeys the associativity axiom
$\varphi^{\cd{f}}_\assoc$.  The set of variables in the program are
\[\Vv = \underbrace{\set{\cd{a}_1 \ldots \cd{a}_k}}_{
\scriptsize 
\begin{array}{c}
\text{these are}\\
\text{letters from }\Delta 
\end{array}
} \cup \bigcup\limits_{i = 0}^N \set{\cd{x}_i, \cd{x}_{i, 2}, \ldots \cd{x}_{i, {|u_i|}}} \cup \bigcup\limits_{i = 1}^N \set{\cd{y}_i, \cd{y}_{i, 2}, \ldots \cd{y}_{i, |v_i|}}
\]
The post-condition $\phi$ is $\cd{x}_0 = \cd{y}_0$.

Our reduction uses the associative function $\cd{f}$ to model
concatenation. A word $w = a_1, \ldots, a_m$ is modeled as the term
$t_w = \cd{f}(a_1, \cd{f}(a_2, \ldots, \cd{f}(a_{m-1}, a_m) \ldots
))$. Intuitively, the execution first creates the words $u_1$, $v_1$,
$u_2 \ldots u_N, v_N$ and assumes $u_1 = v_1$, $u_2 = v_2, \ldots, u_N
= v_N$.  It then creates the words $u_0, v_0$ and checks $u_0 = v_0$
in the postcondition. Proof that $\rho$ is coherent and the reduction
is correct is straightforward, but for completeness, the proof can be
found in \appref{assoc-undec}.
\end{proof}
\end{comment}

\subsection{Commutativity}

We now consider the commutativity axiom, which is the following
\begin{align}
\varphi^f_\comm \triangleq \, & \forall x, y \cdot f(x, y) = f(y, x) 
& \quad \text{(commutativity)}
\end{align}
% In this section we show that such axioms can be effectively handled.
% 
% We use the technique of program instrumentation to handle this
% axiom.  
We augment executions with an auxiliary variable $v^* \not\in
V$ and transform executions using the following homomorphism
that uses the auxiliary variable $v^*$
\[
h^f_\comm(a) = 
\begin{cases}
a \cdot \dblqt{v^*\passign f(y, x)} \cdot \dblqt{\passume(z = v^*)}  & \text {if } a = \dblqt{z \passign f(x, y)}\\
a & \text{ otherwise }
\end{cases}
\]

We show that the above transformation preserves feasibility and
coherence, giving us the following result.
\begin{theorem}
\thmlabel{commutativity}
Verification of coherent programs and checking coherence modulo
commutativity axioms is decidable and is $\pspc-$complete.
\end{theorem}

\begin{comment}
\begin{proof}
The proof uses \lemref{comm-homomorphism} and is similar to
\thmref{refl-irrefl-symm}.
\end{proof}
\end{comment}

% \ucomment{Depth-1 axioms?}

\subsection{Idempotence}

Next we consider the idempotence axiom for a unary function $f$:
% Let us now consider the idempotence axiom. 
% A unary function $f$ is said to be idempotent if
\begin{align}
\varphi^f_\idem \triangleq \, & \forall x \cdot f(x) = f(f(x)) & \quad \text{(idempotence)}
\end{align}

\noindent
Again, we show that there is a simple homomorphism 
$h^f_\idem$ that preserves coherence and feasibility 
% (see~\cite{techreport}) 
(see~\appref{app-comm-assoc})
and reduces the verification question to that modulo $\Aa = \emptyset$,
giving:

% Here again, we use program instrumentation.  We use an auxiliary
% variable $v^* \not\in V$ and use the following homomorphism.
% \[
% h^f_\idem(a) = 
% \begin{cases}
% a \cdot \dblqt{v^*\passign f(y)} \cdot \dblqt{\passume(y = v^*)}  & \text {if } a = \dblqt{y \passign f(x)}\\
% a & \text{ otherwise }
% \end{cases}
% \]

% Here again we have the following decidability result.

\begin{comment}
Again, we have a theorem about the preservation feasibility and coherence.
\begin{lemma}
\lemlabel{idem-homomorphism}
%                                                                       
For any execution $\rho$, the following properties hold.
\begin{enumerate}
\item $\rho$ is feasible modulo $\set{\varphi^f_\idem}$ iff
  $h^f_\idem(\rho)$ is feasible modulo $\emptyset$.
\item $\rho$ is coherent modulo $\set{\varphi^f_\idem}$ iff
  $h^f_\idem(\rho)$ is coherent modulo $\emptyset$.
\end{enumerate}
\end{lemma}

The proof of \lemref{idem-homomorphism} is very similar to the proof
of \lemref{comm-homomorphism}. The modified data model constructed is
identical. The proof is therefore skipped. As before, the preservation
theorem allows us to conclude the decidability of verification.
\end{comment}

\begin{theorem}
\thmlabel{idempotence} Verification of coherent programs and checking
coherence modulo idempotence axioms is $\pspc$-complete.
\end{theorem}

%!TEX root = main.tex

\section{Combining Axioms}
\seclabel{combine}

We have thus far proved decidability results when a relation or functions satisfies certain properties like reflexivity/irreflexivity/symmetry/transitivity or commutativity/idempotence. We now show that all of these results can be combined. That is, we can consider a signature where relations and functions are assumed to satisfy some subset of these properties, and with some being uninterpreted, and the verification problem will remain decidable for coherent programs.

\begin{theorem}
\thmlabel{combination}
Let $\Aa$ be a set of axioms where each relation symbol $R$ is either a total order or satisfies some (possibly empty) subset of properties out of reflexivity, irreflexivity, symmetry, transitivity, and each function symbol $f$ satisfies some (possibly empty) subset out of commutativity and idempotence. The verification problem for coherent programs modulo $\Aa$ is $\pspc$-complete.
\end{theorem}

The proof of the above result proceeds by \emph{eliminating} axioms one at a time.
We first eliminate the relational axioms (reflexivity, irreflexivity, symmetry) in $\Aa$ using program instrumentation. We then eliminate the functional axioms in $\Aa$, again using program instrumentation. 
Our proof relies on this order of elimination of axioms.
At this point, the only axioms remaining are those corresponding to transitivity of a subset of relational symbols, which is handled using the automata construction discussed in the proof of~\thmref{transitivity-reg}.

\begin{comment}
A proof sketch of this theorem in provided in~\appref{app-comb}.
A simple consequence of \thmref{combination} is that verification for coherent programs is decidable even when some of the relations are constrained to be equivalence relations or pre-orders.
\end{comment}

%\input{check-coherence}

%!TEX root = main.tex

\section{Related Work}
\seclabel{related}

The theory of equality with uninterpreted functions (EUF) is a widely used
theory in many verification applications as it has decidable quantifier free fragment.
EUF has been central to advances in verification of microprocessor control~\cite{Burch1994,bls02}
and hardware verification~\cite{als08,hikb96} and property directed model checking~\cite{Hmb17}.
EUF has been used as a popular abstraction in software verification~\cite{bh07structural,bh08}.
Uninterpreted functions have also been studied for equivalence checking
and translation validation~\cite{ps06}.
% Recently, 
Bueno et al~\cite{euforia2019} demonstrated the effectiveness
of uninterpreted programs for verifying SVCOMP benchmarks against control flow properties.

Mathur et al~\cite{coherence2019} introduced the class of coherent 
uninterpreted programs and showed that verification of coherent programs,
with or without recursive function calls, is a decidable problem. 
This is one of the few subclasses of program verification over infinite domains that is known to be decidable.
Previous works~\cite{godoy09,Gulwani2007,Muller2005herbrand} have established decidability of verification of classes of uninterpreted programs 
with heavy syntactic restrictions such as disallowing conditionals inside loops 
or nested loops, etc.
As noted in~\cite{coherence2019}, the notion of
coherence is close to the notion of a bounded 
pathwidth decomposition~\cite{robertson1983graph}.
A term that is created in a coherent execution stays within
some program variable (modulo congruence) until the first time
all variables containing that term are over-written,
and after this point, the execution never computes it again,
and thus, the set of windows that contain a term
form a contiguous segment of the program execution.
Path decomposition and the related notion
of tree decomposition have been exploited
many times in the literature to give decidability in verification~\cite{madhu2011,chatterjee2016,chatterjee2015}.

The work in~\cite{memsafeMMKMV19} extends the work of~\cite{coherence2019}
to \emph{updatable maps} and identifies extensions of coherence
that make verification decidable. It utilizes this to provide implementation
of verification algorithms for memory safety for a class of heap manipulating programs, including traversal algorithms on data structures such as singly linked list, sorted lists, binary search trees etc.
Combining the results of this paper with these results is an interesting
future direction.
%While the reasoning for establishing memory safety can be often done by treating native data as purely uninterpreted, more complex reasoning would inadvertently be required for 
%assertion checking, often requiring specific 
%axioms to establish correctness. For example, a list search routine working over a sorted list
%(see~\secref{illustrative-ex}) would require reasoning with the total order axioms.

%In this work we explore the verification problem for coherent programs
%subject to underlying theories constraining interpretations of function and
%relation symbols in the program vocabulary.
%We show that coherent program verification becomes undecidable in the
%presence of general EPR~\cite{Ramsey1987} axioms, but remain within decidability boundaries
%for order axioms including reflexivity, %irreflexivity, symmetry,
%transitivity and total order.
The class of EPR formulas that consist of universally quantified formulas over relational signatures is a well-known decidable class of first-order logic~\cite{Ramsey1987}.
EPR-based reasoning has been proved powerful for verification of large-scale systems~\cite{paxosepr2017,McMillan2016,Taube2018} and the 
Ivy~\cite{ivy2016,ivy2018} system is one of the most
notable framework that exploits EPR based reasoning for verifying program snippets without recursion.
EPR encoding of order axioms such as reflexivity, 
symmetry, transitivity
and total orders has been used in proving programs  working
over heaps~\cite{Itzhaky2014}.

The work in Kleene Algebra with Tests (KAT)~\cite{Kozen1997} 
considers problems involving unbounded recursion and choice with
abstractions of data, similar to our work. However, while we treat
congruence axioms for equality faithfully in our work, it is unclear
to us how to express these in KAT or its extensions~\cite{kozen1996,kozen2014,doumane2019}. Furthermore, 
the restrictions of coherence studied in~\cite{coherence2019} and the
work here that are based on bounded path-width notions 
seem very different from studies of decidable problems in KAT.
A study of whether our results can be adapted to yield decidable
fragments for KAT is an interesting future direction.

% reduces program verification to problems automata-theoretic questions.
% KAT style reasoning, however, is propositional and
% abstracts away data, and the power of reasoning with
% equality and congruence. Extensions to KAT~\cite{kozen1996,kozen2014,doumane2019}
% augment KAT style reasoning with domain specific knowledge,
% similar to how we extend the work in~\cite{coherence2019}.
% However, it is unclear if a propositional encoding of equality on data
% is possible in this framework, and whether it lends itself to decidability.

A notable verification technique with an automata-theoretic
foundation and that has been very effective in practice
is that of  trace abstraction due to 
Heizmann et al~\cite{Heizmann2009SAS,Heizmann2010,Matthias2013,Farzan2013,Farzan2014,Farzan2015}. In this technique, one constructs  \emph{iteratively} regular sets that (incompletely) capture the
set of all infeasible executions, eventually striving to cover
all failing executions of a program, but handling complex theories such
as arithmetic. In contrast, our work here builds complete automata
in one stroke that accept all infeasible traces over a vocabulary, but
handles only simple theories with restricted sets of axioms, but yielding decidability. Combining these two lines of work for more efficient software verification is an interesting future direction.
%, and the correctness of a program can then eventually be asserted using emptiness question for regular languages.

%!TEX root = main.tex

\section{Conclusions}
\seclabel{conclusions}

By incorporating axioms on functions and relations, decidability results in this paper, enable a more faithfully automatic verification of programs. It is worth noting that the upper bound for all our decidability results is $\pspc$, which is the same as that for Boolean programs. Thus, though we consider programs over infinite domains with additional structure, our verification results have the same complexity as that for programs over Boolean domains.

One future direction is to adapt this technique for practical program verification. In this context, adapting our technique within the automata-theoretic technique of~\cite{Heizmann2009SAS,Matthias2013,Heizmann2010,Farzan2015,Farzan2013} seems most promising. Second, there are several program verification techniques that use EPR, and in several of these, EPR is used mainly to establish a linear order on the universe~\cite{Itzhaky2014}. Automatically verifying such programs using our technique is worth exploring. 

%On the theoretical front, there are several theories that are useful in program verification that could be explored for decidable coherent verification, such as \emph{updatable maps and arrays} (where there are functions that can be updated point-wise, where the domain can be arbitrary or have a linear order) and updatable sets. We believe that such theories would be useful in reasoning with heap-manipulating programs where updatable maps can be used to model pointers in the heap and sets to model heaplets for local reasoning~\cite{Reynolds2002,natproofs2014}. 

\bibliographystyle{splncs04}
\bibliography{references}

\clearpage

\appendix

%!TEX root = main.tex

\section{Handling Relations in Streaming Congruence Closure}
\applabel{automata-relations}

The work in~\cite{coherence2019} omit relations 
and model them as functions.
Specifically, all programs are assumed to have two 
fixed variables $\cd{T}$ and $\cd{F}$ (corresponding to
Boolean constants \texttt{true} and \texttt{false})
that are never re-assigned. In the beginning of each program,
there is an assume $\passume(\cd{T} \neq \cd{F})$
Further, for every $k$-ary relation $R$, there is a 
function $f_R$ and a variable $b_R$.
Every assumption of the form $\dblqt{\passume(R(\vec{\cd{z}}))}$ 
is translated to the sequence 
$\dblqt{b_R \passign f_R(\vec{\cd{z}})} \cdot \dblqt{\passume(b_R = T)}$,
and every assumption of the form $\dblqt{\passume(\neg R(\vec{\cd{z}}))}$ 
is translated to the sequence 
$\dblqt{b_R \passign f_R(\vec{\cd{z}})} \cdot \dblqt{\passume(b_R = F)}$.

This approach adds additional program variables and 
function symbols and further restricts the class of programs 
that are coherent because the memoizing restriction also 
applies to the newly introduced function symbols.
In this paper, we show how to handle relations as first class symbols
without modeling them using function symbols.
For this, we will construct an automaton 
(similar to that in~\cite{coherence2019})
that accepts coherent executions 
(modulo the empty set of axioms $\emptyset$) iff 
they are feasible (modulo $\emptyset$).

Recall that executions are words over the alphabet 
$\Pi = \setpred{ \dblqt{x \passign y}, \dblqt{x \passign f(\vec{z})},  
\dblqt{\passume (x=y)}, \dblqt{\passume (x\neq y)},
\dblqt{\passume (R(\vec{z}))},
\dblqt{\passume (\neg R(\vec{z}))}
}{
x, y, \vec{z} \textit{~are~in~} V
}$.

% Analogous to the sets of equality and disequality assumes 
% ($\alpha(\rho)$ and $\beta(\rho)$ in~\cite{coherence2019} respectively), 
% we can define the set of positive and negative relation assumes as follows:
% \begin{align*}
% \gamma(\varepsilon)(R) &= \emptyset \\
% \gamma(\sigma \cdot a)(R) &= 
% \begin{cases}
% \gamma(\sigma) \cup \set{\tuple{\comp(\sigma, z_1), \ldots ,\comp(\sigma, z_r)}} \\
% \quad \quad \quad \quad \text{ if } a = \dblqt{\passume (R(z_1, \ldots, z_r))} \\
% \gamma(\sigma) \quad \quad \;  \text{otherwise}
% \end{cases}\\
% \delta(\varepsilon)(R) &= \emptyset \\
% \delta(\sigma \cdot a)(R) &= 
% \begin{cases}
% \delta(\sigma) \cup \set{\tuple{\comp(\sigma, z_1), \ldots ,\comp(\sigma, z_r)}} \\
% \quad \quad \quad \quad \text{ if } a = \dblqt{\passume (\neg R(z_1, \ldots, z_r))} \\
% \delta(\sigma) \quad \quad \; \text{otherwise}
% \end{cases}
% \end{align*}

Let us denote by $A_\scc$ our automaton for streaming congruence closure.
The states $Q_\scc$ are either the special reject state $\reject$
or tuples of the form $(\equiv, d, P, \rel^+, \rel^-)$, where
\begin{itemize}
	\item $\equiv$ is an equivalence relation over $V$,
	\item $d$ is a symmetric and irreflexive binary relation 
	over $V/\equiv$ (equivalence classes of $\equiv$),
	\item $P$ is such that for every $k$-ary function $f \in \Sigma$,
	$P(f)$ is a partial mapping from $(V/\equiv)^k \to V/\equiv$, and
	\item $\rel^+$ and $\rel^-$ are such that for every $k$-ary relation
	$R$, $\rel^+(R)$ and $\rel^-(R)$ are sets of k-tuples of
	$V/\equiv$ such that $\rel^+(R) \cap \rel^-(R) = \emptyset$.
\end{itemize}

Notice that the first three components of the state are similar
to~\cite{coherence2019}.
The later two components intuitively accumulate the relational
assumes (corresponding to $\gamma(\cdot)$ and $\delta(\cdot)$).

The transition relation $\delta_\scc$ of the automaton is defined
as follows.
Let $q = (\equiv, d, P, \rel^+, \rel^-)$.
If $q = \reject$, then $\delta_\scc(q, a) = \reject$ for every $a \in \Pi$
(i.e., $\reject$ is an absorbing state).
Otherwise,  we define the state $q' = \delta_\scc(q, a)$ as the 
tuple $(\equiv', d', P', \rel'^{+}, \rel'^-)$ below.
In each of these cases, if $d'$ becomes irreflexive or
there is a relation $R$ such that $\rel'^+(R) \cap \rel'^-(R) \neq \emptyset$,
then we set $q'$ to be $\reject$.
% \ucomment{Come back here.}

\begin{description}
\item[$a = \dblqt{x \passign y}$].\\
Here, if $y = x$, $q' = q$.
Otherwise, the variable $x$ gets updated to be in 
the equivalence class of $y$, and $d', P', \rel'^+$ and
  $\rel'^-$ are updated accordingly:
\begin{itemize}
\item $\equiv' =  
  \longproj{\equiv}{V\setminus\set{x}}
  \cup \setpred{(x,y'),(y',x)}{y'\equiv
  y} \cup \set{(x, x)}$.
\item $d' = \setpred{(\eqcl{x_1}{\equiv'},\eqcl{x_2}{\equiv'})}{x_1,
  x_2 \in V\setminus \set{x}, (\eqcl{x_1}{\equiv},\eqcl{x_2}{\equiv})
  \in d }$ %\ucomment{$\cup \setpred{(\eqcl{x}{\equiv'}, \eqcl{z}{\equiv'}), (\eqcl{z}{\equiv'}, \eqcl{x}{\equiv'})}{ \text{if y is x and } (\eqcl{x}{\equiv},\eqcl{z}{\equiv}) \in d}$}
\item $P'$ is such that for every $r$-ary function $h$,
\[
P'(h)(\eqcl{x_1}{\equiv'},\ldots \eqcl{x_r}{\equiv'}) =
   \begin{cases}
   \eqcl{u}{\equiv'} 
   			% & \mbox{there is a } v \neq x \mbox{ such that }\\ 
      %        &\eqcl{v}{\equiv} = P(f)(\eqcl{\gamma(x_1)}{\equiv},\ldots \eqcl{\gamma(x_1)}{\equiv})\\
                   % & where, \gamma(z) = y \text{ if } z = x \text{, otherwise } \gamma(z) = z\\
                & x \not\in \set{u, x_1,\ldots x_r} \text{ and }\\
   & \eqcl{u}{\equiv} = P(h)(\eqcl{x_1}{\equiv},\ldots \eqcl{x_r}{\equiv})\\
   \undf & \mbox{otherwise}
   \end{cases}
\]
\item $\rel'^+$ is such that for every $k$-ary relation $R$,
\[
\rel'^+(R) = \setpred{
(\eqcl{x_1}{\equiv'},\ldots \eqcl{x_k}{\equiv'})
}{
  x_1, x_2, \ldots, x_k \in V\setminus \set{x},
  (\eqcl{x_1}{\equiv},\ldots \eqcl{x_k}{\equiv} \in \rel^+(R))
} 
\]
\item $\rel'^-$ is such that for every $k$-ary relation $R$,
\[
\rel'^-(R) = \setpred{
(\eqcl{x_1}{\equiv'},\ldots \eqcl{x_k}{\equiv'})
}{
  x_1, x_2, \ldots, x_k \in V\setminus \set{x},
  (\eqcl{x_1}{\equiv},\ldots \eqcl{x_k}{\equiv} \in \rel^-(R))
} 
\]
\end{itemize}

\item[$a = \dblqt{x \passign f(z_1, \ldots z_k)}$].\\
 There are two cases  to consider.
\begin{enumerate}
\item \textbf{Case $P(f)(\eqcl{z_1}{\equiv},\ldots \eqcl{z_k}{\equiv})$
  is defined}.\\
   Let $P(f)(\eqcl{z_1}{\equiv},\ldots \eqcl{z_k}{\equiv})
  = \eqcl{v}{\equiv}$. 
  This case is similar to the case when $a$ is
  $\dblqt{x \passign y}$.  
  That is, when $x \in \eqcl{v}{\equiv}$, then $\equiv' \, = \, \equiv$, $d' = d$ and $P' = P$.
  Otherwise, we have
\begin{itemize}
\item $\equiv' = \longproj{\equiv}{V\setminus\set{x}} \cup \setpred{(x,v'),(v',x)}{v'\equiv v} \cup \set{(x, x)}$
\item $d' = \setpred{(\eqcl{x_1}{\equiv'},\eqcl{x_2}{\equiv'})}{x_1,
  x_2 \in V\setminus \set{x}, (\eqcl{x_1}{\equiv},\eqcl{x_2}{\equiv})
  \in d }$
\item $P'$ is such that for every $r$-ary function $h$,
\[
P'(h)(\eqcl{x_1}{\equiv'},\ldots \eqcl{x_r}{\equiv'}) =
   \begin{cases}
   \eqcl{u}{\equiv'} 
      % & \mbox{there is a } v \neq x \mbox{ such that }\\
      %        & \eqcl{v}{\equiv} = P(f)(\eqcl{\gamma(x_1)}{\equiv},\ldots \eqcl{\gamma(x_k)}{\equiv})\\
      %        & where, \gamma(z) = y \text{ if } z = x \text{, otherwise } \gamma(z) = z\\
   
   & x \not\in \set{u, x_1,\ldots x_r} \text{ and }\\
   & \eqcl{u}{\equiv} = P(h)(\eqcl{x_1}{\equiv},\ldots \eqcl{x_r}{\equiv})\\
   \undf & \mbox{otherwise}
   \end{cases}
\]
\item $\rel'^+$ is such that for every $k$-ary relation $R$,
\[
\rel'^+(R) = \setpred{
(\eqcl{x_1}{\equiv'},\ldots \eqcl{x_k}{\equiv'})
}{
  x_1, x_2, \ldots, x_k \in V\setminus \set{x},
  (\eqcl{x_1}{\equiv},\ldots \eqcl{x_k}{\equiv} \in \rel^+(R))
} 
\]
\item $\rel'^-$ is such that for every $k$-ary relation $R$,
\[
\rel'^-(R) = \setpred{
(\eqcl{x_1}{\equiv'},\ldots \eqcl{x_k}{\equiv'})
}{
  x_1, x_2, \ldots, x_k \in V\setminus \set{x},
  (\eqcl{x_1}{\equiv},\ldots \eqcl{x_k}{\equiv} \in \rel^-(R))
} 
\]
\end{itemize}

\item \textbf{Case $P(f)(\eqcl{z_1}{\equiv},\ldots \eqcl{z_k}{\equiv}$
  is undefined.} \\
  In this case, we remove $x$ from its older
  equivalence class and make a new class that only contains the
  variable $x$.  We update $P$ to $P'$ so that the function $f$ maps
  the tuple $(\eqcl{z_1}{\equiv'}, \ldots, \eqcl{z_k}{\equiv'})$ (if
  each of them is a valid/non-empty equivalence class) to the class
  $\eqcl{x}{\equiv'}$.  The set $d'$ follows easily from the new
  $\equiv'$ and the older set $d$.  Thus,
\begin{itemize}
\item $\equiv' = \longproj{\equiv}{V\setminus\set{x}} \cup
  \set{(x, x)}$
\item $d' = \setpred{(\eqcl{x_1}{\equiv'},\eqcl{x_2}{\equiv'})}{x_1,
  x_2 \in V\setminus \set{x}, (\eqcl{x_1}{\equiv},\eqcl{x_2}{\equiv})
  \in d }$
\item $P'$ behaves similar to $P$ for every function different from $f$.
\begin{itemize}
\item For every $r$-ary function $h \neq f$,
\[
P'(h)(\eqcl{x_1}{\equiv'}, \ldots, \eqcl{x_r}{\equiv'}) = 
   \begin{cases}
   \eqcl{u}{\equiv'} & \mbox{if } x \not\in \set{u, x_1,\ldots x_k}
        \mbox{ and } \\
        & \eqcl{u}{\equiv} = P(h)(\eqcl{x_1}{\equiv},\ldots \eqcl{x_r}{\equiv})\\
   \undf & \mbox{otherwise}
   \end{cases}
\]
\item For the function $f$, we have the following.
\[
P'(f)(\eqcl{x_1}{\equiv'}, \ldots, \eqcl{x_k}{\equiv'}) =
   \begin{cases}
   \eqcl{x}{\equiv'} & \mbox{if } x_i = z_i\ \forall i \text{ and } x \not\in \set{x_1,\ldots x_k}\\
   \eqcl{u}{\equiv'} & \mbox{if } x \not\in \set{u, x_1,\ldots x_k}
        \mbox{ and } \\
        & \eqcl{u}{\equiv} = P(f)(\eqcl{x_1}{\equiv},\ldots \eqcl{x_
k}{\equiv})\\
   \undf & \mbox{otherwise}
   \end{cases}
\]
\item $\rel'^+$ is such that for every $k$-ary relation $R$,
\[
\rel'^+(R) = \setpred{
(\eqcl{x_1}{\equiv'},\ldots \eqcl{x_k}{\equiv'})
}{
  x_1, x_2, \ldots, x_k \in V\setminus \set{x},
  (\eqcl{x_1}{\equiv},\ldots \eqcl{x_k}{\equiv} \in \rel^+(R))
} 
\]
\item $\rel'^-$ is such that for every $k$-ary relation $R$,
\[
\rel'^-(R) = \setpred{
(\eqcl{x_1}{\equiv'},\ldots \eqcl{x_k}{\equiv'})
}{
  x_1, x_2, \ldots, x_k \in V\setminus \set{x},
  (\eqcl{x_1}{\equiv},\ldots \eqcl{x_k}{\equiv} \in \rel^-(R))
} 
\]
\end{itemize}            
\end{itemize}
\end{enumerate}

\item[$a = \dblqt{\passume(x = y)}$].\\
Here, we essentially merge the
  equivalence classes in which $x$ and $y$ belong and perform the ``local
  congruence closure''. In
  addition, $d'$ and $P'$ are also updated as in~\cite{coherence2019}.
\begin{itemize}
\item $\equiv'$ is the smallest equivalence relation on $V$ such that
  (a) $\equiv \cup \set{(x,y)} \subseteq \equiv'$, and (b) for every
  $k$-ary function symbol $f$ and variables $x_1,x_1',x_2,x_2',\ldots
  x_k,x_k',z,z' \in V$ such that $\eqcl{z}{\equiv} = P(f)(\eqcl{x_1}{\equiv},
  \ldots \eqcl{x_k}{\equiv})$, $\eqcl{z'}{\equiv} =
  P(f)(\eqcl{x_1'}{\equiv},\ldots \eqcl{x_k'}{\equiv})$, and
  $(x_i,x_i') \in \equiv'$ for each $i$, we have $(z,z') \in
  \equiv'$.
\item $d' =
  \setpred{(\eqcl{x_1}{\equiv'},\eqcl{x_2}{\equiv'})}{(\eqcl{x_1}{\equiv},\eqcl{x_2}{\equiv})
    \in d }$
\item $P'$ is such that for every $r$-ary function $h$,
\[
P'(h)(\eqcl{x_1}{\equiv'}, \ldots \eqcl{x_r}{\equiv'}) = 
   \begin{cases}
   \eqcl{u}{\equiv'} & \mbox{if } \eqcl{u}{\equiv} = P(h)(\eqcl{x_1}{\equiv},\ldots \eqcl{x_r}{\equiv}) \\
   \undf & \mbox{otherwise}
   \end{cases} 
\]
\item $\rel'^+$ is such that for every $k$-ary relation $R$,
\[
\rel'^+(R) = \setpred{
(\eqcl{x_1}{\equiv'},\ldots \eqcl{x_k}{\equiv'})
}{
  (\eqcl{x_1}{\equiv},\ldots \eqcl{x_k}{\equiv} \in \rel^+(R))
} 
\]
\item $\rel'^-$ is such that for every $k$-ary relation $R$,
\[
\rel'^-(R) = \setpred{
(\eqcl{x_1}{\equiv'},\ldots \eqcl{x_k}{\equiv'})
}{
  (\eqcl{x_1}{\equiv},\ldots \eqcl{x_k}{\equiv} \in \rel^-(R))
} 
\]
\end{itemize}

\item[$a = \dblqt{\passume(x \neq y)}$].\\
In this case, $\equiv' = \equiv$, $P' = P$, 
$\rel'^+ = \rel^+$ and $\rel'^- = \rel^-$. Further,
\[
d' = d \cup \set{(\eqcl{x}{\equiv'},\eqcl{y}{\equiv'}), (\eqcl{y}{\equiv'},\eqcl{x}{\equiv'})}
\]

\item[$a = \dblqt{\passume(R(x_1, x_2, \ldots, x_k))}$].\\
In this case, $\equiv' = \equiv$, $P' = P$, $d' = d$ 
and $\rel'^- = \rel^-$. Further,
\[
\rel'^+(R') = 
\begin{cases}
\rel^+(R) \cup \set{(\eqcl{x_1}{\equiv'},\eqcl{x_2}{\equiv'}, \ldots, \eqcl{x_k}{\equiv'})} & \text{ if } R' = R \\
\rel^+(R') & \text{ otherwise }
\end{cases}
\]
\item[$a = \dblqt{\passume(\neg R(x_1, x_2, \ldots, x_k))}$].\\
In this case, $\equiv' = \equiv$, $P' = P$, $d' = d$ 
and $\rel'^+ = \rel^+$. Further,
\[
\rel'^-(R') = 
\begin{cases}
\rel^-(R) \cup \set{(\eqcl{x_1}{\equiv'},\eqcl{x_2}{\equiv'}, \ldots, \eqcl{x_k}{\equiv'})} & \text{ if } R' = R \\
\rel^-(R') & \text{ otherwise }
\end{cases}
\]
\end{description}

We next give a proof gist for the correctness of the automaton construction.
The bulk of the proof is the same as that given in~\cite{coherence2019}.
Here, we only discuss the details necessary to prove the correctness that relates
to the relational assumes.

Let us define the notion of a \emph{minimal model}.
Intuitively, this model has the same algebraic structure 
(same interpretations for constants and functions) as the 
\emph{initial} term model as defined in~\cite{coherence2019}.
Further, we also add relations in the minimal model on top of the initial term model.
For a set of ground equalities $A$, we will denote by
$\Mm^\textsf{initial}_A = (U^\textsf{initial}_A, \sem{}^\textsf{initial}_A)$ the initial term model given by
the congruence induced by $A$.

\begin{definition}
\deflabel{minimal-model}
Let $\Gamma = \Gamma_\textsf{equalities} \cup \Gamma_\textsf{relations}$ 
be a set of atomic formulae of the form
$(t_1 = t_2) \in \Gamma_\textsf{equalities}$ 
or $R(t_1, \ldots, t_k) \in \Gamma_\textsf{relations}$
where $t_1, \ldots, t_k$ are ground terms over our vocabulary $\Sigma$
and $R$ is a $k$-ary relation in our vocabulary $\Sigma$.
The minimal model $\Mm^\minimal_\Gamma = (U^\minimal, \sem{}^\minimal)$ of $\Gamma$
is defined as follows.
\begin{itemize}
  \item $U^\minimal = U^\textsf{initial}_{\Gamma_\textsf{equalities}}$,
  \item $\sem{c}^\minimal = \sem{c}^\textsf{initial}_{\Gamma_\textsf{equalities}}$ for $c \in \Cc$,
  \item $\sem{f}^\minimal = \sem{f}^\textsf{initial}_{\Gamma_\textsf{equalities}}$ for $f \in \Ff$, and
  \item $\sem{R}^\minimal = \setpred{(\sem{t_1}^\minimal, \ldots, \sem{t_k}^\minimal)}{R(t_1, \ldots, t_k) \in \Gamma_\textsf{relations}}$.
\end{itemize}
\end{definition}

For an execution $\rho$, the minimal model of $\rho$
is defined by the minimal model for the set of equality and 
positive relational atoms in  $\kappa(\rho)$ (i.e., we do not
include the dis-equality and the negative relational 
assumes accumulated by $\rho$) to define the minimal model.
We will use $\Mm_\rho = (U_\rho, \sem{}_\rho)$ to denote this minimal model.

Notice that an execution $\rho$ only defines a relation on the set
of computed terms and thus, the minimal model never relates elements
outside of the set of computed terms using relation symbols.
This is formalized below.
\begin{lemma}
Let $\rho$ be an execution and let $\Mm_\rho$ be the minimal 
model of $\rho$.
Let $(e_1, \ldots e_k) \in (U_\rho)^k$  be a tuple of
elements in the minimal model such that one of $e_1, \ldots, e_k$
is not computed by the execution (i.e., there is an $1\leq i\leq k$
such that for every $t \in \Terms(\rho)$, $\sem{t}_\rho \neq e_i$).
Then, we have $(e_1, \ldots, e_k) \not\in \sem{R}_\rho$ for every $k$-ary relation $R$.
\end{lemma}

An important property about the minimal model defined above is that
there is a relation preserving homomorphism from this model to
any other model that satisfies the assumptions in the execution.
Formally,
\begin{lemma}
Let $\Mm = (U_\Mm, \sem{}_\Mm)$ be a first order structure 
and let $\rho$ be an execution that is feasible in $\Mm$.
Then, there is a morphsim $h : U_\rho \to U_\Mm$
such that 
\begin{itemize}
  \item $h(\sem{c}_\rho) = \sem{c}_\Mm$ for every constant $c$,
  \item $h(\sem{f}_\rho(e_1, \ldots, e_k)) = \sem{f}_\Mm(h(e_1), \ldots, h(e_k))$ for every $k$-ary function $f$, and
  \item for every $e_1, \ldots, e_k \in U_\rho$ 
  and for every $k$-ary function, we have
  \[
  (e_1, \ldots, e_k) \in \sem{R}_\rho \implies (h(e_1), \ldots, h(e_k)) \in \sem{R}_\Mm
  \]
\end{itemize}
\end{lemma}

Finally, we have that the minimal model is a sufficient to check for feasibility of an execution in some model (of course it is also necessary but that is evident). That is,
\begin{lemma}
\lemlabel{minimal-model-existence}
Let $\rho$ be an execution.
If there is model $\Mm$ such that $\rho$ is feasible in $\Mm$,
then $\rho$ is feasible in the minimal model $\Mm_\rho$.
\end{lemma} 

Below, we present necessary inductive hypotheses to prove the correctness of the
automaton construction. The full proof of correctness can be
re-constructed using the following lemma 
and those used by Mathur et al in~\cite{coherence2019}. % and~\lemref{eq-only-emptyset}.

\begin{lemma}
Let $\rho$ be an  execution that is coherent modulo $\emptyset$.
Let $q = (\equiv, d, P, \rel^+, \rel^-)$ be the state reached after reading $\rho$
in the automaton, i.e., $q = \delta^*_\scc(q_0, \rho)$.
If $q \neq \reject$, then we have
\begin{itemize}
  \item for every $x_1, x_2, \ldots, x_k \in V$ and for every $k$-ary relation $R$,
   such that $(\eqcl{x_1}{\equiv}, \eqcl{x_2}{\equiv}, \ldots, \eqcl{x_k}{\equiv}) \not\in \rel^+(R)$, we have $(e_1, e_2, \ldots, e_k) \not\in \sem{R}_\rho$, in the minimal model of $\rho$,
  where $e_i = \sem{\comp(\rho, x_i)}^\minimal$.
  \item for every $x_1, \ldots, x_k \in V$, and for every $k$-ary relation $R$, we have
  $(\eqcl{x_1}{\equiv}, \eqcl{x_2}{\equiv}, \ldots, \eqcl{x_k}{\equiv}) \in \rel^-(R)$ 
  iff
  for every model $\Mm = (U_\Mm, \sem{}_\Mm)$ for which  $\rho$ is feasible in $\Mm$, we have 
  \[(\sem{\comp(\rho, x_1)}_\Mm \ldots, \sem{\comp(\rho, x_k)}_\Mm) \not\in \sem{R}_\Mm\].
\end{itemize}
\end{lemma}

%!TEX root = main.tex

\section{Proofs from~\secref{relations}}
\applabel{app-relations}

%!TEX root = main.tex

\subsection{Proof of~\thmref{epr-undec}}
\applabel{epr-undec}

\begin{proof}
The undecidability is proved through a reduction from Post's
Correspondence Problem (PCP). Recall that PCP is the following
problem.

\textbf{PCP.} Let $\Delta = \set{a_1, a_2, \ldots, a_k}$ be a finite
alphabet ($|\Delta| > 2$). Given strings $\alpha_1, \alpha_2 \ldots
\alpha_N, \beta_1, \beta_2, \ldots, \beta_N \in \Delta^*$ (with
$N>0$), determine if there is a sequence $i_1, i_2, \ldots, i_M$ such
that $1 \leq i_j \leq N$ for every $1 \leq j \leq M$ and
\[
\alpha_{i_1} \cdot \alpha_{i_2} \cdots \alpha_{i_M} = 
\beta_{i_1} \cdot \beta_{i_2} \cdots \beta_{i_M}
\]

It is well know that the PCP problem is undecidable. We will prove
that given a PCP instance $I = (\Delta, \alpha_1,\ldots \alpha_N,
\beta_1, \ldots \beta_N)$, we can construct a set of EPR axioms $\Aa$,
a program $P_\epr$ that is coherent with respect to $\Aa$, and a
postcondition $\phi$ such that $I$ is a YES instance of PCP iff
$P_\epr$ does not satisfy $\phi$.

\begin{figure}[t!]
\rule[1mm]{6cm}{0.4pt} $P_\epr$ \rule[1mm]{6cm}{0.4pt} \\

\begin{minipage}[H]{0.4\textwidth}

\cd{\small (* For $1 \leq p \neq q \leq N-1$ *)}\\
\passume \cd{(z}$_p\neq$ \cd{z$_q$);} \\

\passume \cd{(R(x, y));} \\ 
\passume \cd{(z}$_0\neq$ \cd{M);}\\
\cd{j} $\passign$ \cd{z}$_0$;\\
\pwhile \code{(j} $\neq$ \code{M) \{} \\ 
\rule[1mm]{0.3cm}{0pt}
\cd{i\_j} $\passign$ \cd{g(j);}\\\\
\rule[1mm]{0.3cm}{0pt}
\pif\cd{(i\_j = z$_1$)\{} \\
\rule[1mm]{0.5cm}{0pt}
\cd{x'} $\passign$ \cd{f(x);} \\
\rule[1mm]{0.5cm}{0pt}
\passume \cd{($S$(x, x'));}\\
\rule[1mm]{0.5cm}{0pt}
\passume \cd{($Q_{\alpha_{1,1}}$(x'));}\\
\rule[1mm]{0.5cm}{0pt}
\cd{x} $\passign$ \cd{x';}\\
\rule[1mm]{0.5cm}{0pt}
\vdots\\
\rule[1mm]{0.5cm}{0pt}
\cd{x'} $\passign$ \cd{f(x);} \\
\rule[1mm]{0.5cm}{0pt}
\passume \cd{($S$(x, x'));}\\
\rule[1mm]{0.5cm}{0pt}
\passume \cd{($Q_{\alpha_{1,|\alpha_1|}}$(x'));}\\
\rule[1mm]{0.5cm}{0pt}
\cd{x} $\passign$ \cd{x';}\\
\rule[1mm]{0.5cm}{0pt}
\cd{y'} $\passign$ \cd{f(y);} \\
\rule[1mm]{0.5cm}{0pt}
\passume \cd{($S$(y, y'));}\\
\rule[1mm]{0.5cm}{0pt}
\passume \cd{($Q_{\beta_{1,1}}$(y'));}\\
\rule[1mm]{0.5cm}{0pt}
\cd{y} $\passign$ \cd{y';}\\
\rule[1mm]{0.5cm}{0pt}
\vdots\\
\rule[1mm]{0.5cm}{0pt}
\cd{y'} $\passign$ \cd{f(y);} \\
\rule[1mm]{0.5cm}{0pt}
\passume \cd{($S$(y, y'));}\\
\rule[1mm]{0.5cm}{0pt}
\passume \cd{($Q_{\beta_{1,|\beta_1|}}$(y'));}\\
\rule[1mm]{0.5cm}{0pt}
\cd{y'} $\passign$ \cd{y';}\\
\rule[1mm]{0.3cm}{0pt}
\cd{\}} \\
\end{minipage}
\begin{minipage}[H]{0.4\textwidth}
\rule[1mm]{0.3cm}{0pt}
\vdots \\
\rule[1mm]{0.3cm}{0pt}
\pelse\cd{\{} \\
\rule[1mm]{0.5cm}{0pt}
\cd{x'} $\passign$ \cd{f(x);} \\
\rule[1mm]{0.5cm}{0pt}
\passume \cd{($S$(x, x'));}\\
\rule[1mm]{0.5cm}{0pt}
\passume \cd{($Q_{\alpha_{N,1}}$(x'));}\\
\rule[1mm]{0.5cm}{0pt}
\cd{x} $\passign$ \cd{x';}\\
% \end{minipage}
% \begin{minipage}[H]{0.33\textwidth}
\rule[1mm]{0.5cm}{0pt}
\vdots\\
\rule[1mm]{0.5cm}{0pt}
\cd{x'} $\passign$ \cd{f(x);} \\
\rule[1mm]{0.5cm}{0pt}
\passume \cd{($S$(x, x'));}\\
\rule[1mm]{0.5cm}{0pt}
\passume \cd{($Q_{\alpha_{N,|\alpha_N|}}$(x'));}\\
\rule[1mm]{0.5cm}{0pt}
\cd{x} $\passign$ \cd{x';}\\

\rule[1mm]{0.5cm}{0pt}
\cd{y'} $\passign$ \cd{f(y);} \\
\rule[1mm]{0.5cm}{0pt}
\passume \cd{($S$(y, y'));}\\
\rule[1mm]{0.5cm}{0pt}
\passume \cd{($Q_{\beta_{N,1}}$(y'));}\\
\rule[1mm]{0.5cm}{0pt}
\cd{y} $\passign$ \cd{y';}\\
\rule[1mm]{0.5cm}{0pt}
\vdots\\
\rule[1mm]{0.5cm}{0pt}
\cd{y'} $\passign$ \cd{f(y);} \\
\rule[1mm]{0.5cm}{0pt}
\passume \cd{($S$(y, y'));}\\
\rule[1mm]{0.5cm}{0pt}
\passume \cd{($Q_{\beta_{N,|\beta_N|}}$(y'));}\\
\rule[1mm]{0.5cm}{0pt}
\cd{y} $\passign$ \cd{y';}\\
\rule[1mm]{0.3cm}{0pt}
\cd{\}} \\
\rule[1mm]{0.3cm}{0pt}
\code{j} $\passign$ \code{s(j);} \\ 
\code{\}} \\
\code{@post:} $\neg$\cd{R(x, y)} \\
\end{minipage}
\caption{Program $P_\epr$ for showing verification is undecidable when there are relations and obey EPR axioms in~\figref{undec-epr-axioms}.}
\figlabel{undec-epr}
\end{figure}

\begin{figure}[h]
\begin{eqnarray}
\forall x, y_1, y_2 \cdot \cd{R}(x, y_1) \land \cd{R}(x, y_2) \implies y_1 = y_2 \\
\forall x_1, x_2, y \cdot \cd{R}(x_1, y) \land \cd{R}(x_2, y) \implies x_1 = x_2 \\
\forall x, y_1, y_2 \cdot \cd{S}(x, y_1) \land \cd{S}(x, y_2) \implies y_1 = y_2 \\
\forall x_1, y_1, x_2, y_2 \cdot \cd{R}(x_1, y_1) \land \cd{S}(x_1, x_2) \land \cd{S}(y_1, y_2) \implies \cd{R}(x_2, y_2) \\
\text{For every } a \in \Delta, \text{ we have } \forall x, y \cdot \cd{R}(x, y) \land \cd{Q}_a(x) \implies \cd{Q}_a(y) \\
\text{For every } a \in \Delta, \text{ we have } \forall x, y \cdot \cd{R}(x, y) \land \cd{Q}_a(y) \implies \cd{Q}_a(x) \\
\text{For every } a \neq b \in \Delta, \text{ we have } \forall x, y \cdot \cd{Q}_a(x)  \implies \neg \cd{Q}_b(x)
\end{eqnarray}
\caption{Axioms for the relations used in $P_\epr$.}
\figlabel{undec-epr-axioms}
\end{figure}

Let us fix a PCP instance $I$. The desired program $P_\epr$ (with post
condition $\phi$) is shown in \figref{undec-epr} and the set of EPR
axioms $\Aa$ is shown in \figref{undec-epr-axioms}.

The signature $\Sigma$ consists of unary functions $\cd{f}, \cd{g}$
and $\cd{s}$.  The set of relations in $\Sigma$ is
\[\Rr = \set{\cd{R}, S} \cup \setpred{Q_a}{a \in \Delta}.\]
The relations \cd{R} and \cd{S} are binary, while the rest are unary relations.
The set of variables in the program are 
\[\Vv = \set{\cd{z}_1 \ldots, \cd{z}_{N-1}} \cup \set{\cd{x}, \cd{x'}, \cd{y}, \cd{y'}, \cd{y}, \cd{z}_0,  \cd{j}, \cd{i\_j}, \cd{M}}
\]

Intuitively, the program constructs two strings that prove that $I$ is
a YES instance of PCP --- the positions on one string are indexed by
the variable $\cd{x}$ and positions on the second string are indexed
by the variable $\cd{y}$. Variable $\cd{M}$ intuitively stores the
number of $\alpha_i$'s that need to be concatenated to get a
solution. The value of $\cd{M}$ is fixed by the data model; this way
of exploiting data models to get ``nondeterminism'' is key in this
reduction. The variable $\cd{z}_0$ stores ``0'', and the variables
$\cd{z}_i$ ($i > 0$) store indices of strings in the input instance
$I$. In each iteration of the $\pwhile$-loop, the index of the next
pair of strings in the solution is ``picked'' by applying the
(uninterpreted) function $\cd{g}$; here again the data model that
interprets $\cd{g}$ resolves the non-determinism. Once the index is
picked, the appropriate strings are ``concatenated''. This happens
step-by-step by generating the next index by applying function
$\cd{f}$, and fixing the symbol at that position. Here the relation
$Q_a$ plays a role; if $Q_a(\cd{x})$ holds then intuitively it means
that symbol $a$ appears in position $\cd{x}$ of the string. Finally,
after the next pair is concatenated, the index of the number of
strings in the solution (a.k.a. $\cd{j}$) is ``incremented'' (by using
$\cd{s}$).

The relations $\cd{R}$ and $S$ play an important role. $S$ is the
successor relation on string positions, and so appropriate $\passume$s
on $S$ are inserted whenever $\cd{f}$ is used. The relation $\cd{R}$
relates positions of the two constructed strings if the prefix upto
that position is identical in the two strings --- we start with
requiring that the first positions are related by $\cd{R}$ and our
post condition demands that the last two positions are not
$\cd{R}$-related to say that the constructed strings are not a
solution to the PCP instance.

The axioms in $\Aa$ ensure that the relations $\cd{R}, S$, and $Q_a$
are interpreted consistently with the above intuition. Axioms (1) and
(2) require that a position in the first/second string is $\cd{R}$
related to at most one position in the second/first string. Axioms
(5) and (6) say that $\cd{R}$-related positions have the same
symbol, while axiom (4) says that if two positions are $\cd{R}$
related then so are their ``successors'' (i.e., $S$-related
elements). Axiom (3) requires $S$ to behave like a successor relation
--- any position as at most one $S$-related position. Finally axiom
(7) intuitively says that there is at most one symbol at any
position.

% The correctness of this reduction and the fact that $P_\epr$ is coherent modulo $\Aa$ is proved in~\appref{epr-undec}; it is based on the intuitions outlined above.

We will now prove the correctness of the reduction outlined in
\figref{undec-epr} and \figref{undec-epr-axioms}.

Let us first argue why $P_\epr$ is coherent modulo the axioms $\Aa$
in~\figref{undec-epr-axioms}.

We first argue that in any execution $\rho$ of $P_\epr$,
there are no equalities implied by the relational assumes.
The only candidate axioms that might imply equalities are
(1), (2) and (3).
In any execution $\rho$, 
the only relational assumes of the form $\cd{R}(t_1, t_2)$
that are implied are of the form 
$\cdm{R(f^n(\init{x}), f^n(\init{y}))}$ ($\cdm{n} \geq 0$)
and thus for a given $t_1$, there is a syntactically unique $t_2$
for which $\cd{R}(t_1, t_2)$ is implied on the computed set of terms,
and thus there is no implied equality using (1) or (2).
Next, the only assumptions of the form $S(t_1, t_2)$
that are observed are of the form $S\cdm{(f^n(\init{z}), f^{n+1}(\init{z}))}$
($\cdm{n} \geq 0$ and $\cd{z} \in \set{\cd{x},\cd{y}}$).
Thus, no equality assumes are implied by (3).

Now, the only equality-\passume~in $\rho$ 
is the one at the end of the \pwhile~loop -- $\passume(\cd{j = M})$. 
At the point where this assume is seen, neither $\cd{j}$ nor $\cd{M}$ have any superterms
and thus there are no implied equalities due to this assume.

Let us now see why $\rho$ is memoizing. 
The terms in \cd{j} are always growing : $\cd{s}^n(\hat{\cd{j}})$ in the $n^{th}$ iteration. 
So both the assignments $\dblqt{\cd{i\_j} \passign \cd{g(j)}}$ and $\dblqt{\cd{j} \passign \cd{s(j)}}$ are memoizing as they never recompute terms.
The same reasoning also applies to the terms in \cd{x} and \cd{y}.

Let us now argue the correctness of the reduction.\\

\noindent
$(\Rightarrow)$. Let us assume that the given PCP instance is a YES instance.
Then, there is a sequence $i_1, i_2, \ldots i_M$ such that
$\alpha_{i_1} \cdot \alpha_{i_2} \cdots \alpha_{i_M} = \cdot \beta_{i_2} \cdot \beta_{i_2} \cdots \beta_{i_M}$.
We can now construct a model that satisfies the EPR axioms in~\figref{undec-epr-axioms}
and violates the post condition.
In this model, \cd{s} is the successor function over $\mathbb{N}$, \cd{z}$_0$ is the number $0$ and \cd{g} maps \cd{j} to $i_\cd{j}$ based on the witness sequence above. 
Further, $\cd{z}_r$ is interpreted as the number $r$. 
The variables \cd{x} and \cd{y} map to $\hat{\cd{x}}$ and $\hat{\cd{y}}$ respectively, which are distinct elements.
The function \cd{f} is such that $\cd{f}^i(\hat{\cd{x}}) \neq \cd{f}^j(\hat{\cd{x}})$ and $\cd{f}^i(\hat{\cd{y}}) \neq \cd{f}^j(\hat{\cd{y}})$ for every $i\neq j \in \mathbb{N}$
and further $\cd{f}^i(\hat{\cd{x}}) \neq \cd{f}^j(\hat{\cd{y}})$ for every $i, j \in \mathbb{N}$. The relations $\cd{Q}_a$ are interpreted as follows: $\cd{Q}_a(\cd{f}^n(\hat{\cd{x}}))$ holds
iff $a$ is the $n^{th}$ character in the sequence $\alpha_{i_1} \cdots \alpha_{i_M}$.
Similarly, $\cd{Q}_a(\cd{f}^n(\hat{\cd{y}}))$ holds
iff $a$ is the $n^{th}$ character in the sequence $\beta_{i_1} \cdots \beta_{i_M}$.
Then, since $\alpha_{i_1} \cdots \alpha_{i_M} = \beta_{i_1} \cdots \beta_{i_M}$, we must have \cd{R(x, y)} at the end of the computation.\\

\newcommand{\rels}{\textsf{rels}}
\noindent
$(\Leftarrow)$. In this case we have a feasible execution $\rho$ with the statement \passume~\cd{R(x, y)} at the end.
% First, we note that the loop is executed at least once. Let $m > 0$ be the number of times the loop is executed.

Consider the initial term model $\Tt$ for the vocabulary $\Sigma$ (without the relations)
and the starting constants $\hat{\Vv} = \setpred{\hat{x}}{x \in \Vv}$.
We show that it is possible to extend the term model $\Tt$ with interpretations of relations such that the resulting model $\Tt_\rels$ is such that $\rho$ is feasible on $\Tt_\rels$.
In fact, the extension is the following model: each binary and unary relation is interpreted to be the smallest relation that satisfies the \passume's in $\rho$ as well as the EPR axioms.
This is well defined because the assumes on relations in $\rho$ are all positive assumes and all EPR axioms are monotonic, except possibly the last one, which can be handled easily: \cd{Q}$_a(t)$ holds iff $\rho$ explicitly demands it.
As can be seen, $\Tt_\rels$ does not violate any negative assume on the relations since there are none.
Further, all equality and disequality assumes are unaffected as in $\Tt_\rels$, there are no terms $t_x, t_y, t_{x_1}, t_{x_2}, t_{y_1}, t_{y_2}$ that can be instantiated for variables in the axioms (1), (2) and (3), as these relations are smallest.
Thus, $\rho$ is feasible on $\Tt_\rels$.
% \ucomment{Argue that if there is a model in which  $\rho$ is feasible, then $\rho$ is feasible on $\Tt_\rels$.}

Now from this model, we will construct the sequence $i_1,\ldots, i_M$.
The length of this sequence $M$ will be the number of times the $\pwhile$~loop is executed.
Clearly, the loop is executed at least once and thus $M > 0$.
Let $t_\cd{x} = \cd{f}^{n_1}(\hat{\cd{x}})$ and $t_{y} = \cd{f}^{n_1}(\hat{\cd{y}})$ be the values of the variables \cd{x} and \cd{y} (in the term model $\Tt_\rels$).
We first argue that $n_1 = n_2$. Assume on the contrary that $n_1 < n_2$ (w.l.o.g.).
Then, one can inductively show that $\cd{R}(\cd{f}^{n_1}(\hat{\cd{x}}), \cd{f}^{n_1}(\hat{\cd{y}}))$; this is because for every $i < n_1$, we have $\cd{S}(\cd{f}^{i}(\hat{\cd{x}}), \cd{f}^{i+1}(\hat{\cd{x}}))$, $\cd{S}(\cd{f}^{i}(\hat{\cd{y}}), \cd{f}^{i+1}(\hat{\cd{y}}))$ and also $\cd{R}(\hat{\cd{x}}, \hat{\cd{y}})$.
But then, in the term model we have $\cd{f}^{n_1}(\hat{\cd{y}}) \neq \cd{f}^{n_2}(\hat{\cd{y}})$ and this violates the assumption at the end of $\rho$ (because of axiom (1)).
Hence, we have $n_1 = n_2.$

Now, the sequence $i_1, \ldots i_M$ can be deduced by the conditional branches in the while loop: the index $i_j$ is the index of the branch taken in the $j^{th}$ iteration.
Let $\alpha = \alpha_{i_1} \cdot \alpha_{i_2} \cdots \alpha_{i_M}$ and
$\beta = \beta_{i_1} \cdot \beta_{i_2} \cdots \beta_{i_M}$.
First we note that $n_1 = |\alpha|$ and $n_2 = |\beta|$ and thus $|\alpha| = |\beta|$.
Let $\alpha_n$ and $\beta_n$ be the $n^{th}$ characters of $\alpha$ and $\beta$ respectively. Then, one can see that $\cd{Q}_{\alpha_n}(\cd{f}^{n}(\hat{\cd{x}}))$
and $\cd{Q}_{\beta_n}(\cd{f}^{n}(\hat{\cd{y}}))$ hold in the term model.
Now, axioms (5), (6) and (7) ensure that $\alpha_n = \beta_n$.
Thus, $\alpha = \beta$. 

\end{proof}

\subsection{Proof of~\thmref{refl-irrefl-symm}}
\applabel{app-refl-irrefl-symm}

\subsubsection{Homomorphisms for Irreflexivity and Symmetry.}

For irreflexivity, whenever we see the computation of a term 
using an assignment of the form $\dblqt{x := f(\vec{z})}$, 
we insert an assume statement that demands that $\neg R(x,x)$ holds.
That is, we instrument executions using the following homomorphism.
\[
h^R_\irrefl(a) = 
\begin{cases}
a \cdot \dblqt{\passume(\neg R(x, x))}  & \text {if } a = \dblqt{x \passign f(\vec{z})} \\
a & \text{ otherwise }
\end{cases}
\]

For the symmetry axiom on a relation $R$, whenever we see an
\emph{assumption} of the form $\dblqt{\passume(R(x, y))}$, we insert
an assumption that $R(y,x)$ holds.  In other words, we use the
following homomorphism.
\[
h^R_\sym(a) = 
\begin{cases}
a \cdot \dblqt{\passume(R(y, x))} & \text {if } a = \dblqt{\passume(R(x, y))} \\
a \cdot \dblqt{\passume(\neg R(y, x))} & \text {if } a = \dblqt{\passume(\neg R(x, y))} \\
a & \text{ otherwise }
\end{cases}
\]

Proof of~\thmref{refl-irrefl-symm} follows from the more general result~\thmref{combination}

%!TEX root = main.tex

\subsection{Proof of~\thmref{transitivity} (Transitivity Axioms)}
\applabel{app-transitivity}

In this section, we prove coherence modulo transitivity is decidable.
More precisely, let $\Rr_\trans$ be the set of binary relations that
are transitive and let $\Aa_\trans = \setpred{\varphi^R_\trans}{R \in \Rr_\trans}$.
We will show that the set $\cohexec(\Sigma, V, \Aa_\trans)$ is a regular language:

\begin{reptheorem}{thm:transitivity-reg}
Let $\Sigma$ be a first order signature and $V$ a finite set of
program variables. Let $\Aa = \setpred{\varphi^R_\trans}{R \in
  \Rr_\trans}$ for some set of relation symbol $\Rr_\trans$ in
$\Sigma$. The following observation hold.
\begin{enumerate}
\item There is a finite automaton $\Ff_\trans$ (effectively
  constructable) of size $O(2^{\text{poly}(|V|)})$ such that for any
  coherent execution $\rho$ that is coherent modulo $\Aa$, 
  $\Ff_\trans$ accepts $\rho$ iff $\rho$ is
  feasible.
\item There is a finite automaton $\Cc_\trans$ (effectively
  constructible) of size $O(2^{\text{poly}(|V|)})$ such that $L(\Cc_\trans) =
  \cohexec(\Sigma, V, \Aa)$.
\end{enumerate}
\end{reptheorem}

For this, we modify the automaton construction 
in~\appref{automata-relations} to accommodate transitive relations.

The states of the automaton are still the same as that 
described in~\appref{automata-relations}. 
Further, the transition function $\delta_\scc$ is such that
for a state $q \neq \reject$, $\delta_\scc(q, a)$
is the same as before when $a \not\in \setpred{\dblqt{\passume(R(x, y))}, \dblqt{\passume(\neg R(x, y))}}{R \in \Rr_\trans} $.
Below we give the modified transitions for these cases.

The intuitive idea behind the modification is as follows.
For $R \in \Rr_\trans$ component $\rel^+(R)$ stores the pairs of
equivalence classes which are implied by the transitive closure
of the observed assume statements $\dblqt{\passume(R(x, y))}$.
For example, if the execution observes $\dblqt{\passume(R(x, y))}$
and $\dblqt{\passume(R(y, z))}$, then
the component $\rel^+(R)$ stores  the pair $(\eqcl{x}{\equiv}, \eqcl{z}{\equiv})$
in addition to the pairs $(\eqcl{x}{\equiv}, \eqcl{y}{\equiv})$ 
and $(\eqcl{y}{\equiv}, \eqcl{z}{\equiv})$.
Next, for every $R \in \Rr_\trans$, the component $\rel^-(R)$
also adds additional pairs $(c_1, c_2)$ of equivalence classes
for which $\neg R(c_1, c_2)$
is implied by the positive and negative assumes in the execution.
More precisely, if the execution observes
$\passume(R (x, y))$ and $\passume(\neg R (x, z))$,
then one can infer the constraint $\neg R(y, z)$,
and in this case, we also add $(\eqcl{y}{\equiv}, \eqcl{z}{\equiv})$ in
$\rel^-(R)$ in addition to $(\eqcl{x}{\equiv}, \eqcl{z}{\equiv})$.
Similarly, if the execution observes
$\passume(R (y, z))$ and $\passume(\neg R (x, z))$,
then one can infer the constraint $\neg R(x, y)$,
and in this case, we also add $(\eqcl{x}{\equiv}, \eqcl{y}{\equiv})$ in
$\rel^-(R)$ in addition to $(\eqcl{x}{\equiv}, \eqcl{z}{\equiv})$.

Let us now give the formal definition of $\delta_\scc(q, a)$
when $q \neq \reject$ and when $a \in \setpred{\dblqt{\passume(R(x, y))}, \dblqt{\passume(\neg R(x, y))}}{R \in \Rr_\trans}$.
As before, if $\rel^+(R) \cap \rel^-(R) \neq \emptyset$, we go to the state $\reject$.

\begin{description}
\item[$a = \dblqt{\passume(R\big(x, y)\big)}$].\\
In this case, $\equiv' = \equiv$, $P' = P$, $d' = d$.
Further, $\rel'^+(R') = \rel^+(R')$ and $\rel'^-(R') = \rel^-(R')$
for every $R' \neq R$.
Further,
\begin{itemize}
	\item $\rel'^+(R)$ is the smallest set such that 
		\begin{enumerate}[label=(\alph*)]
			\item $\rel^+(R) \subseteq \rel'^+(R)$, and 
			\item is transitively closed, i.e., for all $x, y, z \in V$ if $(\eqcl{x}{\equiv'}, \eqcl{y}{\equiv'}) \in \rel'^+(R)$ and $(\eqcl{y}{\equiv'}, \eqcl{z}{\equiv'}) \in \rel'^+(R)$ then $(\eqcl{x}{\equiv'}, \eqcl{z}{\equiv'}) \in \rel'^+(R)$.
		\end{enumerate}

	\item $\rel'^+(R)$ is the smallest set such that 
	\begin{enumerate}[label=(\alph*)]
		\item $\rel^-(R) \subseteq \rel'^-(R)$, and
		\item for all $x, y, z \in V$ if $(\eqcl{x}{\equiv'}, \eqcl{y}{\equiv'}) \in \rel'^+(R)$ and $(\eqcl{x}{\equiv'}, \eqcl{z}{\equiv'}) \in \rel'^-(R)$
	then $(\eqcl{y}{\equiv'}, \eqcl{z}{\equiv'}) \in \rel'^+(R)$, and
		\item for all $x, y, z \in V$ if $(\eqcl{x}{\equiv'}, \eqcl{y}{\equiv'}) \in \rel'^+(R)$ and $(\eqcl{x}{\equiv'}, \eqcl{z}{\equiv'}) \in \rel'^-(R)$
	then $(\eqcl{y}{\equiv'}, \eqcl{z}{\equiv'}) \in \rel'^+(R)$.
\end{enumerate}
\end{itemize}

\item[$a = \dblqt{\passume(\neg R\big(x, y)\big)}$].\\
In this case, $\equiv' = \equiv$, $P' = P$, $d' = d$ 
and $\rel'^+ = \rel^+$ and $\rel^-(R') = \rel'^-(R')$ for every $R' \neq R$. 
Further,
$\rel'^+(R)$ is the smallest set such that 
	\begin{enumerate}[label=(\alph*)]
		\item $\rel^-(R) \subseteq \rel'^-(R)$, and
		\item for all $x, y, z \in V$ if $(\eqcl{x}{\equiv'}, \eqcl{y}{\equiv'}) \in \rel'^+(R)$ and $(\eqcl{x}{\equiv'}, \eqcl{z}{\equiv'}) \in \rel'^-(R)$
	then $(\eqcl{y}{\equiv'}, \eqcl{z}{\equiv'}) \in \rel'^+(R)$, and
		\item for all $x, y, z \in V$ if $(\eqcl{x}{\equiv'}, \eqcl{y}{\equiv'}) \in \rel'^+(R)$ and $(\eqcl{x}{\equiv'}, \eqcl{z}{\equiv'}) \in \rel'^-(R)$
	then $(\eqcl{y}{\equiv'}, \eqcl{z}{\equiv'}) \in \rel'^+(R)$.
\end{enumerate}
\end{description}

In order to argue correctness, we extend the notion of
\emph{minimal model} to transitivity.

\begin{definition}
Let $\Aa_\trans$ be the set of transitivity axioms on 
some finite set of binary relations $\Rr_\trans$.
Let $\rho$ be an execution and let $A$ be the set of
equalities in $\kappa(\rho)$.
Let $\Mm_\rho$ be the minimal model for $\rho$.
Define the minimal transitive model (with respect to $\Rr_\trans$)
of $\rho$ to be the model $\Mm^\trans_\rho = (U^\trans_\rho, \sem{}^\trans_\rho)$
such that $U^\trans_\rho = U_\rho$,
$\sem{c}^\trans_\rho = \sem{c}_\rho$ for every $c \in \Cc$,
$\sem{f}^\trans_\rho = \sem{f}_\rho$ for every $f \in \Ff$
and 
$\sem{R}^\trans_\rho = \sem{R}_\rho$ for every $R \in \Rr \setminus \Rr_\trans$.
Further, for every $R \in \Rr_\trans$, define
$\sem{R}^\trans_\rho$ to be the smallest transitive set containing
$\sem{R}_\rho$.
\end{definition}

Notice that the execution $\rho$ only defines a relation on the set
of computed terms, and thus the transitive closure of the 
observed assumes also stays with the set of computed terms.
This is formalized below.
\begin{lemma}
Let $\rho$ be an execution and let $\Mm^\trans_\rho$ be the minimal transitive model
as defined above.
Let $e_1, e_2 \in U^\trans_\rho$  be elements in the minimal model such that either $e_1$ or $e_2$ is
not computed by the execution (i.e., there is an $i \in \set{1, 2}$ such that
for every $t \in \Terms(\rho)$, $\sem{t}_{\rho}^\trans \neq e_i$).
Then, we have $(e_1, e_2) \not\in \sem{R}^\trans_\rho$.
\end{lemma}

An important property about the minimal transitive model defined above is that
there is a relation preserving homomorphism from this model to
any other model that satisfies the assumptions in the execution and the transitivity axioms.
Formally,
\begin{lemma}
Let $\Mm = (U_\Mm, \sem{}_\Mm)$ be a first order model 
and let $\rho$ be an execution that is feasible in $\Mm$, modulo $\Aa_\trans$.
Then, there is a morphsim $h : U^\trans_\rho \to U_\Mm$
such that 
\begin{itemize}
	\item $h(\sem{f}^\trans_\rho(e_1, \ldots, e_k)) = \sem{f}_\Mm(h(e_1), \ldots, h(e_k))$ for every $k$-ary function $f$, and
	\item for every $e_1, \ldots, e_k \in U^\trans_\rho$ 
	and for every $k$-ary function, we have
	\[
	(e_1, \ldots, e_k) \in \sem{R}^\trans_\rho \implies (h(e_1), \ldots, h(e_k)) \in \sem{R}_\Mm
	\]
\end{itemize}
\end{lemma}

Finally, we have that the minimal model is enough to check for feasibility of an execution in some model. That is,
\begin{lemma}
Let $\rho$ be an execution that is feasible in $\Mm$.
Let $\Aa_\trans$ be the set of transitivity axioms for relations in
$\Rr_\trans$.
If there is model $\Mm$ such that $\rho$ is feasible in $\Mm$,
then $\rho$ is feasible in the minimal model $\Mm^\trans_\rho$.
\end{lemma} 

We prove the correctness of the automaton construction by inducting
on the length of the word.
For this, we will be using the following inductive invariants.

\begin{lemma}
Let $\Aa_\trans$ be the set of transitivity axioms on 
some finite set of binary relations $\Rr_\trans$.
Let $\rho$ be an  execution that is coherent modulo $\Aa_\trans$.
Let $q = (\equiv, d, P, \rel^+, \rel^-)$ be the state reached after reading $\rho$
in the automaton, i.e., $q = \delta^*_\scc(q_0, \rho)$.
If $q \neq \reject$, then we have (here $R \in \Rr_\trans$)
\begin{itemize}
	\item for every $x, y \in V$ such that $(\eqcl{x}{\equiv}, \eqcl{y}{\equiv}) \not\in \rel^+(R)$, we have $(e_x, e_y) \not\in \sem{R}^\trans_\rho$, in the minimal model of $\rho$,
	where $e_x = \sem{\comp(\rho, x)}^\trans_\rho$ and $e_y = \sem{\comp(\rho, y)}^\trans_\rho$.
	\item for every $x, y \in V$, $(\eqcl{x}{\equiv}, \eqcl{y}{\equiv}) \in \rel^-(R)$ 
	iff
	for every model $\Mm = (U_\Mm, \sem{}_\Mm)$ for which  $\rho$ is feasible in $\Mm$, we have $(\sem{\comp(\rho, x)}_\Mm, \sem{\comp(\rho, y)}_\Mm) \not\in \sem{R}_\Mm$.
\end{itemize}
\end{lemma}

\subsection{Proof of~\thmref{strict-partial-order}}

Follows from the more general result~\thmref{combination}

\subsection{Proof of~\thmref{total}}

We first observe that when executions have only positive $R$ assumes,
checking properties modulo $\Aa^R_\sto$ is equivalent to checking
properties modulo $\Aa^R_\spo$. This will allow us to reduce the case
of strict total orders to the case of strict partial orders.

\begin{lemma}
\lemlabel{total-to-partial}
Let $\Aa$ be a set of first order sentences that do not mention
$R$. Let $\rho$ be an execution that does not have any symbols of the
form $\dblqt{\passume(\neg R(x,y))}$. Then the following two
observations hold.
\begin{enumerate}
\item $\rho$ is feasible modulo $\Aa \cup \Aa^R_\sto$ iff $\rho$ is
  feasible modulo $\Aa \cup \Aa^R_\spo$; note that $\rho$ may or may
  not be coherent.
\item $\rho$ is coherent modulo $\Aa \cup \Aa^R_\sto$ iff $\rho$ is
  coherent modulo $\Aa \cup \Aa^R_\spo$.
\end{enumerate}
\end{lemma}

\begin{proof}
We first argue about feasibility.
One direction is obvious: if $\rho$ is feasible modulo $\Aa \cup \Aa^R_\sto$, then
$\rho$ is feasible modulo $\Aa \cup \Aa^R_\spo$.
Let us consider the other direction.
Let $\Mm = (\Uu_\Mm, \sem{}_\Mm)$ be a model in which $\rho$ is feasible, such that $\Mm \models \Aa \cup \Aa^R_\spo$.
Notice that $\sem{R}_\Mm$ is a partial order on $\Uu_\Mm$.
Let $S$ be any linear extension of $\sem{R}_\Mm$.
Consider the model $\Mm' = (\Uu_{\Mm'}, \sem{}_{\Mm'})$, where $\Uu_{\Mm'} = \Uu_\Mm$
and for every constant symbol $c$, function symbol $f$ and relation symbol $Q$ different from $R$, we have $\sem{c}_{\Mm'} = \sem{c}_\Mm$, $\sem{f}_{\Mm'} = \sem{f}_\Mm$ and $\sem{Q}_{\Mm'} = \sem{Q}_\Mm$.
Finally, let $\sem{R}_{\Mm'} = S$.
First, observe that $\Mm' \models \Aa$ because no sentence in $\Aa$ mentions $R$.
Second, $\Mm' \models \Aa^R_\sto$ by construction.
Finally, $\sem{R}_\Mm \subseteq \sem{R}_{\Mm'}$ and thus $\Mm' \models \kappa(\rho)$.

Let us now note that for any two computed terms $t_1, t_2 \in \Terms(\rho)$, we have
$t_1 \congcl{\Aa \cup \Aa_\sto^R \cup \kappa(\rho)} t_2$ iff $t_1 \congcl{\Aa \cup \Aa_\spo^R \cup \kappa(\rho)} t_2$. The proof of this observation is similar to the proof of~\lemref{preservation-term-equalities-relations} and is skipped.

Now, based on the above observations, one can easily conclude that 
$\rho$ is coherent modulo $\Aa \cup \Aa^R_\sto$ iff $\rho$ is
  coherent modulo $\Aa \cup \Aa^R_\spo$.
\end{proof}

The proof of~\thmref{total} follows from~\thmref{strict-partial-order} and~\lemref{total-to-partial}.

% \input{app-strict-partial-order}

% \input{app-strict-total-order}

%!TEX root = main.tex

\section{Proofs from~\secref{functions}}
\applabel{app-functions}

%!TEX root = main.tex

\subsection{Proof of~\thmref{assoc-undec}}
\applabel{assoc-undec}

%!TEX root = main.tex

\begin{figure}[h]
\rule[1mm]{6cm}{0.4pt} $P_\assoc$ \rule[1mm]{6cm}{0.4pt} \\

\begin{minipage}[H]{0.4\textwidth}

\cd{\small (* generate $u_1$ *)} \\
\cd{\small (* $u_{1,i}$ is the $i^{th}$ letter in $u_1$ *)}\\
\cd{x}$_{1,2}$ $\passign$ \cd{f}($u_{1,1}$, $u_{1,2}$);\\
\cd{x}$_{1,3}$ $\passign$ \cd{f}($x_{1,2}$, $u_{1,3}$);\\
\vdots\\
\cd{x}$_{1, |u_1|}$ $\passign$ \cd{f}($x_{1,|u_1|-1}$, $u_{1,|u_1|}$);\\
\cd{x}$_{1}$ $\passign$ \cd{x}$_{1, |u_1|}$;\\

\cd{\small (* generate $v_1$. *)}\\
\cd{y}$_{1,2}$ $\passign$ \cd{f}($v_{1,1}$, $v_{1,2}$);\\
\cd{y}$_{1,3}$ $\passign$ \cd{f}($y_{1,2}$, $v_{1,3}$);\\
\vdots\\
\cd{y}$_{1,|v_1|}$ $\passign$ \cd{f}($y_{1,|v_1|-1}$, $v_{1,|v_1|}$);\\
\cd{y}$_{1}$ $\passign$ \cd{y}$_{1, |v_1|}$;\\

\cd{\small (* assume $u_1 = v_1$. *)}\\
\passume ($\cd{x}_1 = \cd{y}_1$);\\

\cd{\small (* generate $u_2$.  *)}\\
\vdots \\
\cd{\small (* generate $v_2$.  *)}\\
\vdots\\
\cd{\small (* assume $u_2 = v_2$. *)}\\
\passume ($\cd{x}_2 = \cd{y}_2$);\\\\
\vdots\\
\end{minipage}
\begin{minipage}{0.4\textwidth}
\vdots
\cd{\small (* generate $u_N$.  *)}\\
\vdots \\
\cd{\small (* generate $v_N$.  *)}\\
\vdots\\
\cd{\small (* assume $u_N = v_N$. *)}\\
\passume ($\cd{x}_N = \cd{y}_N$);\\

\cd{\small (* generate $u_0$. *)}\\
\cd{x}$_{0,2}$ $\passign$ \cd{f}($u_{0,1}$, $u_{0,2}$);\\
\cd{x}$_{0,3}$ $\passign$ \cd{f}($x_{0,2}$, $u_{0,3}$);\\
\vdots\\
\cd{x}$_{0, |u_0|}$ $\passign$ \cd{f}($x_{0,|u_0|-1}$, $u_{0,|u_0|}$);\\
\cd{x}$_{0}$ $\passign$ \cd{x}$_{0, |u_0|}$;\\

\cd{\small (* generate $v_0$. *)}\\
\cd{y}$_{0,2}$ $\passign$ \cd{f}($v_{0,1}$, $v_{0,2}$);\\
\cd{y}$_{0,3}$ $\passign$ \cd{f}($y_{0,2}$, $v_{0,3}$);\\
\vdots\\
\cd{y}$_{0,|v_0|}$ $\passign$ \cd{f}($y_{0,|v_0|-1}$, $v_{0,|v_0|}$);\\
\cd{y}$_{0}$ $\passign$ \cd{y}$_{0, |v_0|}$;\\

\passume ($\cd{x}_0 \neq \cd{y}_0$);\\
\end{minipage}
% }
\caption{Execution $\rho_\assoc$ for showing checking feasibility of a single coherent execution with one associative function is undecidable. 
}
\figlabel{undec-assoc-wordproblem}
\end{figure}

To prove~\thmref{assoc-undec}, we recall a classical computational problem
called the word problem for a semi-group. Recall that a semi-group is
an algebra consisting of a universe on which a single associative
binary operation (often denoted $\circ$) is defined. A semi-group $S$
is \emph{generated} from a finite set $\Delta$, every element in the
universe of $S$ can be constructed starting from $\Delta$ using the
operation $\circ$. The word problem over semi-groups is the following.\\

\noindent
\textbf{Word Problem over Semi-Groups}. Let $\Delta$ be a finite set
and $\circ$ be the concatenation operation. Given word identities $u_1
= v_1$, $u_2 = v_2$, \ldots $u_n = v_n$, and an additional identity
$u_0 = v_0$, determine if for \emph{any} semi-group $S$ generated from
$\Delta$ in which the identities $u_i = v_i$, for $1 \leq i \leq n$
hold, whether $u_0 = v_0$ holds.

This problem is knwon to be undecidable.
\begin{theorem}[Post'47~\cite{post1947}]
Word problem for finitely generated semigroups is undecidable.
\end{theorem}

Using Post's result, we prove undecidability to check the feasibility
of a single coherent execution.

We show the following reduction. Given an instance $I =
(\Delta,\circ,u_1,v_1,\ldots u_n,v_n,u_0,v_0)$ there is an execution
$\rho$ that is coherent modulo $\set{\varphi^f_\assoc}$ such that $I$
is a YES instance of the work problem iff $\rho$ is infeasible modulo
$\set{\varphi^f_\assoc}$.

The constructed execution $\rho$ is shown in
\figref{undec-assoc-wordproblem}. The signature $\Sigma$ consists of a
binary function $\cd{f}$ which obeys the associativity axiom
$\varphi^{\cd{f}}_\assoc$.  The set of variables in the program are
\[\Vv = \underbrace{\set{\cd{a}_1 \ldots \cd{a}_k}}_{
\scriptsize 
\begin{array}{c}
\text{these are}\\
\text{letters from }\Delta 
\end{array}
} \cup \bigcup\limits_{i = 0}^N \set{\cd{x}_i, \cd{x}_{i, 2}, \ldots \cd{x}_{i, {|u_i|}}} \cup \bigcup\limits_{i = 1}^N \set{\cd{y}_i, \cd{y}_{i, 2}, \ldots \cd{y}_{i, |v_i|}}
\]
The post-condition $\phi$ is $\cd{x}_0 = \cd{y}_0$.

Our reduction uses the associative function $\cd{f}$ to model
concatenation. A word $w = a_1, \ldots, a_m$ is modeled as the term
$t_w = \cd{f}(a_1, \cd{f}(a_2, \ldots, \cd{f}(a_{m-1}, a_m) \ldots
))$. Intuitively, the execution first creates the words $u_1$, $v_1$,
$u_2 \ldots u_N, v_N$ and assumes $u_1 = v_1$, $u_2 = v_2, \ldots, u_N
= v_N$.  It then creates the words $u_0, v_0$ and checks $u_0 = v_0$
in the postcondition. Proof that $\rho$ is coherent and the reduction
is correct is straightforward, but for completeness, the proof can be
found in \appref{assoc-undec}.

We prove that the execution $\rho$ shown in
\figref{undec-assoc-wordproblem} is coherent and the reduction is
correct.

Let us first argue why $\rho_\assoc$ is coherent modulo associativity
of $\cd{f}$.  This follows because all created terms are being
retained in some program variables.

Now, we argue the correctness of the reduction.\\

$(\Leftarrow)$. Assume that the given instance of the word problem is a NO instance.
%In this case, we have to show that $\rho_\assoc$ is infeasible.
Then, there is a semi group $(A, \circ)$ and a homomorphism $h : S \to A$ such that for each $1 \leq i \leq N$, $h(u_i) = h(v_i)$ and $h(u_0) \neq h(v_0)$.
Then, the model $\Mm = (U_\Mm, \sem{}_\Mm)$ with $U_\Mm = A$ and $\sem{f}_\Mm = \circ$
is the model on which $\rho_\assoc$ is feasible. Further, $\sem{f}_\Mm$ is associative
and all the assumptions in $\rho_\assoc$ hold in the model.

$(\Rightarrow)$. Assume the execution $\rho_\assoc$ is feasible modulo associativity.
That is, there is a model $\Mm = (U_\Mm, \sem{}_\Mm)$ such that $\sem{f}_\Mm$ is associative
and all the assumes in the execution are true in the model.
Then, clearly $U_\Mm$ with $\sem{f}_\Mm$ as concatenation is a semigroup.
Further, there is a homomorphism $h$ from $S$ to $A = (U_\Mm, \sem{f}_\Mm)$
given by $h(\cd{a}_i) = \sem{\cd{a}_i}_\Mm$ for every $\cd{a}_i \in \Delta$.
Since the string $u_0$ and $v_0$ are not equal in $A$, the equality $u_0 = v_0$, the given instance is a NO instance of the word problem.

\subsection{Proofs of \thmref{commutativity} and \thmref{idempotence}}
\applabel{app-comm-assoc}

\subsubsection{Homomorphism for Idempotence.}
We use an auxiliary
variable $v^* \not\in V$ and use the following homomorphism.
\[
h^f_\idem(a) = 
\begin{cases}
a \cdot \dblqt{v^*\passign f(y)} \cdot \dblqt{\passume(y = v^*)}  & \text {if } a = \dblqt{y \passign f(x)}\\
a & \text{ otherwise }
\end{cases}
\]

Proof of~\thmref{commutativity} and~\thmref{idempotence} Follows from the more general result~\thmref{combination}.

%!TEX root = main.tex

\section{Combinations}
\applabel{app-comb}

\begin{definition}[Closure of binary relations]
Let $R \subseteq S \times S$ be a binary relation on a set $S$.
Let $p \in \set{\refl,\irrefl,\sym, \trans}$ and let $\varphi^R_p$
be the axiom of reflexivity, irreflexivity, symmetry or transitivity of $R$
(depending upon what $p$ is).
Then, the $p$-closure of $R$, denoted $p(R)$ is defined as
\begin{itemize}
	\item the smallest binary relation $R'\subseteq S \times S$
	such that $R \subseteq R'$ and $R'$ satisfies $\varphi^{R'}_p$,
	if $p \in \set{\refl,\sym, \trans}$, or
	\item the largest binary relation $R' \subseteq S \times S$
	such that $R \supseteq R'$ and $R'$ satisfies $\varphi^{R'}_p$,
	if $p = \irrefl$
\end{itemize}
\end{definition}

\begin{definition}[Relational Closure Extension]
Let $\Mm = (U_\Mm, \sem{}_\Mm)$ be a first order model over the
signature $\Sigma = (\Cc, \Ff, \Rr)$.
Let $R \in \Rr$ be a binary relation and let $p \in \set{\refl,\sym,\irrefl,\trans}$. 
The $(p, R)$-closure of $\Mm$, denoted $\closureext{R}{p}(\Mm)$ 
is the model
$\Mm' = (U_{\Mm'}, \sem{}_{\Mm'})$, where
$U_{\Mm'} = U_\Mm$, 
\begin{itemize}
	\item for every $c \in \Cc$, we have $\sem{c}_{\Mm'} = \sem{c}_{\Mm}$,
	\item for every $f \in \Ff$, we have $\sem{f}_{\Mm'} = \sem{f}_{\Mm}$, and
	\item for every $Q \in \Rr \setminus \set{R}$, we have $\sem{Q}_{\Mm'} = \sem{Q}_\Mm$, and
	\item $\sem{R}_{\Mm'} = p(\sem{R}_\Mm)$.
\end{itemize} 
We define the $(p, R)$-transitive closure of $\Mm$, denoted 
$\transclosureext{R}{p}(\Mm)$ to be the model
$\Mm' = \closureext{R}{\trans}(\closureext{R}{p}(\Mm))$.
\end{definition}

\begin{proposition}
\proplabel{terms-evaluate-same-relational-closure}
Let $\Sigma$ be a FO signature, $R \in \Rr$, 
$p \in \set{\refl,\sym,\irrefl,\trans}$
and let $\Mm$ be a $\Sigma$ structure.
Let $t$ be a term over the signature $\Sigma$.
Then, $\sem{t}_{\Mm} = \sem{t}_{\closureext{R}{p}(\Mm)} = \sem{t}_{\transclosureext{R}{p}(\Mm)} \in U_\Mm$.
\end{proposition}

\begin{definition}[Invariance Under Relational Closure Extension]
Let $\Aa$ be a set of first order sentences over $\Sigma$.
Let $R \in \Rr$ be a binary relation and let $p \in \set{\refl, \irrefl, \sym, \trans}$.
$\Aa$ is said to be invariant under $(p, R)$-closure extension if
for every first order structure $\Mm$, we have
\[
\Mm \models \Aa \implies \closureext{R}{p}(\Mm) \models \Aa
\]
Similarly, $\Aa$ is said to be invariant under $(p, R)$-transitive closure extension if
for every first order structure $\Mm$, we have
\[
\Mm \models \Aa \implies \transclosureext{R}{p}(\Mm) \models \Aa
\]
\end{definition}

\begin{lemma}[Preservation of invariance under unions]
Let $\Aa_1, \Aa_2$ be two sets of first order sentences over $\Sigma$.
Let $R \in \Rr$ be a binary relation and let $p \in \set{\refl, \irrefl, \sym, \trans}$.
If both $\Aa_1$ and $\Aa_2$ are invariant under $(p, R)$-closure extension, then
so is $\Aa_1 \cup \Aa_2$.
\end{lemma}

\begin{proof}
Follows easily from definitions.
\end{proof}

\begin{lemma}[Invariance Under Relational Closure Extensions]
Let $\Sigma = (\Cc, \Ff, \Rr)$ be a FO signature, let $R \in \Rr$ be a binary relation
and let $p\in \set{\refl, \irrefl, \sym}$.
Then we have the following.
\begin{enumerate}
	\item The empty set of axioms $\Aa = \emptyset$ is 
	invariant under $(p,R)$-closure extension.
	\item The singleton set $\Aa = \set{\varphi}$ 
	($\varphi$ is a FO sentence over $\Sigma$) is 
	invariant under $(p,R)$-closure extension if one of the following holds:
	\begin{enumerate}[label=\alph*.]
		\item $\phi$ does not syntactically mention the symbol $R$.
		\item $p = \refl$ and $\phi \in \set{\varphi_\refl^R, \varphi_\sym^R, \varphi_\trans^R}$
		\item $p = \irrefl$ and $\phi \in \set{\varphi_\irrefl^R, \varphi_\sym^R, \varphi_\trans^R}$
		\item $p = \sym$ and $\phi \in \set{\varphi_\irrefl^R, \varphi_\refl^R, \varphi_\sym^R}$
	\end{enumerate}
	\item The singleton set $\Aa = \set{\varphi}$ is 
	invariant under $(p,R)$-transitive-closure extension if $p = \sym$ and $\phi = \varphi_\trans^R$.
\end{enumerate}
\end{lemma}

\begin{definition}[1-element extensions]
Let $\Mm = (U_\Mm, \sem{}_\Mm)$ be a first order model over the
signature $\Sigma = (\Cc, \Ff, \Rr)$.
The one variable extension $\oneext(\Mm)$ of $\Mm$ is another model
$\Mm' = (U_{\Mm'}, \sem{}_{\Mm'})$, where
$U_{\Mm'} = U_\Mm \uplus \set{e_*}$, where $e_* \not\in U_\Mm$ is a
fresh element, and
\begin{itemize}
	\item for every $c \in \Cc$, we have $\sem{c}_{\Mm'} = \sem{c}_{\Mm} \in U_{\Mm}$,
	\item for every $R \in \Rr$ of arity $r$, we have $\sem{R}_{\Mm'} = \sem{R}_\Mm$,
	\item for every $f \in \Ff$ of arity $r$, we have 
	\begin{align*}
		\sem{f}_{\Mm'}(e_1, \ldots, e_r) = 
		\begin{cases}
		\sem{f}_{\Mm}(e_1, \ldots, e_r) & \text{ if } e_* \not\in \set{e_1, \ldots, e_r} \\
		e_* & \text{otherwise}
		\end{cases}
	\end{align*}
\end{itemize} 
\end{definition}
The above is well defined, in that for every model $\Mm$,
there is a unique (upto first order isomorphisms\footnote{
	More precisely, for every first order structure $\Mm$ over
	$\Sigma$, for two 1 element extensions $\Mm_1$ and $\Mm_2$ of $\Mm$,
	and for every first order formula
	$\phi$ over $\Sigma$, $\Mm_1 \models \phi$ iff $\Mm_2 \models \phi$.
}) 
one element extension $\oneext(\Mm)$.

\begin{proposition}
\proplabel{terms-evaluate-same-1-element}
Let $\Sigma$ be a FO signature and let $\Mm$ be a $\Sigma$ structure.
Let $t$ be a term over the signature $\Sigma$.
Then, $\sem{t}_{\Mm} = \sem{t}_{\oneext(\Mm)} \in U_\Mm$.
\end{proposition}

\begin{definition}[Invariance Under 1-element Extension]
Let $\Aa$ be a set of first order sentences over $\Sigma$.
$\Aa$ is said to be invariant under 1-element extension if
for every first order structure $\Mm$, we have
\[
\Mm \models \Aa \implies \oneext(M) \models \Aa
\]
\end{definition}

\begin{lemma}[Preservation of invariance under unions]
Let $\Aa_1, \Aa_2$ be two sets of first order sentences over $\Sigma$.
If both $\Aa_1$ and $\Aa_2$ are invariant under 1-element extensions, then
so is $\Aa_1 \cup \Aa_2$.
\end{lemma}

\begin{proof}
Follows easily from definitions.
\end{proof}

\begin{lemma}[Invariance Under One Element Extension]
\lemlabel{invariance-one-element}
Let $\Sigma = (\Cc, \Ff, \Rr)$ be a FO signature.
Then,
\begin{enumerate}
	\item The empty set of axioms $\Aa = \emptyset$ is 
	invariant under 1-element extensions.
	\item Let $\varphi$ be a FO sentence and let $\Aa = \set{\varphi}$.
	Then $\Aa$ is invariant under 1-element extensions
	if one of the following holds:
	\begin{enumerate}[label=\alph*.]
		\item $\phi$ is quantifier free. That is $\phi$ is a boolean combination of ground equality atoms or predicate atoms.
		\item $\phi \in \set{\varphi^f_\comm, \varphi^f_\idem}$, 
		where $f \in \Ff$ is either a unary or binary function.
		\item $\phi = \varphi^R_\trans$, where $R \in \Rr$ is a binary relation.
	\end{enumerate}
\end{enumerate}
\end{lemma}

\begin{lemma}
\lemlabel{same-terms}
Let $R\in \Rr$ be a binary relation, $f \in \Ff$ be a unary function,
$g \in \Ff$ be a binary function and
Let $\rho$ be an execution
Then for every variable $x$,
\begin{enumerate}
	\item $\comp(\rho,x) = \comp(h^R_p(\rho),x)$ for $p \in \set{\refl,\irrefl,\sym}$,
	\item $\comp(\rho,x) = \comp(h^f_\idem(\rho),x)$, and
	\item $\comp(\rho,x) = \comp(h^g_\comm(\rho),x)$
\end{enumerate}
\end{lemma}
\begin{proof}
The only difference between $\rho$ and $h^R_p(\rho)$ is the fact that
$h^R_p(\rho)$ has additional $\passume$s. The observation therefore,
follows.
\end{proof}

\begin{corollary}
\corlabel{kappa-contain}
For $p \in \set{\refl,\irrefl,\sym}$ and any execution $\rho$,
$\kappa(\rho) \subseteq \kappa(h^R_p(\rho))$.
\end{corollary}

\begin{definition}[Execution-restriction]
Let $\rho$ be an execution, 
$R \in \Rr$ be a relation and $\Mm$ be a model over which $\rho$ is feasible.
The $(R,\rho)$-restriction of $\Mm$, denoted $\execrestrict^R(\Mm, \rho)$
is the model $\Mm' = (U_{\Mm'}, \sem{}_{\Mm'})$, where
$U_{\Mm'} = U_\Mm$, 
\begin{itemize}
	\item for every $c \in \Cc$, we have $\sem{c}_{\Mm'} = \sem{c}_{\Mm}$,
	\item for every $f \in \Ff$, we have $\sem{f}_{\Mm'} = \sem{f}_{\Mm}$, and
	\item for every $Q \in \Rr \setminus \set{R}$, we have $\sem{Q}_{\Mm'} = \sem{Q}_\Mm$, and
	\item $\sem{R}_{\Mm'} = \setpred{(\sem{t_1}_\Mm, \ldots, \sem{t_k}_\Mm)}{R(t_1, \ldots, t_k) \in \kappa(\rho)}$.
\end{itemize} 
\end{definition}

\begin{lemma}[Preservation of Feasibility]
\lemlabel{preservation-feasibility-relations}
Let $\rho$ be an execution, $R \in \Rr$ and 
$p \in \set{\refl,\irrefl,\sym}$.
Let $\Aa$ be a set of first order sentences over $\Sigma$
such that if $p = \sym$ and $\varphi^R_\trans \in \Aa$, then
the only other sentences in $\Aa$ that mention $R$ are
$\varphi^R_\refl, \varphi^R_\irrefl$ or $\varphi^R_\trans$.
Further, assume that $\Aa$ is
invariant under $(p,R)$-closure extension 
(or $(p,R)$-transitive closure extension if $\varphi^R_\trans \in \Aa$).
$\rho$ is feasible modulo $\Aa \uplus \set{\varphi^R_p}$
iff
$h^R_p(\rho)$ is feasible modulo $\Aa$.
\end{lemma}

\begin{proof}
($\Rightarrow$). Let $\rho$ be feasible modulo $\Aa \uplus \set{\varphi^R_p}$.
Then, there is a model $\Mm$ such that $\Mm \models \Aa \uplus \set{\varphi^R_p}$
and $\Mm \models \kappa(\rho)$.
Now, we first note that all the ground predicates 
$\phi \in \kappa(h^R_p(\rho)) \setminus \kappa(\rho)$, we have that
$\Mm \models \phi$ because $\Mm \models \set{\varphi^R_p}$.
This means, that $h^R_p(\rho)$ is feasible in $\Mm$, thus implying that
$h^R_p(\rho)$ is feasible modulo $\Aa$.

($\Leftarrow$).
Let $\rho' = h^R_p(\rho)$ be feasible modulo $\Aa$.
Then, there is a model $\Mm$ such that $\Mm \models \Aa$
and $\Mm \models \kappa(h^R_p(\rho))$.

Now we define a model $\Mm'$ as follows.
\begin{enumerate}
	\item If $p \neq \sym$ or $\varphi^R_\trans \not\in \Aa$, then
	$\Mm' = \closureext{R}{p}(\Mm)$.
	\item If $p = \sym$, $\varphi^R_\trans \in \Aa$ and $\varphi^R_\refl \not\in \Aa$, then
	$\Mm' = \transclosureext{R}{\sym}(\Mm'')$, 
	where $\Mm'' = \closureext{R}{\trans}(\execrestrict^R(\Mm, h^R_\sym(\rho)))$.
	\item If $p = \sym$, $\varphi^R_\trans \in \Aa$ and $\varphi^R_\refl \in \Aa$, then
	$\Mm' = \transclosureext{R}{\refl}(\Mm'')$, 
	where $\Mm'' = \transclosureext{R}{\refl}(\execrestrict^R(\Mm, h^R_\sym(\rho)))$.
\end{enumerate}

% \noindent
\textbf{Observations about $\Mm''$}.
We first make some simple observations about the intermediate 
model $\Mm''$ defined in the last two cases above.
\begin{enumerate}
	\item First, $\sem{R}_{\Mm''}$ is a symmetric relation.
	\item Second, $\Mm'' \models \Aa$.
	To see this, observe that $\Mm'' \models \Aa'$, where $\Aa'$ is the subset
	of $\Aa$ that do not contain sentences involving $R$.
	Further, $\Mm'$ is transitively (and also reflexively, as appropriate)
	closed and the only sentences that mention $R$ are reflexivity, irreflexivity or transitivity axioms.
	\item Third, $\sem{R}_{\Mm''} \subseteq \sem{R}_\Mm$.
	\item Fourth, $\Mm'' \models \kappa(h^R_\sym(\rho))$.
	This is because all terms evaluate to the same elements in $\Mm$ and 
	$\Mm''$ so all equality and disequality assumes in 
	$\kappa(h^R_\sym(\rho))$ hold in $\Mm''$.
	Similarly, for every relation $Q$ different from $R$, 
	all the assumptions involving $Q$ in $\kappa(h^R_\sym(\rho))$ hold.
	Let us consider the assumptions involving $R$. 
	All positive assumptions hold because of the way $\Mm''$ is defined.
	Let us consider a negative assume $\neg R(t_1, t_2) \in \kappa(h^R_\sym(\rho))$.
	Let $e_1 = \sem{t_1}_\Mm$ and $e_2 = \sem{t_2}$.
	We know that $\Mm \models \kappa(h^R_\sym(\rho))$ and thus
	$(e_1, e_2) \not\in \sem{R}_\Mm$ (also since $\sem{R}_\Mm$ is symmetric, we have
	$(e_2, e_1) \not\in \sem{R}_\Mm$)
	Assume on the contrary that $(e_1, e_2) \in \sem{R}_{\Mm''}$.
	Then, there are two possible cases.
	The first case is $e_1 = e_2$ (and thus $(e_1, e_2)$ is in $\sem{R}_{\Mm''}$
	because of reflexive closure). 
	But then $\Aa$ contains the reflexivity axiom and
	thus $(e_1, e_1) \in \sem{R}_\Mm$ giving us a contradiction.
	The second case is that $e_1 \neq e_2$. In this case, there must be
	elements $f_1, \ldots f_k$ such that $(f_i, f_{i+1}) \in \sem{R}_\execrestrict^R(\Mm, \rho)$, $e_1 = f_1$ and $e_2 = f_k$. 
	But then, since $\Mm \models \kappa(h^R_\sym(\rho))$, 
	it must be that  $(f_i, f_{i+1}) \in \sem{R}_\Mm$.
	Also, since $\Mm \models \varphi^R_\trans$, we must have $(e_1, e_2) \in \sem{R}_\Mm$ because of transitivity, giving us a contradiction.
	\item $\Mm' = \Mm''$.
\end{enumerate}

We will now argue that $\Mm' \models \Aa \uplus \set{\varphi^R_p}$
and is also feasible in $\rho$, which will imply $\rho$ is 
feasible modulo $\Aa \uplus \set{\varphi^R_p}$.
Since $\Aa$ is invariant under $(p,R)$-closure extension 
(or $(p,R)$ transitive closure extension, as appropriate),
we have that $\Mm' \models \Aa$ 
(as $\Mm'$ is either the closure extension of $\Mm$ or $\Mm''$, both of which satisfy $\Aa$).
Further, $\Mm' \models \varphi^R_p$ by definition.
Thus, $\Mm' \models \Aa \uplus \set{\varphi^R_p}$.
Now, we argue that $\Mm' \models \kappa(\rho)$.
Let $t_1 \bowtie t_2 \in \kappa(\rho)$ be an equality or disequality atom
in $\kappa(\rho)$.
First, observe that $t_1 \bowtie t_2 \in \kappa(h^R_p(\rho))$ 
as $\kappa(\rho) \subseteq \kappa(h^R_p(\rho))$, and thus
$\Mm \models t_1 \bowtie t_2$ 
(and when $\Mm''$ is defined, $\Mm'' \models t_1 \bowtie t_2$).
Further, observe that for every term $t$, 
$\sem{t}_{\Mm'} = \sem{t}_\Mm$ 
(or $\sem{t}_{\Mm'} = \sem{t}_{\Mm''}$ as appropriate)
(see \propref{terms-evaluate-same-relational-closure}).
This means that $\Mm' \models t_1 \bowtie t_2$
A similar argument ensures that
for every predicate $\psi$ of the form 
$Q(t_1, \ldots, t_k)$ or $\neg Q(t_1, \ldots, t_k)$
(where $Q$ is different from $R$) in $\kappa(\rho)$,
$\Mm' \models \psi$.
Finally, we argue about predicate atoms involving $R$.
We do a case-by-case analysis depending upon what $p$ is.
\begin{description}
\item[$p \in \set{\sym,\refl}$] 
For every positive predicate $R(t_1, t_2) \in \kappa(\rho)$, 
we have that $R(t_1, t_2) \in \kappa(h^R_p(\rho))$ 
and thus $\Mm \models R(t_1, t_2)$.
Now, notice that $\sem{R}_{\Mm} \subseteq \sem{R}_{\Mm'}$ 
(because of the way symmetric or reflexive closure is defined), 
we have that
$\Mm' \models R(t_1, t_2)$.
Let us now consider a negative predicate 
$\neg R(t_1, t_2) \in \kappa(\rho) \subseteq \kappa(h^R_p(\rho))$.
Let $e_1 = \sem{t_1}_{\Mm'}$ and $e_2 = \sem{t_2}_{\Mm'}$.
Suppose on the contrary that $(e_1, e_2) \in \sem{R}_{\Mm'}$
but $(e_1, e_2) \not\in \sem{R}_{\Mm}$.

Here we have the following subcases - 
\begin{itemize}
	\item \textbf{Case $p = \refl$}.
	By definition of $\sem{R}_{\Mm'}$, it must be that $e_1 = e_2$.
	However, by definition of $h^R_p(\rho)$, it must be that
	$R(t_1, t_1) \in \kappa(h^R_p(\rho))$ since $t_1$ is a computed term.
	Now since $\Mm \models \kappa(h^R_p(\rho))$, it must be that
	$(e_1, e_1) \in \sem{R}_{\Mm}$ giving us a contradiction.

	\item \textbf{Case $p = \sym$ and $\varphi^R_\trans \not\in \Aa$}.
	By definition of symmetric closure, we must have $(e_2, e_1) \in \sem{R}_\Mm$.
	Now, $\psi = \neg R(t_1, t_2) \in \kappa(\rho)$.
	We therefore have, by definition of $h^R_\sym$,
	$\psi' = \neg R(t_2, t_1)\in \kappa(h^P_R(\rho))$
	and thus $(e_2, e_1) \not\in \sem{R}_\Mm$, giving us a contradiction.

	\item \textbf{Case $p = \sym$, $\varphi^R_\trans \in \Aa$}.
	Argued in the paragraph titled `\textbf{Observations about $\Mm''$}' above.
	
\end{itemize}

\item[$p = \irrefl$]
For a negative predicate $\psi = \neg R(t_1, t_2) \in \kappa(\rho)$, 
it is easy to see that $\Mm' \models \psi$.
Let us consider a positive predicate $\psi = R(t_1, t_2) \in \kappa(\rho)$.
Let $e_1 = \sem{t_1}_{\Mm'}$ and $e_2 = \sem{t_2}_{\Mm'}$.
Since $\kappa(\rho) \subseteq \kappa(h^R_p(\rho))$, we have that
$(e_1, e_2) \in \sem{R}_\Mm$.
Suppose on the contrary that $\Mm' \not\models \psi$ and thus
$(e_1, e_2) \not\in \sem{R}_{\Mm'}$.
By definition of irreflexivity closure, we must have $e_1 = e_2$.
However note that by definition of $h^R_\irrefl$, $\neg R(t_1, t_1) \in \kappa(h^R_p(\rho))$
and thus $(e_1, e_2) \not\in \sem{R}_\Mm$ giving us a contradiction.
\end{description}
Thus, $\rho$ is feasible in $\Mm'$ which is a $\Aa$ model.
\end{proof}

\begin{lemma}[Preservation of Term Equalities]
\lemlabel{preservation-term-equalities-relations}
Let $\Sigma = (\Cc, \Ff, \Rr)$ be an FO signature, 
$R \in \Rr$ be a binary relation and 
$p \in \set{\refl,\irrefl,\sym}$.
Let $\Aa$ be a set of first order sentences over $\Sigma$
such that if $p = \sym$ and $\varphi^R_\trans \in \Aa$, then
the only other sentences in $\Aa$ that mention $R$ are
$\varphi^R_\refl, \varphi^R_\irrefl$ or $\varphi^R_\trans$.
Further, assume that $\Aa$ is
invariant under $(p,R)$-closure extension 
(or $(p,R)$-transitive closure extension if $\varphi^R_\trans \in \Aa$).
For any execution $\rho$, and any
two computed terms $t_1,t_2 \in \Terms(h^R_p(\rho))$,
\[
t_1 \congcl{\Aa \cup \set{\varphi^R_p} \cup \kappa(\rho)} t_2
\quad
\mbox{iff}
\quad
t_1 \congcl{\Aa \cup \kappa(h^R_p(\rho))} t_2.
\]
\end{lemma}

\begin{proof}
First observe that for every $\psi \in \kappa(h^R_p(\rho)) \setminus
\kappa(\rho)$, we have $\Aa \cup \set{\varphi^R_p} \models
\psi$. Therefore, every $\Aa \cup \set{\varphi^R_p}
\kappa(\rho)$-model is also a $\Aa \cup
\kappa(h^R_p(\rho))$-model. 
Hence, if 
$t_1 \congcl{\Aa \cup \kappa(h^R_p(\rho))} t_2$ then 
$t_1 \congcl{\Aa \cup \set{\varphi^R_p} \cup \kappa(\rho)} t_2$.

For the other direction, suppose $t_1 \ncongcl{\Aa \cup
  \kappa(h^R_p(\rho))} t_2$. Then by definition, there is a $\Aa \cup
\kappa(h^R_p(\rho))$ model $\Mm$ such that $\sem{t_1}_{\Mm} \neq
\sem{t_2}_{\Mm}$. Consider the execution $\rho_1 =
\rho\cdot\dblqt{\passume(t_1\neq t_2)}$. Technically $\rho_1$ is not
an execution by our definition. What we mean is to copy the terms
$t_1$ an $t_2$ in fresh variables when they are computed, and assume
that those variables are not equal; we skip doing this
precisely. Observe that $h^R_p(\rho_1) =
h^R_p(\rho)\cdot\dblqt{\passume(t_1\neq t_2)}$. Based on our
assumptions, $h^R_p(\rho_1)$ is feasible in $\Mm$.
By \lemref{preservation-feasibility-relations}, we have $\rho_1$ is feasible in some
$\Aa \cup \set{\varphi^R_p}$-model $\Mm'$. Thus, $\sem{t_1}_{\Mm'}
\neq \sem{t_2}_{\Mm'}$, and so $t_1 \ncongcl{\Aa \cup
  \set{\varphi^R_p} \cup \kappa(\rho)} t_2$.
  % \ucomment{How is this argument using the fact that $t_1, t_2$ are computed terms, and not arbitrary terms?}
\end{proof}

\begin{lemma}[Preservation of Coherence]
\lemlabel{preservation-coherence-relations}
Let $\rho$ be an execution, $R \in \Rr$ and 
$p \in \set{\refl,\irrefl,\sym}$.
Let $\Aa$ be a set of first order sentences over $\Sigma$
such that if $p = \sym$ and $\varphi^R_\trans \in \Aa$, then
the only other sentences in $\Aa$ that mention $R$ are
$\varphi^R_\refl, \varphi^R_\irrefl$ or $\varphi^R_\trans$.
Further, assume that $\Aa$ is
invariant under $(p,R)$-closure extension 
(or $(p,R)$-transitive closure extension if $\varphi^R_\trans \in \Aa$).
$\rho$ is coherent modulo $\Aa \uplus \set{\varphi^R_p}$
iff
$h^R_p(\rho)$ is coherent modulo $\Aa$.
\end{lemma}

\begin{proof}
Follows from \lemref{same-terms} and \lemref{preservation-coherence-relations}.
\end{proof}

\begin{lemma}[Preservation of Feasibility]
\lemlabel{preservation-feasibility-coherence-functions}
Let $\Sigma = (\Cc, \Ff, \Rr)$ be a FO signature,
$f \in \Ff$ be a unary or binary function  and let
$p \in \set{\comm,\idem}$.
Let $\Aa$ be a set of first order sentences over $\Sigma$
invariant under 1-element extension.
Also assume that $\Aa$ has no sentence that mentions $f$.
Let $\rho$ be an execution.
\begin{enumerate}
\item $\rho$ is feasible modulo $\Aa \uplus \set{\varphi^f_p}$ iff $h^f_p(\rho)$ is feasible modulo $\Aa$.
\item $\rho$ is coherent modulo $\Aa \uplus \set{\varphi^f_p}$ iff $h^f_p(\rho)$ is coherent modulo $\Aa$.
\end{enumerate}
\end{lemma}

\begin{proof}
Let us first argue preservation of feasibility.
($\Rightarrow$). 
Let $\rho$ be feasible modulo $\Aa \uplus \set{\varphi^f_p}$.
Then, there is a model $\Mm$ such that $\Mm \models \Aa \uplus \set{\varphi^f_p}$
and $\Mm \models \kappa(\rho)$.
Now, we first note that all the ground predicates 
$\phi \in \kappa(h^f_p(\rho)) \setminus \kappa(\rho)$, we have that
$\Mm \models \phi$ because $\Mm \models \set{\varphi^f_p}$.
This means, that $h^f_p(\rho)$ is feasible in $\Mm$, thus implying that
$h^f_p(\rho)$ is feasible modulo $\Aa$.

($\Leftarrow$).
Let $h^f_p(\rho)$ be feasible modulo $\Aa$.
Then, there is a model $\Mm$ such that $\Mm \models \Aa$
and $\Mm \models \kappa(h^f_p(\rho))$.
Let $\Mm' = \oneext(\Mm)$ and let $e_*$ 
be the extra element added in the construction of $\Mm'$.
By \lemref{invariance-one-element}, we have
$\Mm' \models \Aa$ and $\Mm' \models \kappa(h^f_p(\rho))$.
Now, consider the model $\Mm'' = (\Uu_{\Mm''}, \sem{}_{\Mm''})$ with
$\Uu_{\Mm''} = \Uu_{\Mm'}$ and $\sem{}_{\Mm''}$ is the
same as $\sem{}_{\Mm'}$ except for the interpretation of the
function $f$ defined as follows.
\begin{align*}
		\sem{f}_{\Mm''}(e_1, \ldots, e_k) = 
		\begin{cases}
			\sem{f}_{\Mm'}(e_1, \ldots, e_k) & 
			\begin{aligned}\text{ if there are terms } t_1, \ldots, t_k, t \in \Terms(\rho) \\
			\text{ such that } t = f(t_1, \ldots, t_k), \forall i \cdot \sem{t_i}_{\Mm} = e_i
			\end{aligned} \\
			e_* & \text{ otherwise }
		\end{cases}
\end{align*}
First observe that since $\Aa$ does not mention $f$, we have
$\Mm'' \models \Aa$.
Second, one can use induction (on the structure of terms)
to show that for every computed term 
$t \in \Terms(\rho)$, we have $\sem{t}_{\Mm''} = \sem{t}_{\Mm}$.
This means that $\Mm'' \models \kappa(h^f_p(\rho))$ and thus $\Mm'' \models \kappa(h^f_p(\rho))$.
Finally, $\Mm'' \models \varphi^f_p$ by construction.

The proof for preservation of coherence follows the same structure as that in relations
(\lemref{preservation-term-equalities-relations} and \lemref{preservation-coherence-relations}).
\end{proof}

We now move to the proof of~\thmref{combination}.
Our overall approach for proving this theorem is the following.
We partition the set of axioms as 
$\Aa = \Aa_\rel \uplus \Aa_\fun \uplus \Aa_\trans$,
where $\Aa_\rel$ is the set of relational axioms
except for the transitivity axioms 
(i.e., reflexivity, irreflexivity, symmetry),
$\Aa_\fun$ are the axioms of idempotence and commutativity of different functions,
and $\Aa_\trans$ is the set of transitivity axioms for different relations.
Now given an execution $\rho$, we first successively remove axioms
from $\Aa_\rel$. 
That is, let $\Aa_\rel = \set{\varphi^{R_i}_{p_i}}_{i=1}^k$ be some ordering
on the set of relational axioms.
We build the annotated execution $\rho'$ obtained by successively applying
the corresponding homomorphic transformations $h_i = h^{R_i}_{p_i}$.
That is $\rho_0 = \rho$, $\rho_{i+1} = h_i(\rho_i)$ and $\rho' = \rho_k$.
Our lemmas ensure that $\rho'$ is feasible and coherent modulo $\Aa_\fun \uplus \Aa_\trans$
iff the original execution $\rho$ was feasible and coherent modulo $\Aa$.
Further, $\rho'$ is effectively constructible.
Next, we eliminate the functional axioms using a similar strategy and get
an execution $\rho''$ which is feasible and coherent modulo $\Aa_\trans$
iff $\rho$ is feasible and coherent modulo $\Aa$.

\begin{reptheorem}{thm:combination}
Let $\Aa$ be a set of axioms where each relation symbol $R$ is either a total order or satisfies some (possibly empty) subset of properties out of reflexivity, irreflexivity, symmetry, transitivity, and each function symbol $f$ satisfies some (possibly empty) subset out of commutativity and idempotence. The verification problem for coherent programs modulo $\Aa$ is $\pspc$-complete.
\end{reptheorem}

\begin{proof}
$\pspc$-hardness follows from the $\pspc$-hardness of 
verification modulo $\emptyset$ as proved in~\cite{coherence2019}.
We focus on the $\pspc$ upper bound, for which we will show that the
set of executions that are feasible and coherent modulo $\Aa$ is regular
and accepted by an automaton of size $O(2^{\text{poly}(|V|)})$.
% As observed before, the set of executions $\exec(s)$ 
% of a given program $s \in \stmt$ is regular and is accepted by a NFA whose size
% is linear in $|s|$ and $|V|$.

\begin{comment}
We first start with the observation in \thmref{transitivity-reg} --- 
the set $L_\trans = \setpred{\pi \in \Pi^*}{\pi \text{ is coherent and feasible modulo } \Aa_\trans}$
is regular and accepted by an automaton of size $O(2^{\text{poly}(|V|)})$.
Now, let us order the set of functional axioms using an arbitrary ordering, i.e.,
$\Aa_\fun = \set{\varphi^{f_i}_{q_i}}_{i=1}^{k_\fun}$, 
where $f_i \in \Ff$ and $q_i \in \set{\idem, \comm}$ and $k_\fun = |\Aa_\fun|$.
Define a sequence of languages given by $L_\fun^0 = L_\trans$, $L_\fun^{i+1} = (h^{f_{i+1}}_{q_{i+1}})^{-1}(L^\fun_i)$ and finally let $L^\fun = L^\fun_{k_\fun}$.
Now one can inductively prove that $L^\fun$ is the set of executions that are feasible and coherent modulo $\Aa_\trans \cup \Aa_\fun$ based on~\lemref{preservation-feasibility-coherence-functions}.
Further, since $L^\fun$ is an inverse homomorphic image of $L_\trans$, we have that
$L^\fun$ is regular and is accepted by an automaton of size $O(2^{\text{poly}(|V|)})$.
\end{comment}
Let $L = \exec(s)$ be the set of executions 
of the given coherent program $s$; $L$ is regular.
% and accepted by an NFA of size $O(|s|)$.
Since total orders are reducible to partial orders (\lemref{total-to-partial}) (under appropriate assumptions on the trace), we will assume we only have combinations of the other axioms we consider (and not total orders).
Let $\Aa_\rel = \set{\varphi^{R_i}_{p_i}}_{i=1}^{k_\rel}$ 
be some arbitrary ordering on the set of relational axioms in $\Aa_\rel$.
We define a sequence of languages $L_0, \ldots, L_{k_\rel}$ as :
$L_0 = L$, $L_{i+1} = h^{R_{i+1}}_{p_{i+1}}(L_i)$.
Let $L_\rel = L_{k_\rel}$.
We can inductively argue that - 
\begin{enumerate}
	\item $L_{\rel}$ is regular (since regular languages are closed under homomorphism),
	\item $L_{\rel}$ is feasible modulo $\Aa \setminus \Aa_\rel = \Aa_\fun \uplus \Aa_\trans$ iff $L$ is feasible modulo $\Aa$ (using~\lemref{preservation-feasibility-relations}), and
	\item $L_{\rel}$ is coherent modulo $\Aa_\fun \uplus \Aa_\trans$ iff $L$ is coherent modulo $\Aa$ (using~\lemref{preservation-coherence-relations}). Since the given program $s$ is assumed to be coherent modulo $\Aa$, we have $L_{\rel}$ is indeed coherent modulo $\Aa_\fun \uplus \Aa_\trans$.
\end{enumerate}
Here, by feasibility (resp. coherence) of a language, we mean feasibility (resp. coherence) of each of the strings in the language.

We now analogously get rid of axioms in $\Aa_\fun$ one at a time.
Let $\Aa_\fun = \set{\varphi^{f_i}_{q_i}}_{i=1}^{k_\fun}$ 
be some arbitrary ordering on the set of functional axioms in $\Aa_\fun$.
We define a sequence of languages $K_0, \ldots, K_{k_\fun}$ as :
$K_0 = L_\rel$, $K_{i+1} = h^{f_{i+1}}_{q_{i+1}}(K_i)$.
Let $L_\fun = K_{k_\fun}$.
We can inductively argue that - 
\begin{enumerate}
	\item $L_{\fun}$ is regular (since regular languages are closed under homomorphism).
	\item $L_{\fun}$ is feasible modulo $\Aa_\trans$ iff $L_\rel$ is feasible modulo $\Aa_\fun \uplus \Aa_\trans$ (using~\lemref{preservation-feasibility-coherence-functions}).
	This implies that $L_{\fun}$ is feasible modulo $\Aa_\trans$ iff $L$ is feasible modulo $\Aa$.
	\item $L_{\fun}$ is coherent modulo $\Aa_\trans$ iff $L_\rel$ is coherent modulo $\Aa_\fun \uplus \Aa_\trans$ (using~\lemref{preservation-feasibility-coherence-functions}). Together with the previous observations, we have $L_{\fun}$ is indeed coherent modulo $\Aa_\trans$.
\end{enumerate}
\end{proof}

% \thmlabel{transitivity-reg}
Thus, the verification problem reduces to checking if 
$L_\fun$ is feasible modulo $\Aa_\trans$. 
This problem is decidable as a consequence of~\thmref{transitivity-reg} 
and the fact that $L_\fun$ is coherent modulo $\Aa_\trans$.
In other words, we need to check for containment of two regular languages - 
$L_\fun \subseteq L(\Ff_\trans)$, where $\Ff_\trans$ is the automaton in~\thmref{transitivity-reg}, which is decidable.

The complexity arguments follows from the observation that 
$L$ (and thus $L_\fun$) are recognizable by NFAs of size linear in $|s|$,
$|V|$ and $|\Aa_\fun \uplus \Aa_\rel|$ and are also effectively constructible
and that the containment check can be done in $\pspc$.

\end{document}